\documentclass[10pt]{article}

\usepackage{geometry}  
\usepackage[english]{babel}
\usepackage[utf8]{inputenc}
\usepackage[T1]{fontenc}

\usepackage{hyperref}
\usepackage{xcolor}
\usepackage{subfigure}
\usepackage{booktabs}
\usepackage[table]{xcolor}
\usepackage{amsmath}
\usepackage{amssymb}
\usepackage{amsthm}
\usepackage{bbm}
\usepackage{graphicx}
\usepackage{algorithm}
\usepackage{algpseudocode}
\usepackage{tikz}
\usetikzlibrary{arrows.meta, positioning, fit, backgrounds, tikzmark,decorations.pathreplacing}

\RequirePackage[authoryear]{natbib}

\newtheorem{theorem}{Theorem}[section]
\newtheorem{proposition}[theorem]{Proposition}

\theoremstyle{definition}
\newtheorem{definition}[theorem]{Definition}

\newtheorem*{proof sketch}{Proof sketch}
\newtheorem*{remark}{Remark}

\title{Conformal prediction after data-dependent model selection}
\author{Ruiting Liang\thanks{Committee on Computational and Applied Mathematics, University of Chicago} \and Wanrong Zhu\thanks{Department of Statistics, University of California, Irvine} \and Rina Foygel Barber\thanks{Department of Statistics, University of Chicago}}
\date{\today}

\begin{document}

\maketitle

\begin{abstract}
Given a family of pretrained models and a hold-out set, how can we construct a valid conformal prediction set while selecting a model that minimizes the width of the set? If we use the same hold-out data set both to select a model (the model that yields the smallest conformal prediction sets) and then to construct a conformal prediction set based on that selected model, we suffer a loss of coverage due to selection bias. Alternatively, we could further split the data to perform selection and calibration separately, but this comes at a steep cost if the size of the dataset is limited. In this paper, we address the challenge of constructing a valid prediction set after data-dependent model selection---commonly, selecting the model that minimizes the width of the resulting prediction sets. Our novel methods can be implemented efficiently and admit finite-sample validity guarantees without invoking additional sample-splitting. We show that our methods yield prediction sets with asymptotically optimal width under certain notions of regularity for the model class. The improvement in the width of the prediction sets constructed by our methods is further demonstrated through applications to synthetic datasets in various settings and a real data example.

\end{abstract}

\section{Introduction}\label{sec:intro}
Amidst the rapid expansion of statistical applications across diverse fields, rigorous uncertainty quantification is essential to ensure the reliability and accuracy of statistical outcomes. Among many approaches that have been proposed, conformal prediction, pioneered by \citet{vovk2005algorithmic}, has recently gained massive popularity due to its assumption-lean nature in providing finite-sample valid prediction sets on top of any predictive models. Specifically, the procedure can be paired with any regression algorithm, and provides finite-sample validity guarantee relying only on the assumption of data exchangeability.

However, despite the success of the conformal prediction framework in offering statistical validity, the optimality in the size of conformal prediction sets generally lacks a universal guarantee and often relies on a model-specific analysis with additional assumptions on the data distribution.
In particular, a substantial line of work studies width optimality in settings where, assuming prior information about the data distribution, it is possible to construct a prediction band that offers asymptotically optimal width subject to a conditional, rather than marginal, coverage requirement.
For example, \citet{sesia2020comparison} show that the conformalized quantile regression \citep{romano2019conformalized} achieves the optimal width asymptotically if the quantile regression function estimate is consistent and the data distribution has certain regularity properties. Similar results also hold for the split conformal method and the full conformal method when paired with the residual score \citep{lei2018distribution}. See also the work of \citet{gyorfi2019nearest,burnaev2014efficiency,sadinle2019least,izbicki2020flexible}, among others, for additional examples of settings where width optimality of conformal prediction has been established under specific conditions; in general, obtaining a small conformal prediction set is closely tied to choosing a conformity score that is well suited to the data distribution, as has been studied empirically in \citet{aleksandrova2021nonconformity,johansson2017model,sesia2020comparison}.

Given the necessity of choosing the right model to construct an effective prediction set, and the practical reality that the most suitable model cannot be predetermined without analyzing the data, an integrated framework that combines model selection with the construction of prediction sets would be highly beneficial. This serves as the motivation of this paper. In this work, we address the challenge of constructing a finite-sample valid conformal prediction set while selecting a model that minimizes the width of the prediction set, from a given family of pretrained models. Specifically, fix some target coverage level $1-\alpha \in (0,1)$, e.g., 90\%. Given a collection of pretrained models,  a set of independent and identically distributed hold-out data $ \mathcal{D}_{n}$, and a new test point $X_{n+1}$ from the same distribution, we would like to construct a conformal prediction set $\widehat{C}_{n}$ using a selected model (i.e., selected as the model that offers the smallest prediction set), such that 
\begin{itemize}
	\item the prediction set is valid, i.e.,  $\mathbb{P} \left\{Y_{n+1} \in \widehat{C}_{n}(X_{n+1}) \right\} \ge 1-\alpha$,
	\item  the width of the set $\widehat{C}_{n}$ is as small as possible due to selecting the best model.
\end{itemize}
Several existing works have considered problems similar to the question described above. In Section~\ref{sec:exist_mtd} below, we will closely examine the work of \citet{yang2024selection}, which is most relevant to our work. In their work, \citet{yang2024selection} propose either choosing the model with the shortest interval without correcting for selection to achieve small width (which suffers from selection bias unless the set of candidate models is extremely small), or using a sample splitting based approach (which removes the issue of selection bias, but suffers in statistical efficiency due to reduced sample size). In contrast, our work aims to correct for selection bias without any data splitting. For a more comprehensive discussion of related literature, we defer it to Section~\ref{sec:related_literature}.


\paragraph{Main contributions.}
We propose a novel method with two variants to construct finite-sample valid conformal prediction sets without additional sample-splitting after data-dependent model selection. Both versions are computationally efficient to implement, and offer validity with no assumptions aside from exchangeability of the data. In addition, we show that under certain regularity conditions for the model class, our methods output prediction sets that are asymptotically optimal in width. We will see on simulated and real data that these methods yield tighter prediction sets in finite-sample relative to existing approaches, by removing selection bias without paying the price of data splitting.

\paragraph{Organization.}
We first formulate our problem as well as surveying existing methods in Section~\ref{sec:backg&problem}. We then present our proposed methods as well as coverage guarantees in Section~\ref{sec:method}.  The asymptotic optimality in width of our methods is investigated in Section \ref{sec: efficiency_res}. In Section~\ref{sec:experiment}, we conduct empirical comparisons across different methods to demonstrate the effectiveness of our methods in practice. Finally, we review the related literature and conclude with a brief discussion in Section~\ref{sec:discuss}.

\section{Problem formulation}\label{sec:backg&problem}
In this section, we first review some background on conformal prediction, which serves as the backbone of our prediction sets construction. We will then specify the class of model selection procedures considered in this paper and formulate the problem of interest. Finally, we will briefly review the existing methods related to our problem.

\subsection{Background: conformal prediction}\label{sec:CP}
Consider the learning problem of predicting the response $y\in\mathcal{Y}$ based on an observed feature $x\in\mathcal{X}$. Conformal prediction offers a framework of constructing valid prediction sets without placing any assumption on the data distribution $P$ other than data exchangeability.

\paragraph{Split conformal prediction.} Given a pretrained score function\footnote{Here and throughout the paper, we will employ the terms ``model'' and ``score'' interchangeably, with a slight abuse of the terminology. Although they have distinct technical definitions, this usage is intended to streamline our discussion.}  $S: \mathcal{X} \times \mathcal{Y} \to \mathbb{R}$ and  a hold-out set $\mathcal{D}_{n} = \{(X_i, Y_i)\}_{i\in[n]}$, 
the split conformal method \citep{vovk2005algorithmic, papadopoulos2002inductive} provides a computationally simple way to construct finite-sample valid prediction sets, given by \footnote{To avoid the trivial case, we will assume implicitly that $(1-\alpha)(1+1/n)\leq 1$ throughout.}
\begin{equation}\label{eqn:define_split_CP}\widehat{C}_{\texttt{split-CP}}(X_{n+1}) =  \left\{y\in\mathcal{Y}: S(X_{n+1}, y)\le \textnormal{Quantile}_{(1-\alpha)(1+1/n)} \left( \{S(X_i, Y_i)\}_{i \in [n]}\right)\right\}. \end{equation}
Since the pretrained score $S$ is independent of the hold-out set $\mathcal{D}_{n}$, the split conformal prediction method offers marginal coverage \citep{vovk2005algorithmic, papadopoulos2002inductive},
\[\mathbb{P} \left\{Y_{n+1} \in \widehat{C}_{\texttt{split-CP}}(X_{n+1}) \right\} \ge 1-\alpha,\]
under the assumption that $(X_1, Y_1), \dots, (X_n, Y_n), (X_{n+1}, Y_{n+1})$ are exchangeable.

We note that in practice, implementing split conformal prediction would entail splitting the available training data, with one portion of the data used for producing the pretrained score function $S$ (such as the commonly used residual score, $S(x,y) = |y-f(x)|$, where $f$ is a regression model fitted to the data), and the remaining data used for the calibration step---i.e., for computing the quantile appearing in~\eqref{eqn:define_split_CP}. To streamline notation throughout the paper, we will write $(X_1,Y_1),\dots,(X_n,Y_n)$ to denote the calibration set, i.e., the second portion of the available training data, implicitly assuming that this calibration set is independent of the dataset used for pretraining $S$.

\paragraph{Full conformal prediction.}
If the amount of available data is limited, we may wish to avoid the cost of data splitting inherent in split conformal prediction. Full conformal prediction \citep{vovk2005algorithmic} provides a solution that allows us to train a model and
calibrate it to construct a prediction set without data splitting. 
Specifically, given a dataset $\mathcal{D}_{n}\in(\mathcal{X}\times\mathcal{Y})^n$ and a symmetric algorithm $\mathcal{A}: \cup_{n \ge 0} (\mathcal{X} \times \mathcal{Y})^{n} \to \mathcal{F}$, where $\mathcal{F}$ is the set of measurable functions from $\mathcal{X}\times\mathcal{Y}$ to $\mathbb{R}$, the full conformal prediction set is defined as
\begin{equation}\label{eqn:define_fullCP}
     \widehat{C}_{\texttt{full-CP}}(X_{n+1}) = \left\{y \in \mathcal{Y}: S^{y}(X_{n+1},y) \le \mathrm{Quantile}_{1-\alpha} \left( \left\{S^{y}(X_{i},Y_{i}) \right\}_{i\in [n]} \cup \{S^{y}(X_{n+1},y)\} \right) \right\},
\end{equation}
where $S^{y}=\mathcal{A}(\mathcal{D}_{n} \cup \{(X_{n+1},y)\})$ for each $y$. (Note that the split conformal can be derived as a special case of this procedure by choosing $\mathcal{A}$ to be a trivial algorithm that always returns a fixed model $S$.) As for the split conformal method, full conformal prediction also satisfies a marginal coverage guarantee \citep{vovk2005algorithmic},
\[\mathbb{P} \left\{Y_{n+1} \in \widehat{C}_{\texttt{full-CP}}(X_{n+1}) \right\} \ge 1-\alpha,\]as long as $(X_1, Y_1), \dots, (X_n, Y_n), (X_{n+1}, Y_{n+1})$ are exchangeable.

In this paper, we will work with pretrained models---that is, on the surface, the problem we consider can be viewed as a problem of split conformal prediction, only with multiple choices of the pretrained model $S$. However, the full conformal prediction framework will be essential for our methodology, because the process of choosing one of the available models $S$ will be data-dependent, and we will see that a full conformal type strategy will be needed to avoid selection bias.

\subsection{The model selection criterion}\label{sec:def_L}
We are now ready to consider the setting of a finite collection of pretrained models, $\{S^{\lambda} : \lambda\in\Lambda\}$---for instance, these models might be produced by: running different regression algorithms on the same data; running the same regression algorithm on different subsets of features in the data; or running the same regression algorithm on the same data, but with different values of a tuning parameter (e.g., the penalty parameter for the Lasso).

For any pretrained model $S^{\lambda}$,  we consider the prediction set (at a feature value $x\in\mathcal{X}$) of the form: 
\[C_{q}^{\lambda}(x) = \{y\in\mathcal{Y}: S^{\lambda}(x, y)\le q\},\]
where the threshold $q$ controls the size of the set as well as the coverage. For example, in split conformal prediction as defined in~\eqref{eqn:define_split_CP},  $q$ is given by $ \textnormal{Quantile}_{(1-\alpha)(1+1/n)} \left( \{S(X_i, Y_i)\}_{i \in [n]}\right)$. Throughout, for any model $S^{\lambda}$ and any cutoff $q$, we will write $C_{q}^{\lambda}(x) $ to denote a prediction set at a particular value $x$, while $C_{q}^{\lambda}$ denotes the prediction band in general.

We next specify the criterion of model selection---what does it mean to choose the best model from the set? Formally, we will quantify model effectiveness by a loss function $\mathcal{L}$.
\begin{definition}[Loss function for prediction band width]\label{def:sym_loss}
Define $\mathcal{L}(\lambda, q; x_{1},\dots,x_m)$ as a measurement of the size of the prediction band $C_{q}^{\lambda}$ evaluated at  points $x_{1},\dots, x_m \in \mathcal{X}$ (for $m\geq 1$), where $\mathcal{L}$ is assumed to satisfy the following properties:
\begin{enumerate}
    \item Monotonicity: For any  $S^{\lambda}$ and evaluation points $x_1,\dots,x_m\in\mathcal{X}$,
    $\mathcal{L}$ is non-decreasing with respect to $q$, i.e., $\mathcal{L}(\lambda, q_{1}; x_{1},\dots,x_{m}) \le \mathcal{L}(\lambda, q_{2}; x_{1},\dots,x_{m}), \ \forall q_{1} \le q_{2}\in\mathbb{R}.$
    \item Symmetry: 
    For any $S^{\lambda}$  and cutoff $q\in\mathbb{R}$, and any $x_1,\dots,x_m\in\mathcal{X}$, it holds for any permutation $\pi$ on $\{1,\dots,m\}$ that $\mathcal{L}(\lambda, q; x_{\pi(1)},\dots,x_{\pi(m)}) = \mathcal{L}(\lambda, q; x_{1},\dots,x_m).$
\end{enumerate}
\end{definition}
\noindent A typical choice for this loss might be the average width of the prediction set,
\begin{equation}\label{eqn:loss_mean_size}\mathcal{L}(\lambda,q;x_1,\dots,x_m) = \frac{1}{m}\sum_{i=1}^m |C_q^\lambda(x_i)|\end{equation}
where $|\cdot|$ denotes the size of the set, e.g., the Lebesgue measure if $\mathcal{Y}=\mathbb{R}$, or the cardinality if $\mathcal{Y}$ is a finite set of categorical labels. However, this is not the only choice we might use in practice---for instance, we might choose median, rather than mean, size.\footnote{As an example of another kind of loss function that we might use in a regression setting, the work of \citet{liang2023conformal} considers a selection procedure that aims at  minimizing the squared estimation error, $\frac{1}{n}\sum_{i=1}^n(Y_i - f_\lambda(X_i))^2$, rather than the interval width. Like our work which will be described in detail below, \citet{liang2023conformal}'s approach uses the full conformal framework to then incorporate a hypothesized test point $(X_{n+1},y)$ into the selection process. Note however that this requires a loss $\mathcal{L}$ that depends on $(X,Y)$ values at evaluation (rather than only $X$ values), and is therefore not directly comparable to our framework. See Section~\ref{sec:related_literature} for further discussion.}

To see how this loss function behaves in practice, we now give several examples for specific commonly used conformity scores in the literature:
\begin{itemize}
    \item For the residual score $S^\lambda(x,y) = |y - f_\lambda(x)|$ (where $f_\lambda$ is a pretrained regression function), the loss given in~\eqref{eqn:loss_mean_size} is equal to
    $$\mathcal{L}(\lambda,q;x_1,\dots,x_m) = 2q,$$
    the width of the interval (which is shared across all values $x$, and therefore does not depend on the particular evaluation points $x_1,\dots,x_m$).
    \item For the rescaled residual score $S^\lambda(x,y) = \frac{|y - f_\lambda(x)|}{\sigma_\lambda(x)}$ (where $f_\lambda$ is a pretrained regression function, and $\sigma_\lambda$ is a pretrained local estimate of the scale of the noise), the loss given in~\eqref{eqn:loss_mean_size} is equal to
    $$\mathcal{L}(\lambda,q;x_1,\dots,x_m) = 2q \cdot \frac{1}{m}\sum_{i=1}^m \sigma_\lambda(x_i).$$
    \item For the conformalized quantile regression (CQR) score $S^\lambda(x,y) = \max\{\hat{Q}_{\alpha/2}(x) - y, y - \hat{Q}_{1-\alpha/2}(x)\}$ \citep{romano2019conformalized}, where $\hat{Q}_{\alpha/2}(x)$ and $\hat{Q}_{1-\alpha/2}(x)$ are pretrained estimates of the $\alpha/2$- and $(1-\alpha/2)$-quantiles of $Y|X=x$, the loss~\eqref{eqn:loss_mean_size} is equal to $$\mathcal{L}(\lambda,q;x_1,\dots,x_m) = 2q + \frac{1}{m}\sum_{i=1}^m \left(\hat{Q}_{1-\alpha/2}(x_i) - \hat{Q}_{\alpha/2}(x_i)\right)$$
    (under mild assumptions).
    \item In a setting where the response $Y$ is categorical (and so $\mathcal{Y}$ is a finite set), the conditional density score $S^\lambda(x,y) = - p_\lambda(y\mid x)$ (where $p_\lambda(y\mid x)$ is a pretrained estimate of the conditional probability function for $Y\mid X$), the loss given in~\eqref{eqn:loss_mean_size} is equal to
    \begin{equation}\label{eqn:loss_classification}\mathcal{L}(\lambda,q;x_1,\dots,x_m) = \frac{1}{m}\sum_{i=1}^m |\{y\in\mathcal{Y} : p_\lambda(y\mid x_i) \geq -q\}|.\end{equation}
\end{itemize}
With the loss function $\mathcal{L}$ in place, we are now ready to define the problem of interest:
\begin{quote}
    \textbf{Main aim:} Given a collection of pretrained models $\{S^\lambda : \lambda\in\Lambda\}$, and a calibration data set $(X_1,Y_1),\dots,(X_n,Y_n)$, we aim to output a prediction set at a new test point $X_{n+1}$ that is both valid (under exchangeability of the data) and small (in the sense of aiming to minimize $\mathcal{L}$ over the collection of models).
\end{quote}

\subsection{Existing methods}\label{sec:exist_mtd}
As mentioned earlier in Section~\ref{sec:intro}, the work of \citet{yang2024selection} proposes two methods to solve the problem stated above. We briefly review their methods in this section.

\paragraph{YK-baseline.}
In their first approach, split conformal prediction sets are constructed for each pretrained model $S^{\lambda}$, and the one with the smallest width is selected directly without any correction for selection bias.  We refer to this approach as \texttt{YK-baseline} in this paper, since it provides a baseline against which other methods (that do account for selection bias) can be compared. We can define it precisely as follows. First, for each $\lambda\in\Lambda$, define
\begin{equation}\label{eqn:define_hat_q_lambda}\hat{q}(\lambda) = \textnormal{Quantile}_{(1-\alpha)(1+1/n)}\left(S^{\lambda}(X_1,Y_1),\dots,S^{\lambda}(X_n,Y_n)\right),\end{equation}
so that $C_{\hat{q}(\lambda)}^\lambda(X_{n+1})$ yields the split conformal prediction set using the pretrained score $S^\lambda$ (i.e., without considering any other pretrained models).
Select the ``best'' model as\footnote{\citet{yang2024selection}'s work does not formally specify a mechanism for defining a ``best'' prediction set, so here we choose a specific definition for concreteness, defining ``best'' to mean achieving the minimal loss (e.g., minimal width) over the $n+1$ calibration and test feature values.}
\begin{equation}\label{eqn:define_hat_lambda}\hat\lambda = \arg\min_{\lambda\in\Lambda} \mathcal{L}(\lambda,\hat{q}(\lambda);X_1,\dots,X_{n+1}),\end{equation}
the final prediction set is given by
\begin{equation}\label{eqn:YK-baseline_define}\widehat{C}_{\texttt{YK-baseline}}(X_{n+1}) = \left\{y \in \mathcal{Y}: S^{\hat\lambda}(X_{n+1},y)\leq \hat{q}(\hat\lambda)\right\}.\end{equation}

\begin{remark}[Breaking ties]
Above in~\eqref{eqn:define_hat_lambda}, and throughout the paper, when we define a model $\lambda\in\Lambda$ as a minimizer of the loss, if the minimizer is not unique then an arbitrary tie-breaking rule may be used.
We will consider this more formally in Appendix~\ref{app:tie-breaking}.\end{remark}

By construction, \texttt{YK-baseline} yields a prediction set with the smallest width among all the possible split conformal prediction sets $C_{\hat{q}(\lambda)}^{\lambda}$ (indexed over $\lambda\in\Lambda$). 
In terms of coverage, \cite{yang2024selection} show the following concentration-based guarantee:
\begin{equation}\label{eqn: YK-nonadj_cov}
    \mathbb{P} \left\{Y_{n+1} \in \widehat{C}_{\texttt{YK-baseline}}(X_{n+1}) \right\} \ge 1-\alpha - \frac{\sqrt{\frac{1}{2}\log(2|\Lambda|)} + \frac{1}{3} - \frac{1-\alpha}{\sqrt{n}}}{\sqrt{n}}
\end{equation}
Here they assume  $(X_1,Y_1),\dots,(X_n,Y_n),(X_{n+1},Y_{n+1})$ are i.i.d. rather than exchangeable.
We can interpret this lower bound as saying that, if $|\Lambda|$ is not too large (i.e., $\log(2|\Lambda|)\ll n$), then the selection bias is limited and will not substantially reduce coverage. On the other hand, for large $|\Lambda|$, this lower bound does not give a meaningful guarantee since coverage may be arbitrarily low, which will be confirmed in our experiments in Section~\ref{sec:experiment}. 

\paragraph{YK-adjust.}
By examining the lower bound in \eqref{eqn: YK-nonadj_cov},  one can adjust the value of $\alpha$ to recover the desired coverage level.
Let
$\tilde\alpha = \alpha  - \frac{\sqrt{\frac{1}{2}\log(2|\Lambda|)} + \frac{1}{3} - \frac{1-\alpha}{\sqrt{n}}}{\sqrt{n} (1+1/n)}$. One can construct $\widehat{C}_{\texttt{YK-adjust}}(X_{n+1})$ as in~\eqref{eqn:YK-baseline_define} with $\tilde\alpha$ in place of $\alpha$. Then applying the bound~\eqref{eqn: YK-nonadj_cov}, we have 
\[\mathbb{P} \left\{Y_{n+1} \in \widehat{C}_{\texttt{YK-adjust}}(X_{n+1}) \right\} \ge 1-\alpha.\]
However, as we will see in Section~\ref{sec:experiment}, this adjusted version is highly conservative and often returns prediction sets that are extremely wide.

\paragraph{YK-split.}
To avoid the problem of selection bias incurred by \texttt{YK-baseline}, \citet{yang2024selection} propose an alternative method, based on data splitting, which we will refer to as \texttt{YK-split}. At a high level, the calibration set of $n$ data points is split into two parts, where $n_1$ data points are used to select a model, and the remaining $n_2 = n-n_1$ data points are used to construct the split conformal prediction set with the selected model. Specifically, define
$\hat{q}_1(\lambda) = \textnormal{Quantile}_{(1-\alpha)(1+1/n_1)}(\{S^\lambda(X_i,Y_i)\}_{i\in[n_1]})$, 
and $\hat{q}_2(\lambda) = \textnormal{Quantile}_{(1-\alpha)(1+1/n_2)}(\{S^\lambda(X_i,Y_i)\}_{n_1<i\leq n})$. The prediction set is given by
\[\widehat{C}_{\texttt{YK-split}}(X_{n+1}) = \left\{y\in\mathcal{Y}: S^{\hat\lambda_{1}}(X_{n+1},y)\leq \hat{q}_2(\hat\lambda_1)\right\}, \]
where $\hat\lambda_{1}=\arg\min_{\lambda\in\Lambda}\mathcal{L}(\lambda,\hat{q}_1(\lambda);X_1,\dots,X_{n_1})$.
Since we have split the data into two disjoint portions, this construction avoids selection bias, and thus marginal coverage holds,
\begin{equation*}
    \mathbb{P} \left\{Y_{n+1} \in \widehat{C}_{\texttt{YK-split}}(X_{n+1}) \right\} \ge 1-\alpha,
\end{equation*}
assuming that $(X_1,Y_1),\dots,(X_n,Y_n),(X_{n+1},Y_{n+1})$ are exchangeable.

While data splitting removes the issue of selection bias and restores  marginal coverage, we would naturally expect this to come at a price due to the reduction in sample size (both for the selection of the model $\hat\lambda$, and for computing a threshold for the scores). However, \citet{yang2024selection} show that under certain smoothness conditions for the quantile function and model class, \texttt{YK-split} is guaranteed to return approximately the smallest prediction set for sufficiently large $n$. When $n$ is moderate, though, we will see in our experiments in Section~\ref{sec:experiment} that sample splitting can lead to a wider prediction set.

\section{Conformal prediction with model selection}\label{sec:method}
In this section, we present our proposed methodology: conformal prediction with model selection over a finite collection of models $\Lambda$. We develop two variants of our procedure, called \texttt{ModSel-CP} and \texttt{ModSel-CP-LOO}, which we will describe in detail below. At a high level, our aim is to correct for the selection bias incurred because $\hat\lambda\in\Lambda$ is chosen in a data-dependent way---and to do so, we will use the idea of full conformal prediction, which itself is a mechanism for providing a valid prediction when using any model fitting algorithm (in this case, any model \emph{selection} algorithm).

\subsection{The \texttt{ModSel-CP} method}\label{sec:ModSel}
Recall that the key ingredient that ensures the finite-sample validity of conformal prediction is the exchangeability of the data. The reason why predictive coverage is lost when we use a data-dependent selected model $\hat\lambda\in\Lambda$ (as in \texttt{YK-baseline}) is that $\hat\lambda$ depends on the data nonsymmetrically: $\hat\lambda$ chooses a model based on the quantiles of scores on the calibration points, $(X_1,Y_1),\dots,(X_n,Y_n)$, but does not make use of the test point $(X_{n+1},Y_{n+1})$. Treating the test point differently from the calibration points means that, once $\hat\lambda$ is identified as the selected model, the data points are no longer (conditionally) exchangeable.

This is the starting point of our method: we will aim to remove this issue by restoring symmetry, via incorporating the test point (in addition to the calibration points) into the model selection process. 

Specifically, for each $y \in \mathcal{Y}$, define the test-point-augmented selection as follows:
\begin{equation*}
    \hat{\lambda}(y) = \mathrm{arg} \min \limits_{\lambda \in \Lambda}\mathcal{L}(\lambda,\hat{q}(\lambda,y); X_1,\dots,X_{n+1}),
\end{equation*}
where $\hat{q}(\lambda,y) = \mathrm{Quantile}_{1-\alpha} \left( S^{\lambda}(X_{1},Y_{1}), \dots, S^{\lambda}(X_{n},Y_{n}), S^{\lambda}(X_{n+1},y) \right).$
In other words, $\hat\lambda(y)$ identifies the model that yields the best (e.g., smallest) prediction sets, if calibrated and evaluated on all $n+1$ data points $(X_1,Y_1),\dots,(X_n,Y_n),(X_{n+1},y)$. In contrast, the selected model $\hat\lambda$ in \texttt{YK-baseline} uses only  $(X_1,Y_1),\dots,(X_n,Y_n)$ and does not include a test point. For our proposed method, including the (hypothesized) test point $(X_{n+1},y)$ into the model selection procedure can be viewed as restoring symmetry to our treatment of the $n+1$ data points---the $n$ calibration points $\{(X_i,Y_i)\}_{i\in [n]}$ and the test point $(X_{n+1},Y_{n+1})$.

With this definition in place, we then define the prediction set for \texttt{ModSel-CP} as:
\begin{equation}\label{defn:ModSel}
\widehat{C}_{\texttt{ModSel-CP}} (X_{n+1}) = \left\{y \in \mathcal{Y}: S^{\hat{\lambda}(y)}(X_{n+1},y) \le \hat{q} (\hat{\lambda}(y),y ) \right\}.\end{equation}
Our first main result establishes that this method offers marginal coverage.
\begin{theorem}[Validity of \texttt{ModSel-CP}]\label{thm:ModSel-validity}
    Fix any $\alpha\in(0,1)$. Assume $(X_1,Y_1),\dots,(X_{n+1},Y_{n+1})$ are exchangeable, and $\{S^\lambda:\lambda\in\Lambda\}$ is any collection of pretrained models. Then \textnormal{\texttt{ModSel-CP}} satisfies a marginal coverage guarantee,
    \begin{equation*}
    	\mathbb{P} \left\{Y_{n+1} \in \widehat{C}_{\textnormal{\texttt{ModSel-CP}}}(X_{n+1}) \right\} \ge 1-\alpha.
    \end{equation*}
\end{theorem}
\noindent As we will see in the proof, this marginal coverage result will follow by reformulating \texttt{ModSel-CP} as an instance of full conformal prediction. 
In particular, treating the model selection step (i.e., the process of selecting a model $\hat\lambda(y)$ as a function of the data) as a part of the training algorithm $\mathcal{A}$ in full conformal, the \texttt{ModSel-CP} method can be seen as conformalizing over a new meta-score output by this algorithm $\mathcal{A}$.

\subsubsection{Implementation}  \label{sec:ModSel_implement}
As mentioned above, the construction of our method is based on the full conformal prediction framework, by integrating the test point into the selection procedure in order to restore symmetry and ensure a distribution-free coverage guarantee. In general, full conformal prediction may be extremely computationally expensive to implement, since we need to retrain our model for each possible value $y\in\mathcal{Y}$, as in~\eqref{eqn:define_fullCP}. For our specific constructions of the \texttt{ModSel-CP} prediction sets, initially it may appear to be the case that these constructions are computationally infeasible as well, since \texttt{ModSel-CP} requires identifying the selected model $\hat\lambda(y)$ for each $y\in\mathcal{Y}$, as in~\eqref{defn:ModSel}. However, we will now see that this method can be implemented extremely efficiently.

To state Theorem \ref{thm:ModSel_simplify}, we first need to define some notation. For convenience, we will use the abbreviated notation $\mathcal{L}_{n+1}(\lambda,q)=\mathcal{L}(\lambda,q;X_1,\dots,X_{n+1})$, here and later on in the paper. Define the set of competing models,
\begin{equation}\label{eqn:good_models}
		\mathcal{M} = \left\{\lambda \in \Lambda: \mathcal{L}_{n+1}(\lambda,\hat{q}_{-}(\lambda)) \le \mathcal{L}_{n+1}({\hat\lambda},\hat{q}(\hat\lambda)) \right\},
		\end{equation}
	where $\hat{q}(\lambda)$ and $\hat\lambda$ are defined as in \eqref{eqn:define_hat_q_lambda} and~\eqref{eqn:define_hat_lambda}, respectively, while \[\hat{q}_{-}(\lambda) =  \mathrm{Quantile}_{(1-\alpha)(1+1/n)-1/n} \left(S^\lambda(X_1,Y_1),\dots,S^\lambda(X_n,Y_n)\right).\] 
Similarly, we define
$\mathcal{M}_{-} = \left\{\lambda \in \Lambda: \mathcal{L}_{n+1}(\lambda, \hat{q}_{-}(\lambda)) < \mathcal{L}_{n+1}(\hat{\lambda},\hat{q}(\hat{\lambda})) \right\} \subseteq \mathcal{M}.$
 
We can observe that, by construction, $\hat{q}(\lambda,y)\geq \hat{q}_-(\lambda) $ must hold for all $\lambda$ and all $y$.
Consequently, $\lambda\in\mathcal{M}$ indicates that it is possible, if the test point score $S^\lambda(X_{n+1},y)$ is sufficiently small, to have $\mathcal{L}_{n+1}(\lambda,\hat{q}(\lambda,y)) \leq \mathcal{L}_{n+1}(\hat\lambda,\hat{q}(\hat\lambda))$, which makes the value $\lambda$ a potential candidate for $\hat{\lambda}(y)$---this explains how the set $\mathcal{M}$ is related to the calculation of the \texttt{ModSel-CP} prediction set (and similarly for $\mathcal{M}_-$).
We now formalize this connection.

\begin{theorem}\label{thm:ModSel_simplify}
 For any given dataset $\mathcal{D}_n$ and test point $X_{n+1}$, the \textnormal{\texttt{ModSel-CP}} prediction set defined in~\eqref{defn:ModSel} satisfies 
 \begin{equation}\label{eqn:ModSel_simple}
     \widehat{C}_{\textnormal{\texttt{ModSel-CP}}} (X_{n+1}) \subseteq \left\{y \in \mathcal{Y}: \min_{\lambda\in\mathcal{M}}\mathcal{L}_{n+1}(\lambda,S^\lambda(X_{n+1},y))\le\mathcal{L}_{n+1}({\hat\lambda},\hat{q}(\hat\lambda)) \right\},
 \end{equation}
 \begin{equation}\label{eqn:ModSel_simple_lower}
     \widehat{C}_{\textnormal{\texttt{ModSel-CP}}} (X_{n+1}) \supseteq \left\{y \in \mathcal{Y}: \min_{\lambda\in\mathcal{M}_-}\mathcal{L}_{n+1}(\lambda,S^\lambda(X_{n+1},y))<\mathcal{L}_{n+1}({\hat\lambda},\hat{q}(\hat\lambda)) \right\}.
 \end{equation}

 Furthermore,
 if it holds almost surely that:
 \begin{itemize}
     \item The elements in the set $\{\mathcal{L}_{n+1}(\lambda,\hat{q}_-(\lambda))\}_{\lambda\in\mathcal{M}} \cup \{\mathcal{L}_{n+1}(\lambda,\hat{q}(\lambda))\}_{\lambda\in\mathcal{M}}$ are all distinct,
     \item For each $q\in\mathbb{R}$, the set $\{y\in\mathcal{Y}: \mathcal{L}_{n+1}(\lambda, S^{\lambda}(X_{n+1},y)) = q\}$ has measure zero (with respect to some base measure on $\mathcal{Y}$),
 \end{itemize} 
 then both set inclusions~\eqref{eqn:ModSel_simple} and~\eqref{eqn:ModSel_simple_lower} are in fact equalities up to sets of measure zero (with respect to the same base measure on $\mathcal{Y}$).
\end{theorem}
\noindent 
Note that Theorem~\ref{thm:ModSel_simplify} does not provide an exact calculation for the set $\widehat{C}_{\textnormal{\texttt{ModSel-CP}}} (X_{n+1})$, but rather provides upper and lower bounds on this set. Nevertheless, we argue that in most practical scenarios, the upper bound \eqref{eqn:ModSel_simple} can serve as a close approximation to the \texttt{ModSel-CP} set, as the additional conditions specified in the theorem are typically satisfied. To illustrate this, we consider the following example.

\paragraph{Theorem~\ref{thm:ModSel_simplify} in practice: the residual score.} In the case of the residual score $S^{\lambda}(x, y) = |y - f_{\lambda}(x)|$, where $f_{\lambda}$ is a pretrained regression function, the upper bound~\eqref{eqn:ModSel_simple} is equal to a union of intervals,
\begin{multline}\label{eqn:ModSel_simple_res}
    \widehat{C}_{\texttt{ModSel-CP}} (X_{n+1}) \subseteq \left\{y\in\mathbb{R} : \min_{\lambda\in\mathcal{M}}|y-f_\lambda(X_{n+1})| \leq \hat{q}(\hat\lambda)\right\} \\=  \bigcup \limits_{\lambda \in \left\{\lambda' \in \Lambda: \hat{q}_{-}(\lambda') \le \hat{q}(\hat{\lambda}) \right\}} \left[f_{\lambda}(X_{n+1}) - \hat{q}(\hat\lambda) \ , \ f_{\lambda}(X_{n+1})  + \hat{q}(\hat\lambda)\right].
\end{multline}
In the meantime, the lower bound \eqref{eqn:ModSel_simple_lower} is also equal to a union of intervals,
\begin{multline}\label{eqn:ModSel_simple_res_lower}
    \widehat{C}_{\texttt{ModSel-CP}} (X_{n+1}) \supseteq \left\{y\in\mathbb{R} : \min_{\lambda\in\mathcal{M}_{-}}|y-f_\lambda(X_{n+1})| < \hat{q}(\hat\lambda)\right\} \\=  \bigcup \limits_{\lambda \in \left\{\lambda' \in \Lambda: \hat{q}_{-}(\lambda') < \hat{q}(\hat{\lambda}) \right\}} \left(f_{\lambda}(X_{n+1}) - \hat{q}(\hat\lambda) \ , \ f_{\lambda}(X_{n+1})  + \hat{q}(\hat\lambda)\right).
\end{multline}
Since in most regression settings the data are continuously distributed, we have $\mathcal{M}=\mathcal{M}_{-}$ almost surely. Consequently, comparing \eqref{eqn:ModSel_simple_res} and \eqref{eqn:ModSel_simple_res_lower}, we see that the upper and lower bounds coincide up to a set of measure zero. Thus, the \texttt{ModSel-CP} set differs from the upper bound \eqref{eqn:ModSel_simple_res} only on a null set, effectively making the upper bound an alternative expression of the \texttt{ModSel-CP} set.

Returning to the general scores, we can equivalently rewrite the upper bound~\eqref{eqn:ModSel_simple} as
\begin{equation}\label{eqn:ModSel_simple_2}
\widehat{C}_{\textnormal{\texttt{ModSel-CP}}} (X_{n+1}) \subseteq \bigcup_{\lambda\in\mathcal{M}}\left\{y \in \mathcal{Y}: \mathcal{L}_{n+1}(\lambda,S^\lambda(X_{n+1},y))\le\mathcal{L}_{n+1}({\hat\lambda},\hat{q}(\hat\lambda)) \right\}.\end{equation}
Hence,  \texttt{ModSel-CP} --- or more precisely, its upper bound used as a close approximation --- can be efficiently implemented, provided that sets of the form $\left\{y: \mathcal{L}_{n+1}(\lambda, S^{\lambda}(X_{n+1},y)) \le q \right\}$  are  easy to compute.

While the additional conditions in the theorem may occasionally fail to hold (e.g., due to ties among values), such cases most likely arise in discrete data settings, where $\mathcal{Y}$ is itself finite---i.e., a categorical response. However, this does not present a significant challenge to implementation in this case, since we can simply enumerate all possible values $y\in\mathcal{Y}$ and calculate the \texttt{ModSel-CP} prediction set directly---the upper and lower bounds are no longer necessary for computational efficiency.

\subsection{The \texttt{ModSel-CP-LOO} method}
As described above, the \texttt{ModSel-CP} method can be interpreted as a mechanism for restoring symmetry in the selection of $\hat\lambda$, by using all $n+1$ data points (the $n$ calibration points plus the test point) when choosing a model. We will now consider an alternative version of the method, where we restore symmetry in a different way, by considering  a leave-one-out (LOO) strategy for selecting the model.

Specifically, for each $i\in[n]$ and each $y\in\mathcal{Y}$, define the leave-one-out selected model as
\[\hat\lambda_{-i}(y) = \arg\min_{\lambda\in\Lambda}\mathcal{L}(\lambda,\hat{q}_{-i}(\lambda,y) ; X_1,\dots,X_{n+1}),\]
where $\hat{q}_{-i}(\lambda,y)$ is defined as \[\textnormal{Quantile}_{(1-\alpha)(1+1/n)}\left(S^\lambda(X_1,Y_1),\dots,S^\lambda(X_{i-1},Y_{i-1}),S^\lambda(X_{i+1},Y_{i+1}),\dots,S^\lambda(X_n,Y_n),S^\lambda(X_{n+1},y)\right).\]
In other words, while \texttt{ModSel-CP} uses the augmented data set $(X_1,Y_1),\dots,(X_n,Y_n),(X_{n+1},y)$ for model selection, here we instead use a leave-one-out version of this data set, removing the $i$th point $(X_i,Y_i)$.

We are now ready to define the prediction set for \texttt{ModSel-CP-LOO}. As before, we will write $\mathcal{L}_{n+1}(\lambda,q)$ to denote $\mathcal{L}(\lambda,q;X_1,\dots,X_{n+1})$, i.e., the subscript indicates that the width of the prediction set is being evaluated on all $n+1$ data points, $X_1,\dots,X_{n+1}$. 
\begin{multline}\label{defn:LOO}
\widehat{C}_{\texttt{ModSel-CP-LOO}} (X_{n+1}) = \Bigg\{y \in \mathcal{Y}: 
\mathcal{L}_{n+1}\big({\hat\lambda}, S^{\hat\lambda}(X_{n+1},y)\big) \leq {}\\
\textnormal{Quantile}_{(1-\alpha)(1+1/n)}\left(\left\{ \mathcal{L}_{n+1}\big({\hat\lambda_{-i}(y)},S^{\hat\lambda_{-i}(y)}(X_i,Y_i)\big)\right\}_{i\in[n]}\right)\Bigg\}.\end{multline}
As before, this variant of the method also offers marginal coverage.
\begin{theorem}[Validity of \texttt{ModSel-CP-LOO}]\label{thm:LOO-validity}
    Fix any $\alpha\in(0,1)$. Assume $(X_1,Y_1),\dots,(X_{n+1},Y_{n+1})$ are exchangeable, and $\{S^\lambda:\lambda\in\Lambda\}$ is any collection of pretrained models. Then \textnormal{\texttt{ModSel-CP-LOO}} satisfies a marginal coverage guarantee,
    \begin{equation*}
    	\mathbb{P} \left\{Y_{n+1} \in \widehat{C}_{\textnormal{\texttt{ModSel-CP-LOO}}}(X_{n+1}) \right\} \ge 1-\alpha.
    \end{equation*}
\end{theorem}
\noindent Like for the \texttt{ModSel-CP}, we will see in the proof that this result follows by interpreting \texttt{ModSel-CP-LOO} in the language of the full conformal prediction framework by defining a new meta-score via model selection.

Unlike our result in Theorem~\ref{thm:ModSel_simplify} giving an efficient implementation of \texttt{ModSel-CP},  computing an efficient representation of the prediction set for the \texttt{ModSel-CP-LOO} method is more complicated. We defer this question to the Appendix~\ref{sec:appendix_implement_LOO}.

\subsection{Comparing \texttt{ModSel-CP} and \texttt{ModSel-CP-LOO} to \texttt{YK-baseline}}\label{sec:ModSelvsLOO}

As mentioned at the beginning of Section~\ref{sec:method}, the construction of our method is built around the idea of restoring the symmetry required by the conformal procedure, which can be viewed as a correction to alleviate selection bias. Hence, both \texttt{ModSel-CP} and \texttt{ModSel-CP-LOO} can be seen as a selection-bias-corrected version of \texttt{YK-baseline} (though the two methods take different approaches towards the question of how to restore symmetry to perform this correction).
This perspective suggests that \texttt{ModSel-CP} and \texttt{ModSel-CP-LOO} should be more conservative than \texttt{YK-baseline}, in order to correct for selection bias, and indeed the following result confirms this intuition under a mild assumption.
\begin{proposition}\label{prop:containYK-nonadj}
    For any $\alpha\in(0,1)$, any collection of pretrained models $\{S^\lambda:\lambda\in\Lambda\}$, and any data set  $(X_1,Y_1),\dots,(X_n,Y_n)$ and test point $X_{n+1}$,
    it holds deterministically that if $\hat{\lambda}$ is the unique minimizer in~\eqref{eqn:define_hat_lambda} then
    \[\widehat{C}_{\textnormal{\texttt{YK-baseline}}}(X_{n+1}) \subseteq \widehat{C}_{\textnormal{\texttt{ModSel-CP}}}(X_{n+1}), \textnormal{ and } \widehat{C}_{\textnormal{\texttt{YK-baseline}}}(X_{n+1}) \subseteq \widehat{C}_{\textnormal{\texttt{ModSel-CP-LOO}}}(X_{n+1}).\]
\end{proposition}
\noindent Note that the assumption that $\hat{\lambda}$ is the unique minimizer in~\eqref{eqn:define_hat_lambda}  is mild in many common settings---e.g., in a regression setting with continuously distributed data, where a distinct class of models will all have distinct residuals almost surely.

While our two methods, \texttt{ModSel-CP} and \texttt{ModSel-CP-LOO}, share the property that both are strictly more conservative than \texttt{YK-baseline} (under a mild assumption), they differ in the exact way that they enlarge the set from \texttt{YK-baseline} to achieve validity---in Appendix~\ref{sec:ENSvsLOO-additional}, we compare the construction of \texttt{ModSel-CP} and \texttt{ModSel-CP-LOO} from two perspectives to illustrate such differences, which may pave the way for further understanding the performance difference between the two methods.

\section{Asymptotically optimal width}\label{sec: efficiency_res}
While the results above establish the (finite-sample) validity of our prediction sets, we now turn to the question of whether the \texttt{ModSel-CP} and \texttt{ModSel-CP-LOO} methods perform well in terms of the width of the resulting prediction sets. 
We will now verify theoretically that our methods offer prediction sets that have asymptotically optimal width, when the model class $\{S^\lambda : \lambda\in\Lambda\}$ satisfies certain degree of regularity. 
In addition, our experiments in Section~\ref{sec:experiment} below will examine the empirical effectiveness of our methods.

In this section, for clarity of the exposition, we will restrict our attention to a special case: the residual score, $S^\lambda(x,y) = |y-f_\lambda(x)|$ for a class of pretrained models $\{f_\lambda: \lambda\in\Lambda\}$. In Appendix~\ref{app:efficiency} we will generalize these results to a broader setting.
 
\subsection{Theoretical guarantee}
Recall the \texttt{YK-baseline} method~\eqref{eqn:YK-baseline_define} can be rewritten as (under the residual score),
\[\widehat{C}_{\texttt{YK-baseline}}(X_{n+1}) = C^{\hat\lambda}_{\hat{q}(\hat\lambda)}(X_{n+1}) = f_{\hat\lambda}(X_{n+1}) \pm \hat{q}(\hat\lambda).\]
That is, this is the prediction set returned by running split conformal prediction with the model $\lambda=\hat\lambda$. As we have seen, ignoring the fact that $\hat\lambda$ is chosen in a data-dependent way can lead to undercoverage in finite samples---the resulting prediction set may be too small. Thus, to ensure an assumption-free finite-sample coverage guarantee, it is necessary to use a finite-sample valid method such as \texttt{ModSel-CP}, though it may come at a cost of being more conservative than \texttt{YK-baseline}. However, the asymptotic result below reassures us that, once we are in a regime where asymptotics are a good approximation to the data, this cost becomes negligible, i.e., with high probability, the  \texttt{ModSel-CP} and \texttt{ModSel-CP-LOO}
prediction sets are contained in the set
\[C^{\hat\lambda}_{\hat{q}(\hat\lambda) + \kappa}(X_{n+1}) = f_{\hat\lambda}(X_{n+1}) \pm \big(\hat{q}(\hat\lambda) + \kappa\big)\]
for some small $\kappa>0$,
which is a slight inflation of the \texttt{YK-baseline} method.

\paragraph{Assumptions and notation.}
Before presenting our results, we begin with a few assumptions and definitions that allow us to establish regularity conditions on the class of models $\Lambda$. Throughout this section, we will assume that the data points are i.i.d.\ (rather than exchangeable, which is sufficient for the validity results of Section~\ref{sec:method}), and will write $P$ to denote the distribution from which data points $(X_i,Y_i)$ are drawn.

\begin{definition}[Approximately best models for residual score]\label{defn: approx_best_mdl_res}
	Fix any $\Delta>0$. Then the collection of $\Delta$-approximately-best models associated with the model class $\Lambda$ is defined as
	\begin{equation*}
		\Lambda_{\ast}(\Delta) =  \left\{\lambda \in \Lambda: Q_{1-\alpha-\Delta}(\lambda)\leq  \min_{\lambda'\in\Lambda}Q_{1-\alpha+\Delta}(\lambda')\right\},
	\end{equation*}
	where, for each $\lambda\in\Lambda$, 
	$Q_\tau(\lambda)$  denotes the $\tau$-quantile of $S^\lambda(X,Y)$ for $(X,Y)\sim P$. 
\end{definition}
\noindent In other words, $\Lambda_{\ast}(\Delta)$ is a subset of $\Lambda$, identifying all models $\lambda$ that are approximately optimal in terms of the resulting residual of the model. 

Next, we define
\[\kappa(\Delta) = \max_{\lambda,\lambda'\in\Lambda_*(\Delta)} |f_\lambda(X_{n+1}) - f_{\lambda'}(X_{n+1})|, \]
and
\[\kappa_{\texttt{LOO}}(\Delta) = \max_{\lambda,\lambda'\in\Lambda_*(\Delta)}\max_{i=1,\dots,n}|f_\lambda(X_i) - f_{\lambda'}(X_i)|,\]
which we will use in the results for \texttt{ModSel-CP} and \texttt{ModSel-CP-LOO} respectively.
We will discuss these definitions in more detail below, but to give some brief intuition: our theory will rely on $\kappa(\Delta)$ and $\kappa_{\texttt{LOO}}(\Delta)$ being small when $\Delta$ is small, which essentially means that we must have $f_{\lambda}\approx f_{\lambda'}$ for any two models $f_{\lambda},f_{\lambda'}$ that are both approximately optimal in terms of their residual error.

We will need one more piece of notation: we define
\[\mathcal{R}_n(\Lambda) = \sup_{\{t_\lambda\}_{\lambda\in\Lambda}}\mathbb{E}\left[\sup_{\lambda\in\Lambda}\left|\frac{1}{n}\sum_{i=1}^n \xi_i \mathbf{1}\{|Y_i - f_\lambda(X_i)| \leq t_\lambda\}\right|\right], \]
where the expectation is taken with respect to $(X_i,Y_i)\stackrel{\textnormal{i.i.d.}}{\sim}P$ and $\xi_i\stackrel{\textnormal{i.i.d.}}{\sim}\textnormal{Unif}\{\pm 1\}$.  This is the Rademacher complexity \citep{shalev2014understanding} of the class of functions $\left\{\mathbf{1} \left\{|y-f_\lambda(x)| \le t_\lambda \right\}: \lambda \in \Lambda\right\}$ (or rather, the maximum Rademacher complexity, over any choice of the thresholds $\{t_\lambda\}_{\lambda\in\Lambda}$).

\paragraph{Theoretical bound.} We are now ready to state our result for width optimality of the residual score class. (See Theorem~\ref{thm:efficiency_general_} in the Appendix for a more general version.)

\begin{theorem}\label{thm:efficiency_res_}
Let $(X_i,Y_i)\stackrel{\rm i.i.d.}{\sim}P$.
Assume that $|Y-f_\lambda(X)|$ has a continuous distribution under $(X,Y)\sim P$, for each $\lambda\in\Lambda$. 
Let $\Delta_n = 2\mathcal{R}_n(\Lambda)+   \sqrt{\frac{\log(1/\delta)}{2n}}+\frac{2}{n}.$ Then it holds that,
\[\mathbb{P} \left\{\widehat{C}_{\textnormal{\texttt{ModSel-CP}}} (X_{n+1}) \subseteq f_{\hat\lambda}(X_{n+1})\pm \left(\hat{q}(\hat\lambda) + \kappa(\Delta_n)\right) \right\} \ge 1-2\delta.\]
 \[\mathbb{P} \left\{ \widehat{C}_{\textnormal{\texttt{ModSel-CP-LOO}}} (X_{n+1}) \subseteq  f_{\hat\lambda}(X_{n+1})\pm \left(\hat{q}(\hat\lambda) + \kappa_{\textnormal{\texttt{LOO}}}(\Delta_n)\right) \right\} \ge 1-2\delta.\]
\end{theorem}

\subsection{A closer look at the regularity conditions}\label{sec:examples_for_optimality}
Theorem~\ref{thm:efficiency_res_} gives an upper bound on the  prediction sets $\widehat{C}_{\texttt{ModSel-CP}}$ and $\widehat{C}_{\texttt{ModSel-CP-LOO}}$, but these upper bounds are meaningful (i.e., are close to the width of $\widehat{C}_{\textnormal{\texttt{YK-baseline}}}$) only under certain conditions. Specifically, we need $\Delta_n$ to be small (i.e., the Rademacher complexity $\mathcal{R}_n(\Lambda)$ must be small), and we need the continuity terms $\kappa(\Delta_n)$ or $\kappa_{\texttt{LOO}}(\Delta_n)$ to be small as well. In what settings will this hold?

In what follows, we examine an example:
a setting where $\Delta_n$ is vanishing as $n \to \infty$, and both $\kappa(\Delta_n)$ and $\kappa_{\texttt{LOO}}(\Delta_n)$ are indeed small with high probability.
In this example, we will consider a Gaussian linear regression setting, where $\Lambda$ is a collection of pretrained linear models---for instance, obtained by running the Lasso (on a separate pretraining data set) with a range of different values for the penalty parameter.

\begin{proposition}\label{prop:positive_example} 
    Consider a Gaussian linear model setting, where $(X,Y)\in\mathbb{R}^d\times\mathbb{R}$ has the distribution $X\sim \mathcal{N}(0,\Sigma), \ Y\mid X \sim \mathcal{N}(X^\top\lambda_{\rm true},\sigma^2).$
    Consider a class of $s$-sparse linear models $\{f_\lambda : \lambda\in\Lambda, \|\lambda\|_0 \le s\}$, where $\Lambda\subseteq\mathbb{R}^d$ is any collection of coefficient values for models $f_\lambda(x) = x^\top\lambda$ (note in particular that we do not assume $\lambda_{\rm true}\in\Lambda$). Then,  there exists some constant $C>0$ such that
    \[\mathcal{R}_{n}(\Lambda) \le C \cdot  \sqrt{\frac{s \mathrm{log}(ed/s)}{n}}.\]
    Moreover, for any $\Delta\in(0,1-\alpha)$ and $\delta\in(0,1)$, it holds that
    \[\mathbb{P}  \left\{\kappa(\Delta) \le 8\sigma \sqrt{\frac{\eta_{\max}(\Sigma)}{\eta_{\min}(\Sigma)}}\cdot\sqrt{(1+\gamma^2)(1+\phi)^2 - 1}\cdot \sqrt{ s\mathrm{log}\left(\frac{ed}{s} \right) + \mathrm{log}(1/\delta)}\right\} \ge 1-\delta,\]
     where $\eta_{\min}(\Sigma)$ and $\eta_{\max}(\Sigma)$ denote the smallest and the largest eigenvalue of $\Sigma$, and where
    \[ \gamma^2 = \min_{\lambda\in\Lambda} \frac{(\lambda - \lambda_{\rm true})^\top \Sigma (\lambda - \lambda_{\rm true})}{\sigma^2}, \]
    and
    \[\phi = \frac{\Phi^{-1}\left(1 - \frac{\alpha -\Delta}{2}\right)}{\Phi^{-1}\left(1 - \frac{\alpha + \Delta}{2}\right)} - 1 = \mathcal{O}(\Delta),\]
    for $\Phi$ denoting the CDF of the standard normal distribution.
    Similarly, it holds that
    \[\mathbb{P}  \left\{\kappa_{\textnormal{\texttt{LOO}}}(\Delta) \le 8\sigma \sqrt{\frac{\eta_{\max}(\Sigma)}{\eta_{\min}(\Sigma)}}\cdot\sqrt{(1+\gamma^2)(1+\phi)^2 - 1}\cdot \sqrt{ s\mathrm{log}\left(\frac{ed}{s} \right) + \mathrm{log}(n/\delta)}\right\} \ge 1-\delta.\]
\end{proposition}

In particular, Proposition~\ref{prop:positive_example} implies that
\[\kappa(\Delta) = \mathcal{O}_P\left(\sqrt{s\mathrm{log}\left(\frac{ed}{s} \right)} \cdot (\gamma + \Delta)\right),\quad \kappa_{\textnormal{\texttt{LOO}}}(\Delta) = \mathcal{O}_P\left(\sqrt{s\mathrm{log}\left(\frac{ed}{s} \right) + \log(n)} \cdot (\gamma + \Delta)\right).\]
Recall from above that, in the setting of a $d$-dimensional $s$-sparse linear model, we can take $\Delta = \Delta_n \propto \sqrt{\frac{s\mathrm{log}\left(ed/s \right)}{n}}$.
If $\lambda_{\rm true}\in\Lambda$ (and so $\gamma=0$), then we see that
\[\kappa(\Delta_n) = \mathcal{O}_P\left(\frac{s\mathrm{log}\left(ed/s \right)}{\sqrt{n}}\right),\quad \kappa_{\textnormal{\texttt{LOO}}}(\Delta_{n}) = \mathcal{O}_P\left(\frac{s\mathrm{log}\left(ed/s \right) + \log(n)}{\sqrt{n}}\right),\]
i.e., the bounds are vanishing as long as $n\gg \left(s\mathrm{log}(ed/s) \right)^2$.
If instead $\gamma\approx 0$---i.e., $\Lambda$ might not contain the true coefficient vector $\lambda_{\rm true}$, but contains some good approximation---then these upper bounds can be very small as $\Delta\to 0$. For example, if $\|\lambda_{\textnormal{true}}\|_0 = s$ and when $\Lambda$ is obtained by running the Lasso, then the typical rate of estimation is given by $\sqrt{\frac{s\mathrm{log}(d)}{n_{\textnormal{train}}}}$. In other words, $\gamma$ is typically $\mathcal{O} \left(\sqrt{\frac{s\mathrm{log}(d)}{n_{\textnormal{train}}}} \right)$. Since it is often the case where $n_{\textnormal{train}} = n$, i.e., the entire dataset is split into two subsets where one is used for pretraining and the other is used for calibration, we have that $\gamma$ is also vanishing with $n \to \infty$, leading to vanishing $\kappa$ even if the true coefficient $\lambda_{\textnormal{true}}$ is not in the set $\Lambda$. This shows that the regularity conditions required for Theorem~\ref{thm:efficiency_res_} to give a meaningful upper bound will hold in the common setting of linear regression.

\section{Experiments}\label{sec:experiment}
In this section, we study empirically the performance of \texttt{ModSel-CP} and \texttt{ModSel-CP-LOO}, comparing them to the existing methods of \citet{yang2024selection} (namely, \texttt{YK-baseline} and \texttt{YK-split}, along with the adjusted version \texttt{YK-adjust} of the baseline method) to demonstrate their effectiveness in finite samples. Since we are working with settings where the gap between \texttt{ModSel-CP} and the upper bound~\eqref{eqn:ModSel_simple} is at most a set of measure zero, we then implement \texttt{ModSel-CP} using the upper bound \eqref{eqn:ModSel_simple}. For the implementation of \texttt{ModSel-CP-LOO},  see Appendix~\ref{sec:appendix_implement_LOO} for details. Code to reproduce all the experiments is available at \url{https://github.com/RuitingL/ModSel-CP}.

To better assess the improvement in the width of prediction sets, we also  compare with the size of the best model in the class---that is, the minimal width obtained by any single model $\lambda\in\Lambda$ via  conformal prediction, $\min_{\lambda \in \Lambda} \mathbb{E} \left|\widehat{C}^{\lambda}_{\hat{q}(\lambda)}(X_{n+1}) \right|$. Throughout,  we fix $\alpha=0.1$  and take $\mathcal{L}(\lambda,q; x_1,\dots, x_{n+1}) = \frac{1}{n+1}\sum_{i=1}^{n+1} \left|C^{\lambda}_{q}(x_{i}) \right|$. 

\subsection{Simulation}\label{sec:simulation}
For simulations, we consider a regression setting, and simulate i.i.d.\ data from several different distributions. Additional experiments are deferred to the Appendix~\ref{sec:appendxi_addition_sim}.

\subsubsection{Simulation setting}

Fix the dimension of feature $d=300$. Let $Y = X^{\top}\theta + \epsilon,$ where $\theta \in \mathbb{R}^{d}$ and the noise $\epsilon$ is independent of $X$. Fix $\nu = 3$. We run experiments with four different data distributions on $(X,Y)\in\mathbb{R}^d\times\mathbb{R}$:
    \begin{itemize}
        \item \textbf{NormalX $+$ sparse $+$ Gaussian noise:} $X \sim \mathcal{N}(0, \mathbf{I}_d), \quad \epsilon \sim \mathcal{N}(0,1)$, where $\theta_{j} = \mathbf{1}\{j \textnormal{ mod } 20=0\}$ for each $j\in [d]$.

        \item \textbf{NormalX $+$ sparse structure $+$ heavy-tailed noise:} $X \sim \mathcal{N}(0, \mathbf{I}_d), \quad \epsilon \sim t_{\nu}(1)$, where $\theta_{j} = \mathbf{1}\{j \textnormal{ mod } 20=0\}$ for each $j\in [d]$.

        \item \textbf{NormalX $+$ dense structure $+$ Gaussian noise:} $X \sim \mathcal{N}(0, \mathbf{I}_d), \quad \epsilon \sim \mathcal{N}(0,1/d^{2})$, where $ \theta_{j} = \frac{1}{d}$ for each $j\in [d]$.

        \item \textbf{tX $+$ sparse structure $+$ Gaussian noise:} $X \sim t_{\nu}(0,\mathbf{I}_d), \quad \epsilon \sim \mathcal{N}(0,1)$, where $\theta_{j} = \mathbf{1}\{j \textnormal{ mod } 20=0\}$ for each $ j\in [d]$.
    \end{itemize}

\subsubsection{Implementing the methods}
For pretraining the models, we use $n_{\textnormal{train}} = 300$ many data points (independent of the calibration set $(X_1,Y_1),\dots,(X_n,Y_n)$). We consider the residual score class fitted via ridge regression, $\{|y-x^\top\hat\theta_\lambda| : \lambda\in\Lambda\}$. (Experiments conducted with the rescaled residual scores are deferred to the Appendix~\ref{sec:appendix_sim_rescaled_residual}.)
Each $\hat\theta_\lambda$ is obtained  by first uniformly at random selecting $10\%$ features, then fitting a ridge regression (penalty $\eta =0.1$) with the projected data, and finally embedding it back to the $d$-dimensional space (setting the rest of coordinates to be $0$). In other words, each $\lambda\in\Lambda$ corresponds to a (random) selection of 10\% of the features, for inclusion in the model. 
As a final detail, when  implementing the \texttt{YK-split} method, we split the hold-out set in half for selection and calibration separately.

\subsubsection{Results}
For each setting of our experiment, we repeat 5000 independent trials of the simulation. We vary the sample size $n$, and also, the size of the model class, $|\Lambda|$.

Figure~\ref{fig:residual_n100} shows the results for varying model class size $|\Lambda|$ while sample size is fixed at $n=100$, while Figure~\ref{fig:residual_m200} shows the results for varying sample size $n$ while model class size is fixed at $|\Lambda|=200$. The top row of each figure shows the average coverage (with target level $1-\alpha = 90\%$), while the bottom row shows the ratio between the width of prediction set from each method and the smallest conformal prediction set constructed with any single model in the class $\Lambda$. All plots of the results show the average over 5000 trials, with standard errors indicated via shaded regions. Table summary of the results are in Appendix~\ref{app:table_summary_simulation}.

\begin{figure}  
\centering
\includegraphics[scale = 0.32]{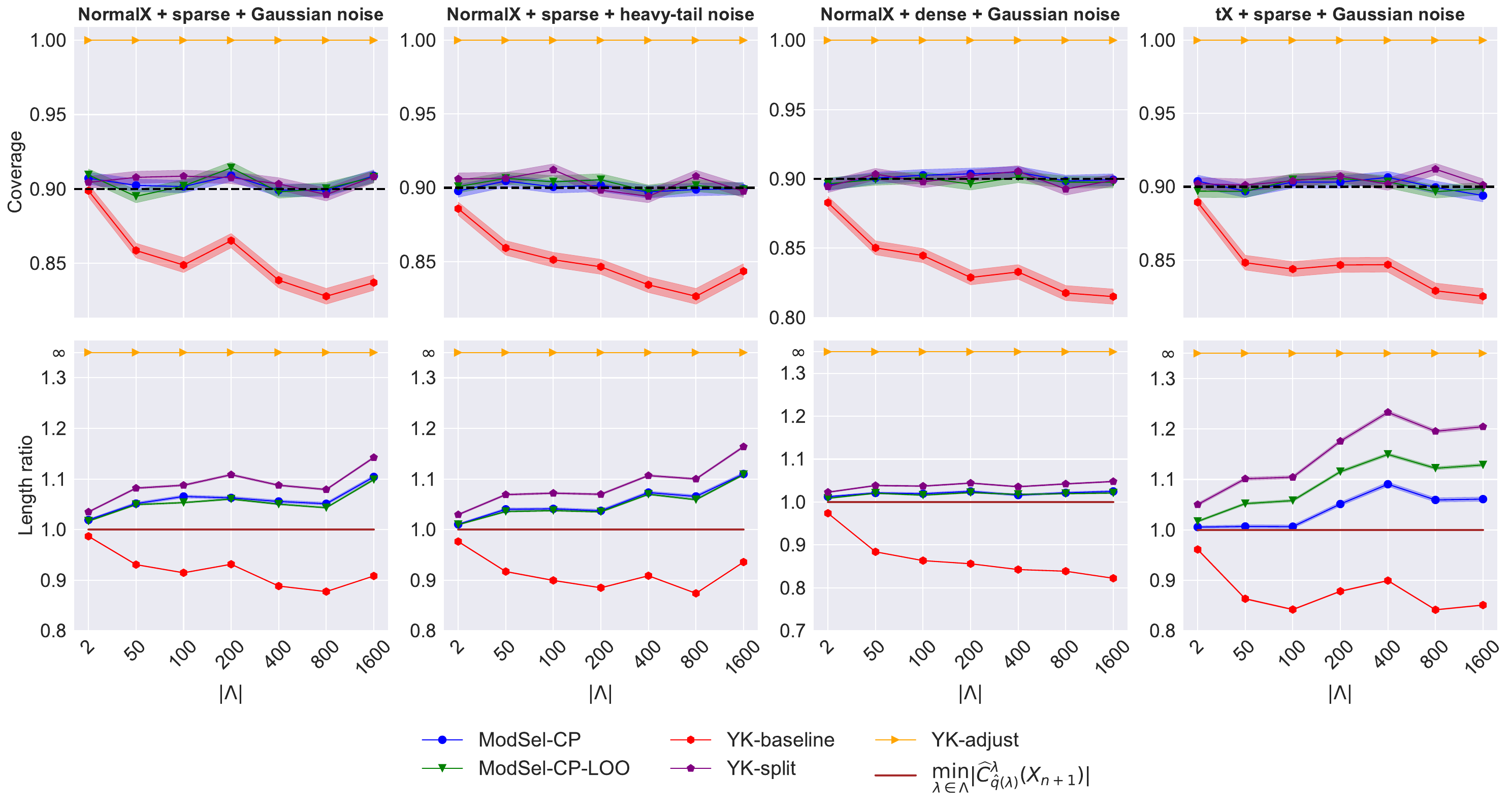}
\caption{Residual score, $n = 100$, varying $|\Lambda|$.} 
\label{fig:residual_n100}
\vspace{0.7cm}

\includegraphics[scale = 0.32]{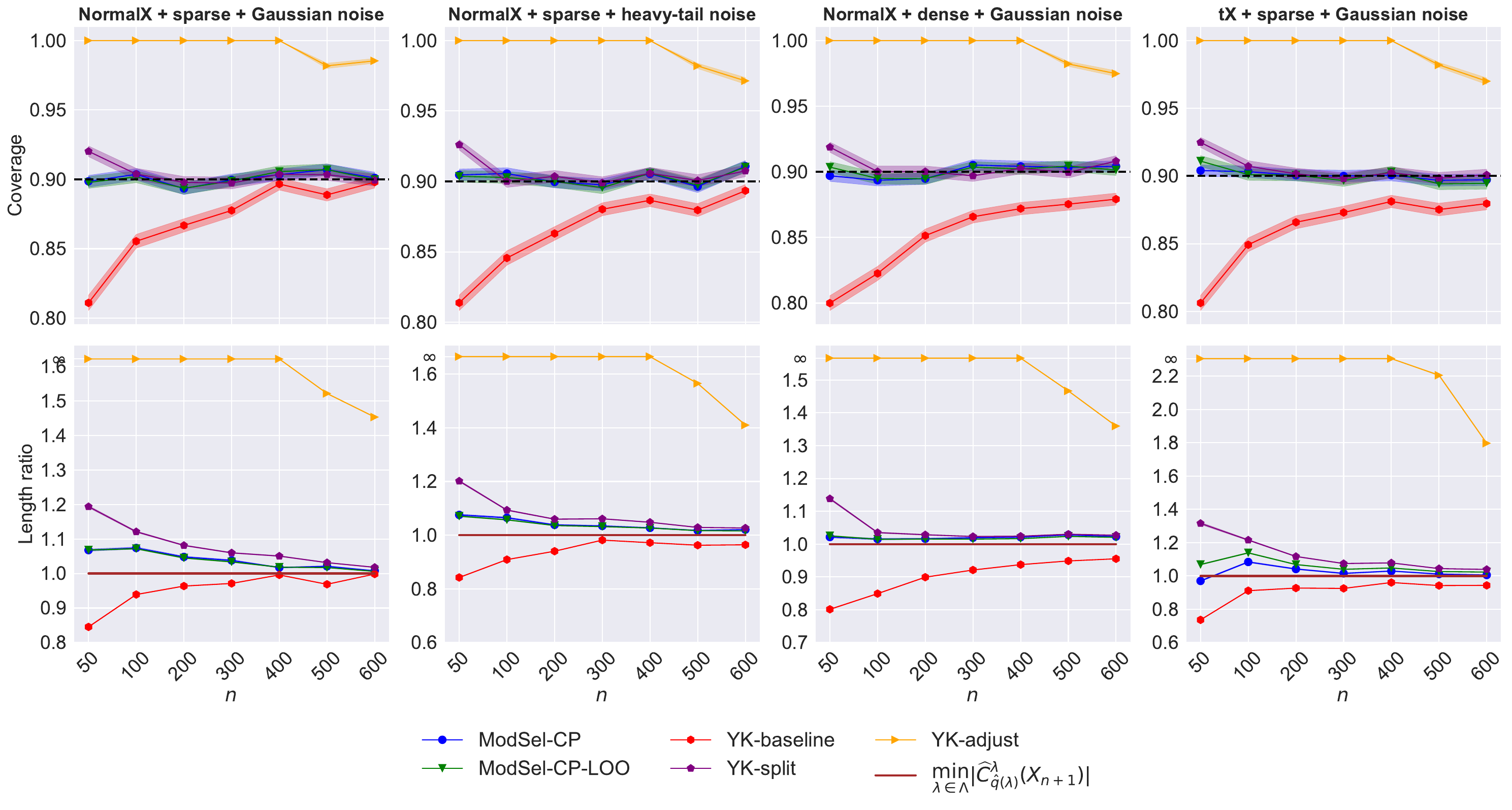}
\caption{Residual score, $|\Lambda|=200$, varying $n$.} 
\label{fig:residual_m200}
\end{figure}

Overall, we observe the following results from our simulation:
\begin{itemize}
    \item All methods except for \texttt{YK-baseline} achieve $1-\alpha$ coverage as expected. In particular, \texttt{YK-baseline} loses coverage significantly in finite samples as the number of models increases (which is as expected, since it does not correct for selection bias), whereas its naively adjusted counterpart, \texttt{YK-adjust}, is severely conservative. It is worth mentioning that, when $n$ is relatively small,  \texttt{YK-adjust} will always return the entire set $\mathcal{Y}$ as the prediction set---this is simply the real line $\mathbb{R}$ in our experiments, i.e., the prediction set has infinite width, even if $|\Lambda|$ is as small as $2$ since the value $n=100$ is not large. (In additional experiments conducted, where  we raise $n$ to $400$, \texttt{YK-adjust} is still very conservative even when its width is finite.)

    \item In terms of the width of prediction sets in finite samples, both \texttt{ModSel-CP} and \texttt{ModSel-CP-LOO} tend to yield smaller prediction sets than \texttt{YK-split} and \texttt{YK-adjust}. 
    Moreover, both \texttt{ModSel-CP} and \texttt{ModSel-CP-LOO} can output prediction sets with width very close to (or even smaller than in some cases) the smallest conformal prediction set constructed with any single model, which is consistent with our theoretical results in Section~\ref{sec: efficiency_res}.

    \item \texttt{ModSel-CP} and \texttt{ModSel-CP-LOO} perform similarly, and generally yield intervals of fairly similar width. Whichever method outputs a smaller prediction set, \texttt{ModSel-CP} or \texttt{ModSel-CP-LOO}, depends on the data generating distribution and the given model class. However, based our experiments and efficiency of implementation, we advocate for \texttt{ModSel-CP} owing to its simplicity.

\end{itemize}

\subsection{Real data example} 
We apply the methods under comparison to the protein structure dataset from UCI repository\footnote{\url{https://archive.ics.uci.edu/dataset/154/protein+data}}. This dataset contains $N=45730$ data points $(X,Y)$ with the dimension of feature $X$ being $d=9$. We first uniformly at random choose $n_{\textnormal{train}} = 300$ from this dataset to train the models. For this experiment, we consider the CQR class, namely $\left\{\max \left(\hat{Q}^{\lambda}_{\beta}(x)-y, y-\hat{Q}^{\lambda}_{1-\beta}(x) \right): \lambda \in \Lambda, \beta \in (0,1) \right\}$. The \texttt{RandomForestQuantileRegressor} function in Python's \texttt{quantile$\_$forest} package is applied to train these quantile regressors, varying \texttt{n$\_$estimator} over $\{100, 200, 300, 400\}$, \texttt{max$\_$feature} over $\left\{\frac{d}{10}, \frac{2d}{10}, \dots, d\right\}$,  and leaving all other parameters at their default values. Along with $10$ values for $\beta$ as an equispaced sequence from $10^{-4}\alpha$ to $4\alpha$, the total number of models is $|\Lambda|=400$. (This experiment setting is a slight modification of the experiment setting in \citet{yang2024selection}'s work.) It is worth pointing out that in this setting where $\beta$ needs to be determined, it is not immediately clear how to apply cross-validation for model selection.

In order to estimate the coverage and the width of the prediction sets, we draw independently $600$ observations without replacement from the rest of dataset (of size $N-n_{\textnormal{train}}$) for $100$ times to create independent versions of the datasets. $500$ of these observations are used to construct prediction sets, i.e. $n=500$, and the remaining $100$ observations are used to estimate the coverage probability. Specifically,
\begin{equation*}
    \textnormal{Empirical marginal coverage} = \frac{1}{100} \sum_{b=1}^{100} \left(\frac{1}{100} \sum_{i=1}^{100} \mathbf{1} \left\{Y_{i}^{(b)} \in C^{(b)}(X_{i}^{(b)}) \right\} \right)
\end{equation*}
\begin{equation*}
    \textnormal{Empirical average width} = \frac{1}{100} \sum_{b=1}^{100} \left(\frac{1}{100} \sum_{i=1}^{100} \left| C^{(b)}(X_{i}^{(b)}) \right| \right)
\end{equation*}
Here, $(X^{(b)}_{i}, Y^{(b)}_{i})$ is the $i$-th test point in the $b$-th version of a dataset, and $C^{(b)}(\cdot)$ is a prediction set constructed from the $b$-th version of a dataset.

\subsubsection{Results}
The result of this experiment is given below,  with standard errors shown in parentheses.
\begin{table}[h!]
    \centering
    \renewcommand{\arraystretch}{2}
    \begin{tabular}{c|c|c|c|c|c}
    \rowcolor{gray!20}
        \ & \texttt{ModSel-CP} & \texttt{ModSel-CP-LOO} & \texttt{YK-baseline} & \texttt{YK-adjust} & \texttt{YK-split}  \\
        \hline
        Coverage & 0.899 \ (0.004) & 0.899 \ (0.004) & 0.881 \ (0.004)& 0.992 \ (0.001) & 0.902 \ (0.004)\\ 
        \hline
        Width & 14.18 \ (0.09) & 14.14 \ (0.08) & 13.56 \ (0.06)& 19.76 \ (0.06) & 14.44 \ (0.10)\\
        \hline
    \end{tabular}
    \caption{Results on the protein structure dataset. Evaluation results of coverage and width are reported in the table, with standard errors for the averages shown in parentheses. }
\end{table}
For comparison, the smallest width obtained by any single model in this experiment is $\min \limits_{\lambda \in \Lambda} \mathbb{E} \left|\widehat{C}^{\lambda}_{\hat{q}(\lambda)}(X_{n+1}) \right| = 14.01$. In general, here we see similar trends as in the simulation setting. The width of the \texttt{ModSel-CP} and \texttt{ModSel-CP-LOO} intervals are similar to the smallest width $14.01$, while \texttt{YK-adjust} is much too conservative, and the width of \texttt{YK-split} is slightly larger as well. In this particular example, the undercoverage of \texttt{YK-baseline} is quite mild relative to the severe selection bias issue in the simulations, but nonetheless there is some loss of coverage.

\section{Discussion and related work}\label{sec:discuss}

\subsection{Related literature}\label{sec:related_literature}
Many previous works have considered the problem of constructing valid prediction sets while optimizing the width of the set.
In the field of conformal prediction, prior work such as \citet{sesia2020comparison, lei2018distribution, burnaev2014efficiency} (among others), suggests that as long as we have a consistent estimator of the underlying distribution, the conformal prediction set will be width-optimal asymptotically. Therefore, optimizing the width of prediction sets ultimately hinges on training a good predictor and conformity score that can accurately capture the underlying data distribution.

To better understand the range of different approaches to this question, we can view a conformal prediction method as following the flowchart illustrated below: any (split) conformal based method needs to train a fitted model (or multiple models) and design the corresponding score function(s) during the pretraining phase, and then use the calibration data to set a threshold for the score as well as to perform model selection if appropriate.

\begin{center}
\resizebox{\linewidth}{!}{
\begin{tikzpicture}[
    node distance=2cm and 2cm,
    block/.style={rectangle, minimum width=3cm, minimum height=2.4cm, text centered, draw=black, font=\normalsize},
    blueblock/.style={block, fill=blue!20},
    greenblock/.style={block, fill=green!25},
    redblock/.style={block, fill=red!25},
    arrow/.style={thick, -{Stealth[scale=1.2]}}
  ]

  \node (a) [blueblock, text width=3.5cm, align=center] {Step (a)\\ \vspace{0.5em}Train predictors};
  \node (b) [blueblock, right=of a, text width=3.5cm, align=center] {Step (b)\\\vspace{0.5em}Choose \\conformity score};
  \node (c) [redblock, right=of b, text width=3.5cm, align=center] {Step (c) \\ \vspace{0.5em}Model selection};
  \node (d) [greenblock, right=of c, text width=3.5cm, align=center] {Step (d)\\\vspace{0.5em}Conformal calibration};

  \begin{pgfonlayer}{background}
    \node [rectangle, rounded corners, fill=yellow!20, draw=none, inner sep=0.7cm, fit=(a)(b), label=below:{\large Pretraining stage using $\mathcal{D}_{\mathrm{train}}$}] {};
  \end{pgfonlayer}
  \begin{pgfonlayer}{background}
    \node [rectangle, rounded corners, fill=yellow!20, draw=none, inner sep=0.7cm, fit=(c)(d), label=below:{\large Calibration stage using $\mathcal{D}_{\mathrm{cal}} \cup \{X_{\mathrm{test}}\}$}] {};
  \end{pgfonlayer}

  \draw [arrow] (a) -- (b);
  \draw [arrow] (b) -- (c);
  \draw [arrow] (c) -- (d);

\end{tikzpicture}
}
\end{center}

With this illustration of the different components of the pipeline for implementing a method within the framework of conformal prediction, we can now broadly categorize various efforts proposed in the literature to optimize width of the resulting intervals:
\begin{itemize}
    \item \textbf{Train a better predictor---improve step (a):} \citet{stutzlearning} consider training an optimal classifier that can yield small conformal prediction sets by simulating conformal calibration during training (step (a)). Their method produces a single, improved classifier---measured by conformal prediction set width, since its training process is specifically designed to align with the goal of subsequent conformal prediction set construction.
    Relatedly, the conditional boosting technique developed in \citep{cherian2024llm} can be seen as a conditionally valid counterpart of such alignment, as they adopted the conformal calibration (step (d)) proposed by \citet{gibbs2025conformal} for set construction.

    \item \textbf{Design a more adaptive conformity score---improve step (b):} By designing scores that can adapt better locally, we can achieve improved conditional validity empirically, and thus overall smaller sets. See, for example, the work of \citet{lei2018distribution, romano2019conformalized, romano2020classification,angelopoulosuncertainty}.

    \item \textbf{Ensemble for a better predictor/score during pretraining:} \citet{xie2024boosted} propose to add an extra boosting step after step (a) and (b) during the pretraining stage, in order to obtain a boosted score function. To certain extent, our work shares a similar idea of ensembling as theirs. However, the key difference is that, in the work of \citet{xie2024boosted}, the boosting step happens during the training stage using only the training data, and the output of this process is a single score ready for the latter calibration (step (d)). In contrast, our work employs model selection to ensemble during the conformal calibration, using the calibration data $\mathcal{D}_{\textnormal{cal}}$ and the test point.

    \item \textbf{Add a model selection step (c) during calibration:}
    As mentioned earlier in Section~\ref{sec:exist_mtd}, \citet{yang2024selection} propose selecting from a set of conformal prediction sets with the smallest width to improve efficiency. However, their methods either suffer from a loss of coverage or exhibit lower statistical efficiency. Our work also contributes from this perspective: including an additional model selection step (c) during the calibration stage to optimize width while maintaining marginal validity.
\end{itemize}

\noindent Moving beyond the field of conformal prediction, \citet{bai2022efficient,fan2023utopia,kiyani2024length,chen2021learning} formulate the problem from the perspective of constrained optimization, i.e. minimizing the average prediction intervals width with the constraint of valid empirical coverage. 
We point out that our work is different from theirs because we exploit the data exchangeability to achieve exact finite-sample coverage without further sample-splitting. Whereas, they only achieve approximate coverage in theory---a PAC type of statement, and require bounded complexity of the model class for the theory to be applicable, since their calibration process involves only the calibration data $\mathcal{D}_{\textnormal{cal}}$ (without using the test point). Notably, although \citet{bai2022efficient} do mention an alternative version of their method that achieves the exact marginal coverage, it requires additional sample-splitting as in \citet{yang2024selection}, which can be costly for a small sample size.

Another related line of work considers merging a collection of prediction sets in a non-conservative way while trying to maintain the validity. For example, \citet{gasparin2024merging,gasparin2024conformal,cherubin2019majority,solari2022multi} use majority vote to combine conformal prediction sets that may have arbitrary dependence.  \citet{gasparin2024merging,gasparin2024conformal} show that the merged set will have at least $1-2\alpha$ marginal coverage, if given sets all have $1-\alpha$ marginal coverage, and the size of the merged set will be no worse than the maximum size of those given sets. Because the primary focus of this line of work is on effective sets aggregation, which does not necessarily minimize the size of the merged set, whereas we are interested in selecting models that can be used to construct small prediction sets, this approach is studying a fundamentally different problem than the one we address.

Lastly, we note that prior work of \citet{liang2023conformal} also explored the idea of incorporating the hypothesized test point into the model selection step (c) as summarized in the flowchart above.
However, our work distinguishes from theirs in terms of the model selection criterion, and more fundamentally the goal of model selection. Specifically, we focus on selecting the model with the minimum conformal prediction set width, which achieves an alignment between selection objective and the latter conformal calibration (set construction). Such alignment results in the simplified alternative expression of \texttt{ModSel-CP}, i.e., the (slightly more conservative) upper bound \eqref{eqn:ModSel_simple}. In contrast, \citet{liang2023conformal} focus on selecting the model with the minimum prediction error, e.g., the empirical mean of square error in the regression setting, which is not fully aligned with the latter conformal prediction set construction that deals primarily with quantile thresholding.

\subsection{Conclusions and future direction}
Constructing a valid prediction set while keeping its width as small as possible is vital for effective prediction task, as it enhances the precision and reliability of predictions. In this paper, we propose novel, easy-to-implement methods to address this problem, showing our methods have finite-sample validity without additional data-splitting and demonstrating their effectiveness both theoretically and empirically.

The results of this work suggest a number of potential directions for further study. First, while \texttt{ModSel-CP} has a very simple alternative expression for computing the prediction set, the computation of the \texttt{ModSel-CP-LOO} prediction set is (at least to our knowledge) more complex. This suggests many avenues for further study, both in terms of understanding how to efficiently implement the method, and in terms of finding more intuitive expressions of the construction to enable a tighter theoretical analysis, such as  theoretical comparison with \texttt{ModSel-CP} in terms of width of prediction sets in finite-sample, which has thus far only been investigated empirically.
Finally, as mentioned previously, our methods are based on the idea of full conformal, meaning that exchangeability plays an important role in guaranteeing validity. Recent developments (e.g. \citep{barber2023conformal}) have extend the conformal prediction beyond the exchangeable setting. Whether these methods can be adapted to the non-exchangeable case remains unexplored. We leave these open questions for future work.

\subsection*{Acknowledgements}
R.L. and R.F.B. were partially supported by the Office of Naval Research via grant N00014-24-1-2544. R.F.B. was partially supported by the National Science Foundation via grant DMS-2023109.

\bibliographystyle{plainnat}
\bibliography{bib}

\appendix

\section{Proofs for theoretical results in Section~\ref{sec:method}}

\subsection{Preliminaries: breaking ties when selecting a model}\label{app:tie-breaking}
In this section, we formally outline the tie-breaking procedure for selecting a model when there are multiple models attaining the minimum of the loss.

Denote the tie-breaking mechanism as $\tau$, which is a map that inputs a subset of model indices $\Lambda' \subseteq \Lambda$ and an auxiliary random seed $\xi \sim \mathrm{Unif}([0,1])$, and returns an element from that subset of indices: \[(\Lambda', \xi) \mapsto \tau(\Lambda', \xi) \in \Lambda'. \]
Hence, selecting a model $\hat{\lambda}$ \eqref{eqn:define_hat_lambda} formally means that, \[\hat{\lambda} = \tau\left( \mathrm{arg}\min \limits_{\lambda \in \Lambda} \mathcal{L}(\lambda,\hat{q}(\lambda); X_{1},\dots, X_{n+1}), \xi \right),\] for a random seed $\xi \sim \mathrm{Unif}([0,1])$ drawn independently of the data $(X_{1},Y_1),\dots, (X_n,Y_n),$ $X_{n+1}$.
Similarly, selecting a model $\hat{\lambda}(y)$ formally means that, \[\hat{\lambda}(y) = \tau\left( \mathrm{arg}\min \limits_{\lambda \in \Lambda} \mathcal{L}(\lambda,\hat{q}(\lambda,y); X_{1},\dots, X_{n+1}), \xi \right),\] with the same random seed $\xi$. The same tie-breaking mechanism applies to $\hat{\lambda}_{-i}(y)$ for each $i \in [n]$ as well. 

We note that throughout the paper,  $\hat{\lambda}, \hat{\lambda}(y)$ and $\hat{\lambda}_{-i}(y)$, for any $y$ and any $i\in [n]$, are single elements obtained by applying the tie-breaking map $\tau$  with some random seed $\xi \sim \mathrm{Unif}([0,1])$, and we omit explicitly writing the map $\tau$ outside any $\mathrm{argmin}$ for notation simplicity.  The flexibility of choosing any map $\tau$ means that we can break ties (i.e., choose $\hat\lambda$ or $\hat\lambda(y)$ or $\hat\lambda_{-i}(y)$ from the set of minimizers) by any desired rule, e.g., choosing a value $\lambda$ uniformly at random from the set of minimizers, or, applying a deterministic rule such as always choosing the smallest possible value of $\lambda$.

\subsection{Proofs of validity results: Theorems~\ref{thm:ModSel-validity} and~\ref{thm:LOO-validity}}
In this section we prove the marginal coverage guarantees for \texttt{ModSel-CP} and \texttt{ModSel-CP-LOO}. At a high level, the methods' validity follows from the fact that they are both special cases of the full conformal framework.
\begin{proof}[Proof of Theorem~\ref{thm:ModSel-validity}]
    By definition of the method~\eqref{defn:ModSel}, we can write
    \begin{multline*}\widehat{C}_{\texttt{ModSel-CP}} (X_{n+1}) = {}\\\left\{y \in \mathcal{Y}: S^{\hat{\lambda}(y)}(X_{n+1},y) \le \textnormal{Quantile}_{1-\alpha}\left(S^{\hat\lambda(y)}(X_1,Y_1),\dots,S^{\hat\lambda(y)}(X_n,Y_n),S^{\hat\lambda(y)}(X_{n+1},y)\right) \right\}.\end{multline*}
Therefore, the coverage event $Y_{n+1}\in \widehat{C}_{\texttt{ModSel-CP}} (X_{n+1})$ holds if and only if
\begin{equation}\label{eqn:coverage_event_ModSel}S^{\hat{\lambda}_{n+1}}(X_{n+1},Y_{n+1}) \le \textnormal{Quantile}_{1-\alpha}\left(S^{\hat\lambda_{n+1}}(X_1,Y_1),\dots,S^{\hat\lambda_{n+1}}(X_n,Y_n),S^{\hat\lambda_{n+1}}(X_{n+1},Y_{n+1})\right),\end{equation}
where we define
\[\hat\lambda_{n+1}= \hat\lambda(Y_{n+1}) .\]
Note that, by construction, $\hat\lambda_{n+1}$ is a symmetric function of the $n+1$ data points $\{(X_i,Y_i)\}_{i\in [n+1]}$. By exchangeability, then, the event given in~\eqref{eqn:coverage_event_ModSel} above has equal probability as
\[S^{\hat{\lambda}_{n+1}}(X_i,Y_i) \le \textnormal{Quantile}_{1-\alpha}\left(S^{\hat\lambda_{n+1}}(X_1,Y_1),\dots,S^{\hat\lambda_{n+1}}(X_n,Y_n),S^{\hat\lambda_{n+1}}(X_{n+1},Y_{n+1})\right),\]
for any $i\in[n+1]$. Therefore,
\begin{align*}
    &\mathbb{P}\{Y_{n+1}\in\widehat{C}_{\texttt{ModSel-CP}} (X_{n+1})\}\\
    &=\mathbb{P}\left\{S^{\hat{\lambda}_{n+1}}(X_{n+1},Y_{n+1}) \le \textnormal{Quantile}_{1-\alpha}\left(S^{\hat\lambda_{n+1}}(X_1,Y_1),\dots,S^{\hat\lambda_{n+1}}(X_{n+1},Y_{n+1})\right)\right\}\\
    &=\frac{1}{n+1}\sum_{i=1}^{n+1}\mathbb{P}\left\{S^{\hat{\lambda}_{n+1}}(X_i,Y_i) \le \textnormal{Quantile}_{1-\alpha}\left(S^{\hat\lambda_{n+1}}(X_1,Y_1),\dots,S^{\hat\lambda_{n+1}}(X_{n+1},Y_{n+1})\right)\right\}\\
    &=\mathbb{E}\left[\frac{1}{n+1}\sum_{i=1}^n\mathbf{1}\left\{S^{\hat{\lambda}_{n+1}}(X_i,Y_i) \le \textnormal{Quantile}_{1-\alpha}\left(S^{\hat\lambda_{n+1}}(X_1,Y_1),\dots,S^{\hat\lambda_{n+1}}(X_{n+1},Y_{n+1})\right)\right\}\right]\\
    &\geq \mathbb{E}[1-\alpha] = 1-\alpha,
\end{align*}
where the inequality holds by definition of the quantile.
\end{proof}

\begin{proof}[Proof of Theorem~\ref{thm:LOO-validity}]
    By definition of the method~\eqref{defn:LOO}, the coverage event \[ \left\{Y_{n+1}\in \widehat{C}_{\texttt{ModSel-CP-LOO}} (X_{n+1})\right\}\] holds if and only if
\[\mathcal{L}_{n+1}(\hat\lambda,S^{\hat\lambda}(X_{n+1},Y_{n+1})) \leq  \textnormal{Quantile}_{(1-\alpha)(1+1/n)}\left(\{\mathcal{L}_{n+1}(\hat\lambda_{-i}(Y_{n+1}),S^{\hat\lambda_{-i}(Y_{n+1})}(X_i,Y_i))\}_{i\in[n]}\right).\]
Now define
\[\hat\lambda^*_{-i} = \begin{cases} \hat\lambda_{-i}(Y_{n+1}), & i\in[n],\\ \hat\lambda, & i=n+1.\end{cases}\]
Then we can rewrite the above expression as
\[\mathcal{L}_{n+1}(\hat\lambda^*_{-(n+1)},S^{\hat\lambda^*_{-(n+1)}}(X_{n+1},Y_{n+1})) \leq  \textnormal{Quantile}_{(1-\alpha)(1+1/n)}\left(\{\mathcal{L}_{n+1}(\hat\lambda^*_{-i},S^{\hat\lambda^*_{-i}}(X_i,Y_i))\}_{i\in[n]}\right),\]
or equivalently (by standard calculations on the quantile),
\begin{equation}\label{eqn:coverage_event_LOO}\mathcal{L}_{n+1}(\hat\lambda^*_{-(n+1)},S^{\hat\lambda^*_{-(n+1)}}(X_{n+1},Y_{n+1})) \leq  \textnormal{Quantile}_{1-\alpha}\left(\{\mathcal{L}_{n+1}(\hat\lambda^*_{-i},S^{\hat\lambda^*_{-i}}(X_i,Y_i))\}_{i\in[n+1]}\right).\end{equation}
Now we check that this construction satisfies symmetry---that is, we want to see that $\hat\lambda^*_{-i}$ is constructed analogously for $i\in[n]$ as for $i=n+1$. By definition, for $i\in[n]$, we have
\[\hat\lambda^*_{-i} = \arg\min_{\lambda\in\Lambda}\mathcal{L}(\lambda,\hat{q}_{-i}(\lambda,Y_{n+1}) ; X_1,\dots,X_{n+1}),\]
where
\begin{multline*}\hat{q}_{-i}(\lambda,Y_{n+1}) = {}\\\textnormal{Quantile}_{(1-\alpha)(1+1/n)}\left(S^\lambda(X_1,Y_1),\dots,S^\lambda(X_{i-1},Y_{i-1}),S^\lambda(X_{i+1},Y_{i+1}),\dots,S^\lambda(X_{n+1},Y_{n+1})\right),\end{multline*}
while for $i=n+1$, we have
\[\hat\lambda^*_{-(n+1)} = \hat\lambda = \arg\min_{\lambda\in\Lambda}\mathcal{L}(\lambda,\hat{q}(\lambda);X_1,\dots,X_{n+1}) \]where as before,\[\hat{q}(\lambda) = \textnormal{Quantile}_{(1-\alpha)(1+1/n)}\left(S^{\lambda}(X_1,Y_1),\dots,S^{\lambda}(X_n,Y_n)\right).\]
In other words, this construction is defined in the same way for each $i$ (whether $i\in[n]$ or $i=n+1$), differing only by removing data point $(X_i,Y_i)$ from the optimization.
Then, by exchangeability of the data the event given in~\eqref{eqn:coverage_event_LOO} above has equal probability as
\[\mathcal{L}_{n+1}(\hat\lambda^*_{-j},S^{\hat\lambda^*_{-j}}(X_j,Y_j)) \leq  \textnormal{Quantile}_{1-\alpha}\left(\{\mathcal{L}_{n+1}(\hat\lambda^*_{-i},S^{\hat\lambda^*_{-i}}(X_i,Y_i))\}_{i\in[n+1]}\right),\]
for any $j\in[n+1]$. Therefore,
\begin{align*}
    &\mathbb{P}\{Y_{n+1}\in\widehat{C}_{\texttt{ModSel-CP-LOO}} (X_{n+1})\}\\
    &=\mathbb{P}\left\{\mathcal{L}_{n+1}(\hat\lambda^*_{-(n+1)},S^{\hat\lambda^*_{-(n+1)}}(X_{n+1},Y_{n+1})) \leq  \textnormal{Quantile}_{1-\alpha}\left(\{\mathcal{L}_{n+1}(\hat\lambda^*_{-i},S^{\hat\lambda^*_{-i}}(X_i,Y_i))\}_{i\in[n+1]}\right)\right\}\\
    &=\frac{1}{n+1}\sum_{j=1}^{n+1}\mathbb{P}\left\{\mathcal{L}_{n+1}(\hat\lambda^*_{-j},S^{\hat\lambda^*_{-j}}(X_j,Y_j)) \leq  \textnormal{Quantile}_{1-\alpha}\left(\{\mathcal{L}_{n+1}(\hat\lambda^*_{-i},S^{\hat\lambda^*_{-i}}(X_i,Y_i))\}_{i\in[n+1]}\right)\right\}\\
    &=\mathbb{E}\left[\frac{1}{n+1}\sum_{j=1}^n\mathbf{1}\left\{\mathcal{L}_{n+1}(\hat\lambda^*_{-j},S^{\hat\lambda^*_{-j}}(X_j,Y_j)) \leq  \textnormal{Quantile}_{1-\alpha}\left(\{\mathcal{L}_{n+1}(\hat\lambda^*_{-i},S^{\hat\lambda^*_{-i}}(X_i,Y_i))\}_{i\in[n+1]}\right)\right\}\right]\\
    &\geq \mathbb{E}[1-\alpha] = 1-\alpha,
\end{align*}
where the inequality again holds by definition of the quantile.    
\end{proof}

\subsection{Proof of implementation for \texttt{ModSel-CP}: Theorem~\ref{thm:ModSel_simplify}}

In this section, we prove the results of Theorem~\ref{thm:ModSel_simplify}, which offers a simple and efficient expression for computing the upper and lower bounds on the \texttt{ModSel-CP} prediction set. The proof is divided into three parts: first, establishing the upper bound~\eqref{eqn:ModSel_simple}, then the lower bound~\eqref{eqn:ModSel_simple_lower}, and finally showing their agreement up to a set of measure zero under the given conditions. Throughout the proof, we will write
\[\widehat{C}_+(X_{n+1}) = \bigcup_{\lambda\in\mathcal{M}} \left\{y \in \mathcal{Y}: \mathcal{L}_{n+1}(\lambda,S^\lambda(X_{n+1},y))\le\mathcal{L}_{n+1}({\hat\lambda},\hat{q}(\hat\lambda))\right\}\]
and
\[\widehat{C}_-(X_{n+1}) = \bigcup_{\lambda\in\mathcal{M}_-}\left\{y \in \mathcal{Y}: \mathcal{L}_{n+1}(\lambda,S^\lambda(X_{n+1},y))<\mathcal{L}_{n+1}({\hat\lambda},\hat{q}(\hat\lambda))\right\}.\]
These two sets are equivalent to the claimed upper bound~\eqref{eqn:ModSel_simple} and the claimed lower bound~\eqref{eqn:ModSel_simple_lower} on the prediction set $\widehat{C}_{\texttt{ModSel-CP}}(X_{n+1})$, respectively.

\paragraph{Proving the upper bound.}
By definition, we have
\begin{multline*}
    y\in\widehat{C}_{\textnormal{\texttt{ModSel-CP}}} (X_{n+1}) \Longleftrightarrow S^{\hat{\lambda}(y)}(X_{n+1},y) \le \hat{q} (\hat{\lambda}(y),y )   \\\Longrightarrow \mathcal{L}_{n+1}\big(\hat\lambda(y),S^{\hat{\lambda}(y)}(X_{n+1},y)\big) \le \mathcal{L}_{n+1}\big(\hat\lambda(y),\hat{q} (\hat{\lambda}(y),y )\big),
\end{multline*}
where the second step uses the monotonicity of $q\mapsto \mathcal{L}_{n+1}(\lambda,q)$.
Moreover, since 
\[\hat\lambda(y) \in\arg\min_{\lambda\in\Lambda}\mathcal{L}(\lambda,\hat{q}(\lambda,y)),\]
and $\hat{q}(\lambda,y)\leq\hat{q}(\lambda)$ by construction,
this implies 
\[\mathcal{L}_{n+1}\big(\hat\lambda(y),\hat{q} (\hat{\lambda}(y),y )\big)\leq \mathcal{L}_{n+1}\big(\hat\lambda,\hat{q} (\hat{\lambda},y )\big) \leq \mathcal{L}_{n+1}\big(\hat\lambda,\hat{q} (\hat{\lambda} )\big),\]
where the last step again uses the monotonicity of $q\mapsto \mathcal{L}_{n+1}(\lambda,q)$.
Therefore we have
\[y\in\widehat{C}_{\textnormal{\texttt{ModSel-CP}}} (X_{n+1}) \Longrightarrow \mathcal{L}_{n+1}\big(\hat\lambda(y),S^{\hat{\lambda}(y)}(X_{n+1},y)\big) \le \mathcal{L}_{n+1}\big(\hat\lambda,\hat{q} (\hat{\lambda} )\big).\]

To verify that $y\in\widehat{C}_+(X_{n+1})$, then, we only need to show that $\hat\lambda(y)\in\mathcal{M}$. First, we have
\[\hat{q}_-(\hat\lambda(y)) \leq \hat{q}(\hat\lambda(y),y) \Longrightarrow \mathcal{L}_{n+1}\big(\hat\lambda(y),\hat{q}_-(\hat\lambda(y))\big)\leq \mathcal{L}_{n+1}\big(\hat\lambda(y),\hat{q}(\hat\lambda(y),y)\big),\]
where the first bound holds since $\hat{q}_-(\lambda)\leq\hat{q}(\lambda,y)$ for all $\lambda$ by construction, and then we apply monotonicity of $q\mapsto\mathcal{L}_{n+1}(\lambda,q)$. Similarly, we have
\[\hat{q}(\hat\lambda)\geq \hat{q}(\hat\lambda,y) \Longrightarrow \mathcal{L}_{n+1}\big(\hat\lambda,\hat{q}(\hat\lambda)\big)\geq \mathcal{L}_{n+1}\big(\hat\lambda,\hat{q}(\hat\lambda,y)\big).\]
Combining everything, then, we have
\begin{multline*}
    \mathcal{L}_{n+1}\big(\hat\lambda(y),\hat{q}_-(\hat\lambda(y))\big)\leq \mathcal{L}_{n+1}\big(\hat\lambda(y),\hat{q}(\hat\lambda(y),y)\big)
\\= \min_{\lambda\in\Lambda}\mathcal{L}_{n+1}\big(\lambda,\hat{q}(\lambda,y)\big) \leq \mathcal{L}_{n+1}\big(\hat\lambda,\hat{q}(\hat\lambda,y)\big)\leq \mathcal{L}_{n+1}\big(\hat\lambda,\hat{q}(\hat\lambda)\big),
\end{multline*}
where the first and last steps come from the calculations above, while the equality follows from the definition of $\hat\lambda(y)$ as a minimizer. This verifies that $\hat\lambda(y)\in\mathcal{M}$.

\paragraph{Proving the lower bound.} 
Suppose that $y\in\widehat{C}_-(X_{n+1})$. By definition, we must have some $\lambda_{\ast}\in\Lambda$ such that
\[\mathcal{L}_{n+1}\big(\lambda_{\ast},\hat{q}_-(\lambda_{\ast})\big)<\mathcal{L}_{n+1}\big(\hat\lambda,\hat{q}(\hat\lambda)\big)\]
(i.e., $\lambda_{\ast}\in\mathcal{M}_-$) and
\[\mathcal{L}_{n+1}\big(\lambda_{\ast},S^{\lambda_{\ast}}(X_{n+1},y)\big)<\mathcal{L}_{n+1}\big(\hat\lambda,\hat{q}(\hat\lambda)\big).\]
Since by construction of the quantiles, we must have $\hat{q}(\lambda_{\ast},y) \leq \max\{ \hat{q}_-(\lambda_{\ast}), S^{\lambda_{\ast}}(X_{n+1},y)\}$, 
this therefore implies 
 \[\mathcal{L}_{n+1}\big(\lambda_{\ast},\hat{q}(\lambda_{\ast},y)\big)<\mathcal{L}_{n+1}\big(\hat\lambda,\hat{q}(\hat\lambda)\big).\]
We then have
\begin{multline*}\mathcal{L}_{n+1}\big(\hat\lambda(y),\hat{q}(\hat\lambda(y),y)\big) = \min_{\lambda\in\Lambda}\mathcal{L}_{n+1}\big(\lambda,\hat{q}(\lambda,y)\big) \leq \mathcal{L}_{n+1}\big(\lambda_{\ast},\hat{q}(\lambda_{\ast},y)\big) \\ < \mathcal{L}_{n+1}\big(\hat\lambda,\hat{q}(\hat\lambda)\big) = \min_{\lambda\in\Lambda}\mathcal{L}_{n+1}\big(\lambda,\hat{q}(\lambda)\big) \leq \mathcal{L}_{n+1}\big(\hat\lambda(y),\hat{q}(\hat\lambda(y))\big).\end{multline*}
Monotonicity of $q\mapsto\mathcal{L}_{n+1}(\lambda,q)$ then implies
\[\hat{q}(\hat\lambda(y),y) < \hat{q}(\hat\lambda(y)).\]
By construction of these quantiles, then, we must have
\[S^{\hat\lambda(y)}(X_{n+1},y) \leq \hat{q}(\hat\lambda(y),y),\]
i.e., $y\in\widehat{C}_{\texttt{ModSel-CP}}(X_{n+1})$, as desired.

\paragraph{Proving equality up to a set of measure zero under additional conditions.} 
For this last part of the proof, our aim is to prove that, under the additional conditions, the set difference $\widehat{C}_+(X_{n+1})\backslash \widehat{C}_-(X_{n+1})$ has measure zero.
First, since $\mathcal{M}\supseteq\mathcal{M}_-$ by construction, therefore, the set difference can be bounded as
\begin{multline*}\widehat{C}_+(X_{n+1}) \backslash \widehat{C}_-(X_{n+1})\subseteq \underbrace{\left[\bigcup_{\lambda\in\mathcal{M}_-} \left\{y \in \mathcal{Y}: \mathcal{L}_{n+1}(\lambda,S^\lambda(X_{n+1},y))=\mathcal{L}_{n+1}({\hat\lambda},\hat{q}(\hat\lambda))\right\}\right]}_{\textnormal{Term 1}} \\{}\bigcup\underbrace{\left[\bigcup_{\lambda\in\mathcal{M}\backslash \mathcal{M}_-} \left\{y \in \mathcal{Y}: \mathcal{L}_{n+1}(\lambda,S^\lambda(X_{n+1},y))\leq\mathcal{L}_{n+1}({\hat\lambda},\hat{q}(\hat\lambda))\right\}\right]}_{\textnormal{Term 2}}.\end{multline*}
Now we apply the additional conditions. Since we assume that $\{y\in\mathcal{Y}: \mathcal{L}_{n+1}(\lambda, S^{\lambda}(X_{n+1},y)) = q\}$ has measure zero (for any $q\in\mathbb{R}$), this implies that, for every $\lambda\in\Lambda$, the set
\[\left\{y \in \mathcal{Y}: \mathcal{L}_{n+1}(\lambda,S^\lambda(X_{n+1},y))=\mathcal{L}_{n+1}({\hat\lambda},\hat{q}(\hat\lambda))\right\}\]
has measure zero. In particular, this means that Term 1 has measure zero.

Next, we turn to Term 2. For this term, we will simply prove that $\mathcal{M}=\mathcal{M}_-$ (i.e., Term 2 is the empty set). Indeed, for any $\lambda$, if $\lambda\in \mathcal{M}\backslash \mathcal{M}_-$, then we must have
\[\mathcal{L}_{n+1}(\lambda, \hat{q}_{-}(\lambda)) = \mathcal{L}_{n+1}(\hat{\lambda},\hat{q}(\hat{\lambda})),\]
which cannot hold under our additional conditions.

\subsection{Proof of method comparison results: Proposition~\ref{prop:containYK-nonadj}}
In this section, we prove the result of Proposition~\ref{prop:containYK-nonadj}, which compares the \texttt{ModSel-CP}  and \texttt{ModSel-CP-LOO} prediction sets with that produced by \texttt{YK-baseline}.

First consider the \texttt{ModSel-CP} method. Suppose $y
\not\in\widehat{C}_{\texttt{ModSel-CP}}(X_{n+1})$. Then by definition,
\[S^{\hat\lambda(y)}(X_{n+1},y) > \hat{q}(\hat\lambda(y),y), \]
and so
\begin{align}
    \notag\hat{q}(\hat\lambda(y),y)
    &=\textnormal{Quantile}_{1-\alpha}\left(S^{\hat\lambda(y)}(X_1,Y_1),\dots,S^{\hat\lambda(y)}(X_n,Y_n),S^{\hat\lambda(y)}(X_{n+1},y)\right)\\
    \notag&=\textnormal{Quantile}_{(1-\alpha)(1+1/n)}\left(S^{\hat\lambda(y)}(X_1,Y_1),\dots,S^{\hat\lambda(y)}(X_n,Y_n)\right)\\
   \label{eqn:exclude_a_large_value} &=\hat{q}(\hat\lambda(y)).
\end{align}
On the other hand, we have
\begin{align*}
    \hat{q}(\hat\lambda,y)
    &=\textnormal{Quantile}_{1-\alpha}\left(S^{\hat\lambda}(X_1,Y_1),\dots,S^{\hat\lambda}(X_n,Y_n),S^{\hat\lambda}(X_{n+1},y)\right)\\
    &\leq\textnormal{Quantile}_{(1-\alpha)(1+1/n)}\left(S^{\hat\lambda}(X_1,Y_1),\dots,S^{\hat\lambda}(X_n,Y_n)\right)\\
    &=\hat{q}(\hat\lambda).
\end{align*}
By optimality of $\hat\lambda(y)$, we must therefore have
\[\mathcal{L}_{n+1}(\hat{\lambda}(y), \hat{q}(\hat\lambda(y)))= \mathcal{L}_{n+1}(\hat{\lambda}(y),  \hat{q}(\hat\lambda(y),y))
\leq \mathcal{L}_{n+1}(\hat{\lambda},  \hat{q}(\hat\lambda,y)) \leq \mathcal{L}_{n+1}(\hat{\lambda},  \hat{q}(\hat\lambda)).\]
But by our assumption of uniqueness of $\hat\lambda$, this implies that $\hat\lambda = \hat\lambda(y)$.
We then have
\[S^{\hat\lambda}(X_{n+1},y) =S^{\hat\lambda(y)}(X_{n+1},y) > \hat{q}(\hat\lambda(y),y) = \hat{q}(\hat\lambda,y),\]
and thus, by a similar calculation as in~\eqref{eqn:exclude_a_large_value} above (now with $\hat\lambda$ in place of $\hat\lambda(y)$), \[\hat{q}(\hat\lambda,y) = \hat{q}(\hat\lambda).\]
We have therefore shown that $S^{\hat\lambda}(X_{n+1},y) >\hat{q}(\hat\lambda)$, i.e., $y\not\in\widehat{C}_{\texttt{YK-baseline}}(X_{n+1})$.

Next we turn to \texttt{ModSel-CP-LOO}. 
Fix any $y\in\widehat{C}_{\texttt{YK-baseline}}(X_{n+1})$.
First suppose that, for some $i\in[n]$, we have
\begin{equation}\label{eqn:LOO_score_compare_step}
\mathcal{L}_{n+1}(\hat\lambda_{-i}(y),S^{\hat\lambda_{-i}(y)}(X_i,Y_i)) < \mathcal{L}_{n+1}(\hat\lambda,S^{\hat\lambda}(X_{n+1},y)), \   S^{\hat\lambda}(X_{n+1},y) 
\leq S^{\hat\lambda}(X_i,Y_i) .
\end{equation}
Then in particular, we must have $\hat\lambda_{-i}(y)\neq\hat\lambda$, and so by uniqueness of $\hat\lambda$, 
\begin{multline}\label{eqn:quantile_is_above}\mathcal{L}_{n+1}(\hat\lambda,\hat{q}(\hat\lambda)) < \mathcal{L}_{n+1}(\hat\lambda_{-i}(y),\hat{q}(\hat\lambda_{-i}(y))) \\ = \mathcal{L}_{n+1}\left(\hat\lambda_{-i}(y), \textnormal{Quantile}_{(1-\alpha)(1+1/n)}\left( S^{\hat\lambda_{-i}(y)}(X_1,Y_1),\dots,S^{\hat\lambda_{-i}(y)}(X_n,Y_n)\right)\right) \\=\textnormal{Quantile}_{(1-\alpha)(1+1/n)}\left(\{ \mathcal{L}_{n+1}(\hat\lambda_{-i}(y),S^{\hat\lambda_{-i}(y)}(X_j,Y_j))\}_{j\in[n]}\right).\end{multline}
On the other hand, by optimality of $\hat\lambda_{-i}(y)$,
\begin{multline}\label{eqn:quantile_is_below}\textnormal{Quantile}_{(1-\alpha)(1+1/n)}\left( \{\mathcal{L}_{n+1}(\hat\lambda_{-i}(y),S^{\hat\lambda_{-i}(y)}(X_j,Y_j))\}_{j\in[n]\backslash\{i\}}\cup \{\mathcal{L}_{n+1}(\hat\lambda_{-i}(y),S^{\hat\lambda_{-i}(y)}(X_{n+1},y)\}\right) \\= \mathcal{L}_{n+1}(\hat\lambda_{-i}(y),\hat{q}_{-i}(\hat\lambda_{-i}(y),y)) \leq \mathcal{L}_{n+1}(\hat\lambda,\hat{q}_{-i}(\hat\lambda,y))\\
 = \textnormal{Quantile}_{(1-\alpha)(1+1/n)}\left( \{\mathcal{L}_{n+1}(\hat\lambda,S^{\hat\lambda}(X_j,Y_j))\}_{j\in[n]\backslash\{i\}}\cup \{\mathcal{L}_{n+1}(\hat\lambda,S^{\hat\lambda}(X_{n+1},y))\}\right)\\
\leq \textnormal{Quantile}_{(1-\alpha)(1+1/n)}\left( \{\mathcal{L}_{n+1}(\hat\lambda,S^{\hat\lambda}(X_j,Y_j))\}_{j\in[n]}\right) = \mathcal{L}_{n+1}(\hat\lambda,\hat{q}(\hat\lambda)),\end{multline}
where the last inequality holds since we have assumed $S^{\hat\lambda}(X_{n+1},y) \leq S^{\hat\lambda}(X_i,Y_i)$. Comparing the two calculations~\eqref{eqn:quantile_is_above} and~\eqref{eqn:quantile_is_below}, we see that we must have 
\[\mathcal{L}_{n+1}(\hat\lambda_{-i}(y),S^{\hat\lambda_{-i}(y)}(X_i,Y_i))  > \mathcal{L}_{n+1}(\hat\lambda,\hat{q}(\hat\lambda)).\]
This is a contradiction with~\eqref{eqn:LOO_score_compare_step}, since we have assumed $y\in\widehat{C}_{\texttt{YK-baseline}}(X_{n+1})$, i.e., $S^{\hat\lambda}(X_{n+1},y) \leq \hat{q}(\hat\lambda)$ and so $\mathcal{L}_{n+1}(\hat\lambda,S^{\hat\lambda}(X_{n+1},y)) \leq \mathcal{L}_{n+1}(\hat\lambda,\hat{q}(\hat\lambda))$.

Therefore, for all $i\in[n]$, we must have
\begin{equation}\label{eqn:LOO_score_compare_step_imply}
S^{\hat\lambda}(X_{n+1},y) \leq S^{\hat\lambda}(X_i,Y_i) \ \Longrightarrow \  \mathcal{L}_{n+1}(\hat\lambda,S^{\hat\lambda}(X_{n+1},y))\leq \mathcal{L}_{n+1}(\hat\lambda_{-i}(y),S^{\hat\lambda_{-i}(y)}(X_i,Y_i)).
\end{equation}
Since $y\in\widehat{C}_{\texttt{YK-baseline}}(X_{n+1})$, we have
\[S^{\hat\lambda}(X_{n+1},y) \leq \hat{q}(\hat\lambda) = \textnormal{Quantile}_{(1-\alpha)(1+1/n)}(\{S^{\hat\lambda}(X_i,Y_i)\}_{i\in[n]}),\]
or equivalently,
\[\sum_{i=1}^n \mathbf{1}\{S^{\hat\lambda}(X_{n+1},y) \leq S^{\hat\lambda}(X_i,Y_i)\} \geq \alpha(n+1).\]
Therefore, applying~\eqref{eqn:LOO_score_compare_step_imply}, we have
\[\sum_{i=1}^n \mathbf{1}\{\mathcal{L}_{n+1}(\hat\lambda,S^{\hat\lambda}(X_{n+1},y)) \leq \mathcal{L}_{n+1}(\hat\lambda_{-i}(y),S^{\hat\lambda_{-i}(y)}(X_i,Y_i))\} \geq \alpha(n+1),\]
or equivalently, 
\[\mathcal{L}_{n+1}(\hat\lambda,S^{\hat\lambda}(X_{n+1},y)) \leq \textnormal{Quantile}_{(1-\alpha)(1+1/n)}(\{\mathcal{L}_{n+1}(\hat\lambda_{-i}(y),S^{\hat\lambda_{-i}(y)}(X_i,Y_i))\}_{i\in[n]}),\]
i.e., $y\in\widehat{C}_{\texttt{ModSel-CP-LOO}}(X_{n+1})$.

\section{Implementation for \texttt{ModSel-CP-LOO}}\label{sec:appendix_implement_LOO}
In this section, we derive an algorithm for implementing the \texttt{ModSel-CP-LOO} method. Unlike the straightforward implementation of \texttt{ModSel-CP} shown in Theorem~\ref{thm:ModSel_simplify}, for this method the procedure is more complex, but can still be implemented efficiently for many common choices such as the residual score.

At a high level, our strategy is the following. We will partition $\mathcal{Y}$ into many regions, $\mathcal{Y} = \cup_{j \in \mathcal{J}} \mathcal{B}_j$, such that within each region, for all $i\in[n]$, the selected model $\hat\lambda_{-i}(y)$ is constant:
\[\hat\lambda_{-i}(y) = \lambda_{i,j}\textnormal{ for all $i\in[n]$, $j\in \mathcal{J}$, $y\in\mathcal{B}_j$.}\]
 (In practice, $\mathcal{J}$ will typically be a finite collection of intervals.)
Examining the definition~\eqref{defn:LOO} of the prediction set, we then see that it can be rewritten as
\begin{multline*}\widehat{C}_{\texttt{ModSel-CP-LOO}}(X_{n+1}) =\\ \cup_{j \in \mathcal{J}}\mathcal{B}_j \cap \left\{ y\in\mathcal{Y} : \mathcal{L}_{n+1}(\hat\lambda,S^{\hat\lambda}(X_{n+1},y))\leq \textnormal{Quantile}_{(1-\alpha)(1+1/n)}\left(\big\{\mathcal{L}_{n+1}(\lambda_{i,j},S^{\lambda_{i,j}}(X_i,Y_i))\big\}_{i\in[n]}\right)\right\}.\end{multline*}
While it might appear that this has not simplified the calculation, in cases such as the residual score, these sets are more tractable to work with---specifically, for the residual score and with loss equal to interval length,
\begin{multline}\label{eqn:calculate_LOO_interval}\widehat{C}_{\texttt{ModSel-CP-LOO}}(X_{n+1}) ={}\\ \cup_{j \in \mathcal{J}} \mathcal{B}_j \cap \left\{ y\in\mathcal{Y} : |y-f_{\hat\lambda}(X_{n+1})| \leq  \textnormal{Quantile}_{(1-\alpha)(1+1/n)}\left(\big\{|Y_i - f_{\lambda_{i,j}}(X_i)|\big\}_{i\in[n]}\right)\right\}.\end{multline}

The remaining question, then, is how we may calculate these sets $\mathcal{B}_j$. For the case of a finite $\mathcal{Y}$ (i.e., a categorical response), this is simply a question of enumerating all possible values $y$. In a regression setting where $\mathcal{Y}=\mathbb{R}$, however, computing this efficiently will require the following propositions.
First we give a result about the \texttt{ModSel-CP-LOO} method in a general setting---specifically, about the possible values that the selected models $\hat\lambda_{-i}(y)$ may take.
\begin{proposition}\label{prop:Mcal_i}
    For each $i\in [n]$, define \begin{equation}\label{defn:calM_i}
        \mathcal{M}_i = \left\{\lambda\in\Lambda : l_i(\lambda)\leq \min_{\lambda'\in\Lambda}u_i(\lambda')\right\},
    \end{equation}
    for
    \begin{equation*}
        l_{i}(\lambda) = \begin{cases}
            \mathcal{L}_{n+1}(\lambda,\hat{q}(\lambda)), \textrm{ if } \mathcal{L}_{n+1}(\lambda,S^{\lambda}(X_{i},Y_{i})) < \mathcal{L}_{n+1}(\lambda,\hat{q}(\lambda))\\
            \mathcal{L}_{n+1}(\lambda,\hat{q}_{-}(\lambda)), \textrm{ if } \mathcal{L}_{n+1}(\lambda,S^{\lambda}(X_{i},Y_{i})) \ge \mathcal{L}_{n+1}(\lambda,\hat{q}(\lambda))
        \end{cases},
    \end{equation*}
    \begin{equation*}
        u_{i}(\lambda) = \begin{cases}
            \mathcal{L}_{n+1}(\lambda,\hat{q}_{+}(\lambda)), \textrm{ if } \mathcal{L}_{n+1}(\lambda,S^{\lambda}(X_{i},Y_{i})) \le \mathcal{L}_{n+1}(\lambda,\hat{q}(\lambda))\\
            \mathcal{L}_{n+1}(\lambda,\hat{q}(\lambda)), \textrm{ if } \mathcal{L}_{n+1}(\lambda,S^{\lambda}(X_{i},Y_{i})) > \mathcal{L}_{n+1}(\lambda,\hat{q}(\lambda))
        \end{cases},
    \end{equation*}
    where $ \hat{q}(\lambda),\hat{q}_-(\lambda)$ are defined as before, while \[\hat{q}_+(\lambda) = \textnormal{Quantile}_{(1-\alpha)(1+1/n)+1/n}(S^\lambda(X_1,Y_1),\dots,S^\lambda(X_n,Y_n)).\] Then, for all $i\in[n]$ and all $y\in\mathcal{Y}$,
    \[\hat\lambda_{-i}(y)\in\mathcal{M}_i.\]
\end{proposition}
Next we use these sets $\mathcal{M}_i$ to help compute the breakpoints for the special case $\mathcal{Y}=\mathbb{R}$, as is the focus of this section.
\begin{proposition}\label{prop:B_j}
    Under the notation and definitions above, in the case of $\mathcal{Y}=\mathbb{R}$,  assume also that $y\mapsto S^\lambda(x,y)$ is continuous, and $q\mapsto \mathcal{L}_{n+1}(\lambda,q)$ is continuous (e.g., as for the residual score, and for loss given by interval length). Then the partition $\mathbb{R}=\cup_{j \in \mathcal{J}}\mathcal{B}_j$ can be defined with the following set of breakpoints between intervals:
    \begin{equation}\label{eqn:find_equalities_for_B_j}\bigcup_{i\in[n]} \left\{y\in\mathbb{R} : \mathcal{L}_{n+1}(\lambda,\hat{q}_{-i}(\lambda,y)) = \mathcal{L}_{n+1}(\lambda',\hat{q}_{-i}(\lambda',y))\textnormal{ for some $\lambda\neq\lambda'\in\mathcal{M}_i$}\right\},\end{equation}
    where $\mathcal{M}_i$ is defined as in Proposition~\ref{prop:Mcal_i}.
\end{proposition}
\begin{proof}[Proof of Proposition~\ref{prop:Mcal_i}]
Fix any $\lambda_0\in\Lambda$ and 
suppose $\lambda_0 = \hat\lambda_{-i}(y)$. Then by definition of $\hat\lambda_{-i}(y)$, for any $\lambda\in\Lambda$ we have
\begin{align*}
\mathcal{L}_{n+1}(\lambda_0,\hat{q}_{-i}(\lambda_0,y))
&\leq \mathcal{L}_{n+1}(\lambda,\hat{q}_{-i}(\lambda,y))\\
&=\mathcal{L}_{n+1}\left(\lambda,\textnormal{Quantile}_{(1-\alpha)(1+1/n)}\left(\{S^{\lambda}(X_j,Y_j)\}_{j\in[n]\backslash\{i\}}, S^{\lambda}(X_{n+1},y)\right)\right)\\
&\leq\mathcal{L}_{n+1}\left(\lambda,\textnormal{Quantile}_{(1-\alpha)(1+1/n)+1/n}\left(\{S^{\lambda}(X_j,Y_j)\}_{j\in[n]}\right)\right)\\
&=\mathcal{L}_{n+1}\left(\lambda,\hat{q}_+(\lambda)\right).
\end{align*}
Moreover, if $\mathcal{L}_{n+1}(\lambda,S^\lambda(X_i,Y_i)) > \mathcal{L}_{n+1}(\lambda,\hat{q}(\lambda))$, then we have a tighter bound,
\begin{align*}
\mathcal{L}_{n+1}(\lambda_0,\hat{q}_{-i}(\lambda_0,y))
&\leq \mathcal{L}_{n+1}(\lambda,\hat{q}_{-i}(\lambda,y))\\
&=\mathcal{L}_{n+1}\left(\lambda,\textnormal{Quantile}_{(1-\alpha)(1+1/n)}\left(\{S^{\lambda}(X_j,Y_j)\}_{j\in[n]\backslash\{i\}}, S^{\lambda}(X_{n+1},y)\right)\right)\\
&\leq\mathcal{L}_{n+1}\left(\lambda,\textnormal{Quantile}_{(1-\alpha)(1+1/n)}\left(\{S^{\lambda}(X_j,Y_j)\}_{j\in[n]}\right)\right)\\
&=\mathcal{L}_{n+1}\left(\lambda,\hat{q}(\lambda)\right),
\end{align*}
where the second inequality holds specifically because of our assumption on $\lambda$. Combining these cases, then, we see that
$\mathcal{L}_{n+1}(\lambda_0,\hat{q}_{-i}(\lambda_0,y))\leq u_i(\lambda)$
holds for all $\lambda\in\Lambda$, and therefore,
\[\mathcal{L}_{n+1}(\lambda_0,\hat{q}_{-i}(\lambda_0,y))\leq \min_{\lambda\in\Lambda}u_i(\lambda).\]
On the other hand,
\begin{align*}
    \mathcal{L}_{n+1}(\lambda_0,\hat{q}_{-i}(\lambda_0,y))
    &=\mathcal{L}_{n+1}\left(\lambda_0,\textnormal{Quantile}_{(1-\alpha)(1+1/n)}\left(\{S^{\lambda_0}(X_j,Y_j)\}_{j\in[n]\backslash\{i\}}\cup S^{\lambda_0}(X_{n+1},y)\right)\right)\\
    &\geq \mathcal{L}_{n+1}\left(\lambda_0,\textnormal{Quantile}_{(1-\alpha)(1+1/n)-1/n}\left(\{S^{\lambda_0}(X_j,Y_j)\}_{j\in[n]}\right)\right)\\
    &=\mathcal{L}_{n+1}(\lambda_0,\hat{q}_-(\lambda_0)).
\end{align*}
And moreover, if $\mathcal{L}_{n+1}(\lambda_0,S^\lambda(X_i,Y_i)) < \mathcal{L}_{n+1}(\lambda_0,\hat{q}(\lambda_0))$, then we have a tighter bound,
\begin{align*}
    \mathcal{L}_{n+1}(\lambda_0,\hat{q}_{-i}(\lambda_0,y))
    &=\mathcal{L}_{n+1}\left(\lambda_0,\textnormal{Quantile}_{(1-\alpha)(1+1/n)}\left(\{S^{\lambda_0}(X_j,Y_j)\}_{j\in[n]\backslash\{i\}}\cup S^{\lambda_0}(X_{n+1},y)\right)\right)\\
    &\geq \mathcal{L}_{n+1}\left(\lambda_0,\textnormal{Quantile}_{(1-\alpha)(1+1/n)}\left(\{S^{\lambda_0}(X_j,Y_j)\}_{j\in[n]}\right)\right)\\
    &=\mathcal{L}_{n+1}(\lambda_0,\hat{q}(\lambda_0)).
\end{align*}
where the inequality holds specifically because of our assumption on $\lambda$. Combining these cases, then, we see that
$\mathcal{L}_{n+1}(\lambda_0,\hat{q}_{-i}(\lambda_0,y))\geq l_i(\lambda_0)$.

Combining all these calculations, we have verified that $l_i(\lambda_0) \leq \min_{\lambda\in\Lambda}u_i(\lambda)$, i.e., $\lambda_0\in\mathcal{M}_i$.
\end{proof}
\begin{proof}[Proof of Proposition~\ref{prop:B_j}]
It is sufficient to prove the following holds for any single $i\in[n]$: if for some $y_0<y_1$ we have $\hat\lambda_{-i}(y_0)\neq\hat\lambda_{-i}(y_1)$, then there must exist some $y\in[y_0,y_1]$, and some $\lambda\neq\lambda'\in\mathcal{M}_i$, such that $\mathcal{L}_{n+1}(\lambda,\hat{q}_{-i}(\lambda,y)) = \mathcal{L}_{n+1}(\lambda',\hat{q}_{-i}(\lambda',y))$. (This will imply that $y$ lies in the set defined in~\eqref{eqn:find_equalities_for_B_j}, i.e., the set of breakpoints between the intervals $\mathcal{B}_j$.)

Let $\lambda_0 = \hat\lambda_{-i}(y_0)$ and
define $y = \sup\{y' \geq y_0 : \hat\lambda_{-i}(y') = \lambda_0\}$. Then, for any $\epsilon>0$, we have $y'\in[y,y+\epsilon)$ with $\hat\lambda_{-i}(y')\neq \lambda_0$. Then we can choose $\lambda_1\in\Lambda\backslash\{0\}$ as any value such that $\inf\{y'\geq y_0 : \hat\lambda_{-i}(y') = \lambda_1\} = y$.
We will now verify that $\lambda_0,\lambda_1\in\mathcal{M}_i$, and, that $\mathcal{L}_{n+1}(\lambda_0,\hat{q}_{-i}(\lambda_0,y)) = \mathcal{L}_{n+1}(\lambda_1,\hat{q}_{-i}(\lambda_1,y))$.

First, to see that $\lambda_0\in\mathcal{M}_i$, observe that {$\lambda_0 = \hat\lambda_{-i}(y_0)$ and so we must have $\lambda_0\in\mathcal{M}_i$ by the work above.
Similarly,} by definition of $\lambda_1$, we must have some $y'\geq y$ with $\hat\lambda_{-i}(y') = \lambda_1$, so by a similar argument, we have $\lambda_1\in\mathcal{M}_i$.

To complete the proof, we will now verify $\mathcal{L}_{n+1}(\lambda_0,\hat{q}_{-i}(\lambda_0,y)) = \mathcal{L}_{n+1}(\lambda_1,\hat{q}_{-i}(\lambda_1,y))$. For any $\epsilon>0$, we have some $y''_{\epsilon}\in(y-\epsilon,y]$ and some $y'_{\epsilon}\in[y,y+\epsilon)$ such that $\hat\lambda_{-i}(y''_{\epsilon}) = \lambda_0$ and $\hat\lambda_{-i}(y'_{\epsilon}) = \lambda_1$. Therefore, by optimality of these selected models,
\[\mathcal{L}_{n+1}(\lambda_0,\hat{q}_{-i}(\lambda_0,y''_{\epsilon})) \leq \mathcal{L}_{n+1}(\lambda_1,\hat{q}_{-i}(\lambda_1,y''_{\epsilon})) \]
and
\[\mathcal{L}_{n+1}(\lambda_0,\hat{q}_{-i}(\lambda_0,y'_{\epsilon})) \geq \mathcal{L}_{n+1}(\lambda_1,\hat{q}_{-i}(\lambda_1,y'_{\epsilon})). \]
Moreover, since we have assumed $y\mapsto S^\lambda(x,y)$ is continuous, since $\epsilon$ is arbitrarily small we have
\[\lim_{\epsilon\to 0}\hat{q}_{-i}(\lambda_k,y'_{\epsilon}) = \hat{q}_{-i}(\lambda_k,y)\textnormal{ for each $k=0,1$}.\]
Since we have assumed $\mathcal{L}_{n+1}(\lambda,q)$ is continuous in $q$, therefore,
\[\lim_{\epsilon\to 0}\mathcal{L}_{n+1}(\lambda_k,\hat{q}_{-i}(\lambda_k,y'_{\epsilon})) = \lim_{\epsilon\to 0}\mathcal{L}_{n+1}(\lambda_k,\hat{q}_{-i}(\lambda_k,y''_{\epsilon})) = \mathcal{L}_{n+1}(\lambda_k,\hat{q}_{-i}(\lambda_k,y))\textnormal{ for each $k=0,1$}.\]
Therefore, we have
\begin{equation*}
    \mathcal{L}_{n+1}(\lambda_0,\hat{q}_{-i}(\lambda_0,y)= \lim_{\epsilon\to 0}\mathcal{L}_{n+1}(\lambda_0,\hat{q}_{-i}(\lambda_0,y''_{\epsilon})) \leq  \lim_{\epsilon\to 0}\mathcal{L}_{n+1}(\lambda_1,\hat{q}_{-i}(\lambda_1,y''_{\epsilon})) = \mathcal{L}_{n+1}(\lambda_1,\hat{q}_{-i}(\lambda_1,y))
\end{equation*}
and
\begin{equation*}
    \mathcal{L}_{n+1}(\lambda_0,\hat{q}_{-i}(\lambda_0,y)) = \lim_{\epsilon\to 0}\mathcal{L}_{n+1}(\lambda_0,\hat{q}_{-i}(\lambda_0,y'_{\epsilon})) \geq \lim_{\epsilon\to 0}\mathcal{L}_{n+1}(\lambda_1,\hat{q}_{-i}(\lambda_1,y'_{\epsilon}))=\mathcal{L}_{n+1}(\lambda_1,\hat{q}_{-i}(\lambda_1,y)),
\end{equation*}
which proves that $\mathcal{L}_{n+1}(\lambda_0,\hat{q}_{-i}(\lambda_0,y))=\mathcal{L}_{n+1}(\lambda_1,\hat{q}_{-i}(\lambda_1,y))$, as desired.

\end{proof}

\section{Proofs for theoretical results in Section~\ref{sec: efficiency_res}}

\subsection{Proof for asymptotic optimality results (Theorem~\ref{thm:efficiency_res_} and extension)}\label{app:efficiency}

In this section, we will prove Theorem~\ref{thm:efficiency_res_}, which establishes that \texttt{ModSel-CP} and \texttt{ModSel-CP-LOO} achieve asymptotically optimal prediction set size under certain conditions. While Theorem~\ref{thm:efficiency_res_} works in the special case of the residual score, $S^\lambda(x,y) = |y-f_\lambda(x)|$, here we will present a more general result showing that these optimality guarantees can hold more broadly.

We will begin by generalizing some of the notation and definitions to allow for results beyond the setting of a residual score. First, we specify our goal. The (uncorrected) \texttt{YK-baseline} method returns the prediction set
 \[\widehat{C}_{\texttt{YK-baseline}}(X_{n+1}) = C^{\hat\lambda}_{\hat{q}(\hat\lambda)}(X_{n+1}) = \left\{y \in \mathcal{Y} : S^{\hat\lambda}(X_{n+1},y)\leq \hat{q}(\hat\lambda)\right\}.\]
 In a setting where $\mathcal{L}_{n+1}(\lambda,q)$ is strictly increasing in $q$, this is equivalent to
  \[\widehat{C}_{\texttt{YK-baseline}}(X_{n+1}) = \left\{y\in \mathcal{Y} : \mathcal{L}_{n+1}(\hat\lambda,S^{\hat\lambda}(X_{n+1},y))\leq \mathcal{L}_{n+1}(\hat\lambda,\hat{q}(\hat\lambda))\right\}.\]
Here,
 we will aim to show that, under certain conditions, \texttt{ModSel-CP} and \texttt{ModSel-CP-LOO} return a set that is a slight inflation: namely, that the \texttt{ModSel-CP} and \texttt{ModSel-CP-LOO} lie inside a set of the form
\[\left\{y \in \mathcal{Y} : \mathcal{L}_{n+1}(\hat\lambda,S^{\hat\lambda}(X_{n+1},y))\leq \mathcal{L}_{n+1}(\hat\lambda,\hat{q}(\hat\lambda)) + \kappa\right\},\]
for a small value $\kappa$.

\subsubsection{Formal results} First we make some assumptions and definitions. First, define the set of nearly-optimal models,
\[\Lambda_*(\Delta,\Delta') = \left\{\lambda\in\Lambda : \mathbb{E}\left[\mathcal{L}_{n+1}(\lambda,Q_{1-\alpha-\Delta}(\lambda))\right]\leq  \min_{\lambda'\in\Lambda}\mathbb{E}\left[\mathcal{L}_{n+1}(\lambda',Q_{1-\alpha+\Delta}(\lambda'))\right] + 2\Delta'\right\}.\]
Define also
\[\kappa(\Delta,\Delta') = \max_{\lambda,\lambda'\in\Lambda_*(\Delta,\Delta')}\sup_{y\in\mathcal{Y}}\left| \mathcal{L}_{n+1}(\lambda,S^\lambda(X_{n+1},y)) -  \mathcal{L}_{n+1}({\lambda'},S^{\lambda'}(X_{n+1},y)) \right|\]
(which we will use in the results for \texttt{ModSel-CP}), and
\[\kappa_{\texttt{LOO}}(\Delta,\Delta') = \max_{\lambda,\lambda'\in\Lambda_*(\Delta,\Delta')}\max_{i=1,\dots,n}\left| \mathcal{L}_{n+1}(\lambda,S^\lambda(X_i,Y_i)) -  \mathcal{L}_{n+1}({\lambda'},S^{\lambda'}(X_i,Y_i)) \right|\]
(which we will use for \texttt{ModSel-CP-LOO}).

\begin{theorem}\label{thm:efficiency_general_}
Assume $(X_i,Y_i)\stackrel{\rm i.i.d.}{\sim}P$, and that $S^\lambda(X,Y)$ has a continuous distribution under $(X,Y)\sim P$, for each $\lambda\in\Lambda$. Assume that, for any fixed values $\{t_\lambda\}_{\lambda\in\Lambda}$, it holds that \begin{equation}\label{eqn:asm_loss_concentrates}\mathbb{P}\left\{\max_{\lambda\in\Lambda}\left|\mathcal{L}_{n+1}(\lambda,t_\lambda) - \mathbb{E}\left[\mathcal{L}_{n+1}(\lambda,t_\lambda)\right]\right|\leq \Delta_n'\right\}\geq 1-\delta'.\end{equation}
Define  
\[\Delta_n = 2\mathcal{R}_n(\Lambda)+   \sqrt{\frac{\log(1/\delta)}{2n}}+\frac{2}{n}.\]
Then, for \textnormal{\texttt{ModSel-CP}}, it holds that
\[\widehat{C}_{\textnormal{\texttt{ModSel-CP}}} (X_{n+1}) \subseteq \left\{y\in \mathcal{Y} : \mathcal{L}_{n+1}({\hat\lambda},S^{\hat\lambda}(X_{n+1},y)) \leq \mathcal{L}_{n+1}({\hat\lambda},\hat{q}(\hat\lambda))  + \kappa(\Delta_n,\Delta'_n)\right\}\]
with probability $\geq 1-2\delta-2\delta'$. Similarly, for \textnormal{\texttt{ModSel-CP-LOO}}, it holds that
\[ \widehat{C}_{\textnormal{\texttt{ModSel-CP-LOO}}} (X_{n+1}) \subseteq \left\{y\in \mathcal{Y} : \mathcal{L}_{n+1}({\hat\lambda},S^{\hat\lambda}(X_{n+1},y)) \leq \mathcal{L}_{n+1}({\hat\lambda},\hat{q}(\hat\lambda))  + \kappa_{\textnormal{\texttt{LOO}}}(\Delta_n,\Delta'_n)\right\} \]
with probability $\geq 1-2\delta-2\delta'$. 
\end{theorem}

Of course, for this result to be meaningful, as in the residual score setting (in Theorem~\ref{thm:efficiency_res_}) we need to be in a setting where the Rademacher complexity term $\mathcal{R}_n(\Lambda)$ is small, and where the term $\kappa(\Delta_n,\Delta'_n)$ or $\kappa_{\texttt{LOO}}(\Delta_n,\Delta'_n)$ is small---this latter condition again is simply expressing the idea that nearly-optimal models should be similar to each other (see earlier discussion in Section~\ref{sec:examples_for_optimality}).

\subsubsection{The residual score as a special case}
In the residual score setting, we have $\mathcal{L}_{n+1}(\lambda,q) = q$ by definition, and therefore we can take $\Delta'_n=0$ and $\delta'=0$.
By definition we have $\Lambda_*(\Delta) = \Lambda_*(\Delta,0)$ (where $\Lambda_*(\Delta)$ is defined as in Theorem~\ref{thm:efficiency_res_} while $\Lambda_*(\Delta,0)$ is defined as in the general theorem, Theorem~\ref{thm:efficiency_general_}). We also calculate
\begin{align*}
    \kappa(\Delta,0) &= \max_{\lambda,\lambda'\in\Lambda_*(\Delta,0)}\sup_{y\in\mathcal{Y}}\left| \mathcal{L}_{n+1}(\lambda,S^\lambda(X_{n+1},y)) -  \mathcal{L}_{n+1}({\lambda'},S^{\lambda'}(X_{n+1},y)) \right| \\
    &= \max_{\lambda,\lambda'\in\Lambda_*(\Delta)}\sup_{y\in\mathcal{Y}}\left| S^\lambda(X_{n+1},y) - S^{\lambda'}(X_{n+1},y) \right| \\
    &= \max_{\lambda,\lambda'\in\Lambda_*(\Delta)}\sup_{y\in\mathcal{Y}}\left| |y - f_\lambda(X_{n+1})| - |y - f_{\lambda'}(X_{n+1})| \right| \\ 
    &\leq \max_{\lambda,\lambda'\in\Lambda_*(\Delta)}\left|  f_\lambda(X_{n+1})- f_{\lambda'}(X_{n+1}) \right| = \kappa(\Delta)
\end{align*}
and similarly,
\[\kappa_{\texttt{LOO}}(\Delta,0) \leq\kappa_{\texttt{LOO}}(\Delta).\]
Then we can see that Theorem~\ref{thm:efficiency_res_} is simply a special case of the more general result in Theorem~\ref{thm:efficiency_general_}.

\subsubsection{Proof of Theorem~\ref{thm:efficiency_general_} for \texttt{ModSel-CP}}

\paragraph{Step 1: reduce to an event on sets.}
    First, we compute a deterministic bound on the prediction set. For \texttt{ModSel-CP}, by Theorem~\ref{thm:ModSel_simplify} we have
    \[\widehat{C}_{\textnormal{\texttt{ModSel-CP}}} (X_{n+1}) \subseteq \left\{y \in \mathcal{Y}: \min_{\lambda\in\mathcal{M}}\mathcal{L}_{n+1}(\lambda,S^\lambda(X_{n+1},y))\le\mathcal{L}_{n+1}({\hat\lambda},\hat{q}(\hat\lambda)) \right\}.\]
    We can therefore relax this to
    \begin{multline*}\widehat{C}_{\textnormal{\texttt{ModSel-CP}}} (X_{n+1}) \subseteq \Bigg\{y \in \mathcal{Y}: 
    \mathcal{L}_{n+1}({\hat\lambda},S^{\hat\lambda}(X_{n+1},y)) \leq\\{} \mathcal{L}_{n+1}({\hat\lambda},\hat{q}(\hat\lambda))  + \max_{\lambda\in\mathcal{M}}\left| \mathcal{L}_{n+1}(\lambda,S^\lambda(X_{n+1},y)) -  \mathcal{L}_{n+1}({\hat\lambda},S^{\hat\lambda}(X_{n+1},y)) \right| \Bigg\}.\end{multline*}
Moreover, on the event that $\mathcal{M}\subseteq\Lambda_{\ast}(\Delta_n,\Delta'_n)$, we can simplify further: by definition of $\kappa(\Delta_n,\Delta'_n)$ (and using the fact that $\hat\lambda\in\mathcal{M}$ by definition), we have
\begin{equation*}
    \mathcal{M}\subseteq\Lambda_{\ast}(\Delta_n,\Delta'_n) \Longrightarrow  \widehat{C}_{\textnormal{\texttt{ModSel-CP}}} (X_{n+1}) \subseteq \left\{y : \mathcal{L}_{n+1}({\hat\lambda},S^{\hat\lambda}(X_{n+1},y)) \leq \mathcal{L}_{n+1}({\hat\lambda},\hat{q}(\hat\lambda))  + \kappa(\Delta_n,\Delta'_n)\right\}.
\end{equation*}
Therefore,
\begin{multline*}
    \mathbb{P}\left\{\widehat{C}_{\textnormal{\texttt{ModSel-CP}}} (X_{n+1}) \subseteq \left\{y : \mathcal{L}_{n+1}({\hat\lambda},S^{\hat\lambda}(X_{n+1},y)) \leq \mathcal{L}_{n+1}({\hat\lambda},\hat{q}(\hat\lambda))  + \kappa(\Delta_n,\Delta'_n)\right\}\right\} \\\geq 1  - \mathbb{P}\{\mathcal{M}\not\subseteq\Lambda_*(\Delta_n,\Delta'_n)\}.
\end{multline*}
From this point on, we only need to bound this remaining probability, $\mathbb{P}\{\mathcal{M}\not\subseteq\Lambda_*(\Delta_n,\Delta'_n)\}$.

\paragraph{Step 2: concentration bounds.} 
 We now need to establish some concentration results on the sample quantiles of the scores. 
 Fix any $\tau\in[0,1]$.
First, by the definition of Rachemacher complexity, together with a standard symmetrization step \citep[Section 2.1]{koltchinskii2011oracle}, we have that
\begin{align*}
&\mathbb{E}\Bigg[\max_{\lambda\in\Lambda} \Bigg|  \frac{1}{n}\sum_{i=1}^n\Bigg(\mathbf{1}\left\{S^\lambda(X_i,Y_i)\leq Q_{\tau - 2\mathcal{R}_n(\Lambda)-  \sqrt{\frac{\log(1/\delta)}{2n}}}(\lambda) \right\} - \left(\tau - 2\mathcal{R}_n(\Lambda)-  \sqrt{\frac{\log(1/\delta)}{2n}}\right)\Bigg)\Bigg|\Bigg] \\
& = \mathbb{E}\Bigg[\max_{\lambda\in\Lambda} \Bigg| \frac{1}{n}\sum_{i=1}^n\Bigg(\mathbf{1}\left\{S^\lambda(X_i,Y_i)\leq Q_{\tau - 2\mathcal{R}_n(\Lambda)-  \sqrt{\frac{\log(1/\delta)}{2n}}}(\lambda) \right\} - \mathbb{P}\{S^\lambda(X,Y)\leq Q_{\tau - 2\mathcal{R}_n(\Lambda)-  \sqrt{\frac{\log(1/\delta)}{2n}}}(\lambda)\}\Bigg)\Bigg|\Bigg]\\
& \leq 2\mathbb{E}\left[\max_{\lambda\in\Lambda} \left| \frac{1}{n}\sum_{i=1}^n\xi_i\mathbf{1}\left\{S^\lambda(X_i,Y_i)\leq Q_{\tau - 2\mathcal{R}_n(\Lambda)-  \sqrt{\frac{\log(1/\delta)}{2n}}}(\lambda)\right\} \right|\right]   \le 2\mathcal{R}_n(\Lambda),
\end{align*}
where $\xi_i\stackrel{\rm i.i.d.}{\sim}\textnormal{Uniform}\{\pm 1\}$, and where the first step uses the fact that
$$\mathbb{P}\left\{S^\lambda(X,Y)\leq Q_{\tau - 2\mathcal{R}_n(\Lambda)-  \sqrt{\frac{\log(1/\delta)}{2n}}}(\lambda) \right\} = \tau - 2\mathcal{R}_n(\Lambda)-  \sqrt{\frac{\log(1/\delta)}{2n}}$$ (by assumption of the continuous distribution of the scores).
In particular, this implies
\[\mathbb{E}\left[\max_{\lambda\in\Lambda}  \frac{1}{n} \sum_{i=1}^n\mathbf{1}\left\{S^\lambda(X_i,Y_i)\leq Q_{\tau - 2\mathcal{R}_n(\Lambda)-  \sqrt{\frac{\log(1/\delta)}{2n}}}(\lambda)\right\}\right]  \leq \tau -  \sqrt{\frac{\log(1/\delta)}{2n}}.\]
By McDiarmid's inequality,
\begin{multline*}\max_{\lambda\in\Lambda}  \frac{1}{n} \sum_{i=1}^n\mathbf{1}\left\{S^\lambda(X_i,Y_i) < Q_{\tau - 2\mathcal{R}_n(\Lambda)-  \sqrt{\frac{\log(1/\delta)}{2n}}}(\lambda)\right\}\\
< \mathbb{E}\left[\max_{\lambda\in\Lambda}  \frac{1}{n} \sum_{i=1}^n\mathbf{1}\left\{S^\lambda(X_i,Y_i) < Q_{\tau - 2\mathcal{R}_n(\Lambda)-  \sqrt{\frac{\log(1/\delta)}{2n}}}(\lambda)\right\}\right] + \sqrt{\frac{\log(1/\delta)}{2n}} \leq \tau \end{multline*}
with probability at least $1-\delta$, and therefore,
\begin{align}
\notag&\mathbb{P}\left\{\textnormal{Quantile}_\tau\left(\{S^\lambda(X_i,Y_i)\}_{i\in[n]}\right) \geq Q_{\tau- 2\mathcal{R}_n(\Lambda)-   \sqrt{\frac{\log(1/\delta)}{2n}}}(\lambda),\ \forall\ \lambda\in\Lambda\right\}\\
\notag&=\mathbb{P}\left\{\sum_{i=1}^n\mathbf{1}\left\{S^\lambda(X_i,Y_i) < Q_{\tau -  2\mathcal{R}_n(\Lambda)- \sqrt{\frac{\log(1/\delta)}{2n}}}(\lambda)\right\}< n \tau, \ \forall \ \lambda\in\Lambda\right\}\\
\label{eqn:concentration_all_lambda_lowerbd}&=\mathbb{P}\left\{\max_{\lambda\in\Lambda} \sum_{i=1}^n\mathbf{1}\left\{S^\lambda(X_i,Y_i) < Q_{\tau -  2\mathcal{R}_n(\Lambda)- \sqrt{\frac{\log(1/\delta)}{2n}}}(\lambda)\right\}< n \tau\right\}\geq 1-\delta,
\end{align}
where the second equality is due to the definition of the sample quantile.
Next we need an analogous result for the upper bound. We follow essentially the same steps. First we have
\begin{multline*}
    \mathbb{E}\left[\max_{\lambda\in\Lambda} \left| \frac{1}{n} \sum_{i=1}^n\left(\mathbf{1}\left\{S^\lambda(X_i,Y_i)\leq Q_{\tau + 2\mathcal{R}_n(\Lambda)+  \sqrt{\frac{\log(1/\delta)}{2n}}}(\lambda) \right\} - \left(\tau + 2\mathcal{R}_n(\Lambda)+ \sqrt{\frac{\log(1/\delta)}{2n}}\right)\right)\right|\right]\\ \leq 2\mathcal{R}_n(\Lambda),
\end{multline*}
and therefore,
\[\mathbb{E}\left[\min_{\lambda\in\Lambda} \frac{1}{n} \sum_{i=1}^n\mathbf{1}\left\{S^\lambda(X_i,Y_i) \leq Q_{\tau + 2\mathcal{R}_n(\Lambda)+ \sqrt{\frac{\log(1/\delta)}{2n}}}(\lambda)\right\} \right] \geq \tau + \sqrt{\frac{\log(1/\delta)}{2n}}.\]
Applying McDiarmid's inequality, then, with probability at least $1-\delta$,
\[\min_{\lambda\in\Lambda} \frac{1}{n} \sum_{i=1}^n\mathbf{1}\left\{S^\lambda(X_i,Y_i) \leq Q_{\tau + 2\mathcal{R}_n(\Lambda)+ \sqrt{\frac{\log(1/\delta)}{2n}}}(\lambda)\right\} \geq \tau,\]
which implies
\begin{equation}\label{eqn:concentration_all_lambda_upperbd}\mathbb{P}\left\{\textnormal{Quantile}_\tau\left(\{S^\lambda(X_i,Y_i)\}_{i\in[n]}\right) \leq Q_{\tau+ 2\mathcal{R}_n(\Lambda)+   \sqrt{\frac{\log(1/\delta)}{2n}}}(\lambda),\ \forall\ \lambda\in\Lambda\right\} \geq 1-\delta.
\end{equation}
\paragraph{Step 3: bounding the set.} Now we are ready to prove that $\mathcal{M}\subseteq\Lambda_*(\Delta_n,\Delta'_n)$ with high probability.
From this point on, we assume the event in~\eqref{eqn:concentration_all_lambda_lowerbd} holds with $\tau = (1-\alpha)(1+1/n)-1/n$, and the event in~\eqref{eqn:concentration_all_lambda_upperbd}  holds with $\tau = (1-\alpha)(1+1/n)+1/n$, and also that assumption~\eqref{eqn:asm_loss_concentrates} holds with the choice $t_\lambda = Q_{1-\alpha-\Delta_n}(\lambda)$, and again with the choice $t_\lambda = Q_{1-\alpha+\Delta_n}(\lambda)$. (With probability at least $1-2\delta-2\delta'$, these statements are all true, from the work above.) 

Then for any $\lambda\in\mathcal{M}$,
\begin{multline*}
    \mathcal{L}_{n+1}(\lambda,\hat{q}_{-}(\lambda)) \le \mathcal{L}_{n+1}(\hat\lambda,\hat{q}(\hat\lambda)) = \min_{\lambda'\in\Lambda}\mathcal{L}_{n+1}(\lambda',\hat{q}(\lambda'))\\\Longrightarrow \mathcal{L}_{n+1}(\lambda,Q_{1-\alpha-\Delta_n}(\lambda))\leq \min_{\lambda'\in\Lambda}\mathcal{L}_{n+1}(\lambda',Q_{1-\alpha+\Delta_n}(\lambda')),
\end{multline*}
where the first inequality holds by definition of $\mathcal{M}$, while the implication holds by events~\eqref{eqn:concentration_all_lambda_lowerbd} and~\eqref{eqn:concentration_all_lambda_upperbd} with the choices of $\tau$ defined above (along with the fact that $\mathcal{L}_{n+1}$ is monotone in $q$).
By assumption~\eqref{eqn:asm_loss_concentrates}, this implies
\[\mathbb{E}\left[\mathcal{L}_{n+1}(\lambda,Q_{1-\alpha-\Delta_n}(\lambda))\right]\leq  \min_{\lambda'\in\Lambda}\mathbb{E}\left[\mathcal{L}_{n+1}(\lambda',Q_{1-\alpha+\Delta_n}(\lambda'))\right] + 2\Delta'_n.\]
Therefore we have
\[\lambda\in\Lambda_*(\Delta_n,\Delta'_n),\]
as desired.

\subsubsection{Proof of Theorem~\ref{thm:efficiency_general_} for \texttt{ModSel-CP-LOO}}
\paragraph{Step 1: reduce to an event on sets.}
As for \texttt{ModSel-CP}, we begin with a deterministic calculation: since for each $i$ and all $y$ we have $\hat\lambda_{-i}(y)\in\mathcal{M}_i$ (as proved in Proposition~\ref{prop:Mcal_i}), we have
\begin{align*}
\widehat{C}_{\texttt{ModSel-CP-LOO}} (X_{n+1}) &= \Bigg\{y \in \mathcal{Y}: 
\mathcal{L}_{n+1}\big({\hat\lambda}, S^{\hat\lambda}(X_{n+1},y)\big) \leq {}\\
&\hspace{0.7in}\textnormal{Quantile}_{(1-\alpha)(1+1/n)}\left(\left\{ \mathcal{L}_{n+1}\big({\hat\lambda_{-i}(y)},S^{\hat\lambda_{-i}(y)}(X_i,Y_i)\big)\right\}_{i\in[n]}\right)\Bigg\}\\
&\subseteq \Bigg\{y \in \mathcal{Y}: 
\mathcal{L}_{n+1}\big({\hat\lambda}, S^{\hat\lambda}(X_{n+1},y)\big) \leq {}\\
&\hspace{0.7in}\textnormal{Quantile}_{(1-\alpha)(1+1/n)}\left(\left\{ \max_{\lambda\in\mathcal{M}_i}\mathcal{L}_{n+1}\big({\lambda},S^\lambda(X_i,Y_i)\big)\right\}_{i\in[n]}\right)\Bigg\} .\end{align*}
Now we will consider the event that $\mathcal{M}_i\subseteq\Lambda_*(\Delta_n,\Delta'_n)$ for all $i$: by definition of $\kappa_{\texttt{LOO}}(\Delta_n,\Delta'_n)$, and using the fact that $\hat\lambda\in\mathcal{M}_i$ by definition,
\[\mathcal{M}_i\subseteq\Lambda_*(\Delta_n) \Longrightarrow \max_{\lambda\in\mathcal{M}_i}\mathcal{L}_{n+1}\big(\lambda,S^\lambda (X_i,Y_i)\big) \leq \mathcal{L}_{n+1}\big(\hat\lambda,S^{\hat\lambda}(X_i,Y_i)\big) + \kappa_{\texttt{LOO}}(\Delta_n,\Delta'_n), \]
for all $i$. Therefore, if $\mathcal{M}_i \subseteq\Lambda_*(\Delta_n,\Delta'_n) $ holds for all $i=1,\dots,n$, then we have
\begin{multline*}\textnormal{Quantile}_{(1-\alpha)(1+1/n)}\left(\left\{ \max_{\lambda\in\mathcal{M}_i}\mathcal{L}_{n+1}\big({\lambda},S^\lambda(X_i,Y_i)\big)\right\}_{i\in[n]}\right) \leq {}\\ \textnormal{Quantile}_{(1-\alpha)(1+1/n)}\left(\left\{ \mathcal{L}_{n+1}\big({\hat\lambda},S^{\hat\lambda}(X_i,Y_i)\big)\right\}_{i\in[n]}\right)  +\kappa_{\texttt{LOO}}(\Delta_n,\Delta'_n) = \hat{q}(\hat\lambda) + \kappa_{\texttt{LOO}}(\Delta_n,\Delta'_n) . \end{multline*}
This implies that
\begin{multline*}\mathbb{P}\left\{ \widehat{C}_{\texttt{ModSel-CP-LOO}} (X_{n+1}) 
 \subseteq \left\{y\in\mathcal{Y} : \mathcal{L}_{n+1}(\hat\lambda,S^{\hat\lambda}(X_{n+1},y))\leq   \mathcal{L}_{n+1}(\hat{\lambda}, \hat{q}(\hat\lambda)  ) + \kappa_{\texttt{LOO}}(\Delta_n,\Delta'_n) \right\}\right\}\\ \geq 1 - \mathbb{P}\{\cup_{i\in[n]}\mathcal{M}_i \not \subseteq\Lambda_*(\Delta_n,\Delta'_n)\}.\end{multline*}
We now need to bound this last probability.

\paragraph{Step 2: concentration bounds.} The calculations for Step 2 are identical as for \texttt{ModSel-CP}, so we do not repeat them here.

\paragraph{Step 3: bounding the set.}
Now we need to show that $\mathcal{M}_i\subseteq\Lambda_*(\Delta_n,\Delta'_n)$ for all $i$, with high probability.
From this point on, we assume the event in~\eqref{eqn:concentration_all_lambda_lowerbd} holds with $\tau = (1-\alpha)(1+1/n)-1/n$, and the event in~\eqref{eqn:concentration_all_lambda_upperbd}  holds with $\tau = (1-\alpha)(1+1/n)+1/n$, and also that assumption~\eqref{eqn:asm_loss_concentrates} holds with the choice $t_\lambda = Q_{1-\alpha-\Delta_n}(\lambda)$, and again with the choice $t_\lambda = Q_{1-\alpha+\Delta_n}(\lambda)$. (With probability at least $1-2\delta-2\delta'$, these statements are all true, from the work above.) 

Now fix any $i$ and consider any $\lambda\in\mathcal{M}_i$.
By definition of $\mathcal{M}_i$, we have
\begin{multline*}l_i(\lambda)\leq \min_{\lambda'\in\Lambda}u_i(\lambda')
\Longrightarrow  \mathcal{L}_{n+1}(\lambda,\hat{q}_{-}(\lambda)) \leq  \min_{\lambda'\in\Lambda}\mathcal{L}_{n+1}(\lambda',\hat{q}_{+}(\lambda'))\\
\Longrightarrow
\mathcal{L}_{n+1}(\lambda,Q_{1-\alpha-\Delta_n}(\lambda)) \leq \min_{\lambda'\in\Lambda}\mathcal{L}_{n+1}(\lambda',Q_{1-\alpha+\Delta_n}(\lambda')).\end{multline*}
Therefore,
\[\mathbb{E}\left[\mathcal{L}_{n+1}(\lambda,Q_{1-\alpha-\Delta_n}(\lambda))\right] \leq \min_{\lambda'\in\Lambda}\mathbb{E}\left[\mathcal{L}_{n+1}(\lambda',Q_{1-\alpha+\Delta_n}(\lambda'))\right] + 2\Delta'_n,\]
and therefore 
$\lambda\in\Lambda_*(\Delta_n,\Delta'_n)$
as desired.

\subsection{Proof for the regularity condition example: Proposition~\ref{prop:positive_example}}
\paragraph{Bounding the Rademacher complexity.}
Fix any set of $\{t_\lambda : \lambda \in \Lambda\}$, and denote
\[\mathcal{C}_{\Lambda} = \left\{\mathbf{1}\{|y-\lambda^{\top}x| \le t_\lambda\}: \lambda \in \Lambda \right\}.\]
We aim to bound the Rademacher complexity of class $\mathcal{C}_\Lambda$.

First, note that for any $\lambda \in \mathbb{R}^d$ and any $t \in \mathbb{R}$,
\[\mathbf{1}\{|y-\lambda^{\top}x| \le t\} = \mathbf{1} \left\{ y - \lambda^{\top}x - t \le 0 \right\} \cdot \mathbf{1} \left\{y - \lambda^{\top}x +t \ge 0\right\}, \]
which is the product of two linear classifiers. Since it is known that the VC dimension of the class of all $s$-sparse linear classifiers upper bounded by $4s \mathrm{log}(de/s)$ \citep[Lemma 1]{abramovich2018high},  we have that, by \citet[Lemma 3.2.3]{blumer1989Learnability},
\[\mathrm{VC}\textnormal{-dim}(\mathcal{C}_\Lambda) \le 32\mathrm{log}(6) \cdot s \mathrm{log}\left(\frac{ed}{s}\right).\]
In other words, the class $\mathcal{C}_\Lambda$ has finite VC dimension. Thus, it holds from the Haussler's bound \citep{Van1996Weak} that there exist universal constants $C_1>0$, for any probability measure $Q$,
\[ N(\epsilon, \mathcal{C}_\Lambda, L_2(Q)) \le C_1 \cdot \left(s \mathrm{log}\frac{ed}{s} \right) \cdot \left( \frac{16e}{\epsilon^2}\right)^{\left(s \mathrm{log}\frac{ed}{s} \right)},\]
where $N(\epsilon, \mathcal{C}_\Lambda, L_2(Q)) $ is the $\epsilon$-covering number of $\mathcal{C}_\Lambda$ under the metric of $L_2(Q)$.

Finally, it follows from \cite[Theorem 3.11]{koltchinskii2011oracle} that
\begin{align*}
    \mathcal{R}_n(\Lambda) = \sup \limits_{\{t_\lambda : \lambda \in \Lambda\}}\mathcal{R}_n(\mathcal{C}_\Lambda) &\le  \sup \limits_{\{t_\lambda : \lambda \in \Lambda\}} \frac{C_2}{\sqrt{n}} \int_{0}^1 \sqrt{\mathrm{log}(N(\epsilon, \mathcal{C}_\Lambda, L_2(P_n)))} d\epsilon \\
    & \le  \frac{C_2}{\sqrt{n}} \int_{0}^1 \sqrt{\mathrm{log} \left(C_1 \cdot s \mathrm{log}\left(\frac{ed}{s}\right) \right) + s \mathrm{log}\left(\frac{ed}{s}\right) \mathrm{log} \left( \frac{16e}{\epsilon^2}\right) }d\epsilon \\
    & \le C_3 \cdot \sqrt{\frac{s \mathrm{log}\left(\frac{ed}{s}\right)}{n}} \int_{0}^1 \sqrt{1 + \mathrm{log}\left( \frac{1}{\epsilon}\right)} d \epsilon \\
    & \le C \cdot \sqrt{\frac{s \mathrm{log}\left(\frac{ed}{s}\right)}{n}},
\end{align*}
where $C_2, C_3, C>0$ are different constants.

\paragraph{Bounding $\kappa(\Delta)$ and $\kappa_{\textnormal{\texttt{LOO}}}(\Delta)$.}
First, by construction, for any $\lambda$ and for $(X,Y)\sim P$ we have
\[Y - X^\top\lambda \sim \mathcal{N}(0,(\lambda-\lambda_{\rm true})^\top \Sigma (\lambda-\lambda_{\rm true}) + \sigma^2)\]
and so
\[S^\lambda(X,Y) = |Y - X^\top\lambda| \sim \sqrt{\|\lambda-\lambda_{\rm true}\|^2_\Sigma+ \sigma^2} \cdot \chi_1,\]
where $\|a\|_\Sigma = \sqrt{a^\top \Sigma a}$.
Therefore, for any $\tau$,
\[Q_{\tau}(\lambda) = \sqrt{\|\lambda-\lambda_{\rm true}\|^2_\Sigma + \sigma^2} \cdot \Phi^{-1}((1+\tau)/2),\]
since $\Phi^{-1}((1+\tau)/2)$ is the $\tau$-quantile of the $\chi_1$ distribution.
So,
\[\min_{\lambda\in\Lambda} Q_\tau(\lambda) = \Phi^{-1}((1+\tau)/2) \cdot \min_{\lambda\in\Lambda} \sqrt{\|\lambda-\lambda_{\rm true}\|^2_\Sigma + \sigma^2}  = \Phi^{-1}((1+\tau)/2)\cdot \sigma\sqrt{\gamma^2+1},\]
by definition of $\gamma$. Therefore,
\begin{multline*}\Lambda_{\ast}(\Delta) =\left\{\lambda\in\Lambda : \sqrt{\|\lambda-\lambda_{\rm true}\|^2_\Sigma + \sigma^2} \cdot \Phi^{-1}\left(1 - \frac{\alpha+\Delta}{2}\right) \leq \Phi^{-1}\left(1 - \frac{\alpha-\Delta}{2}\right) \cdot \sigma\sqrt{\gamma^2+1}\right\} \\=\left\{\lambda\in\Lambda :  \|\lambda-\lambda_{\rm true}\|_\Sigma \leq \sigma \sqrt{(\gamma^2+1) \cdot \left(\frac{\Phi^{-1}\left(1 - \frac{\alpha-\Delta}{2}\right)}{\Phi^{-1}\left(1 - \frac{\alpha+\Delta}{2}\right)}\right)^{2} - 1}\right\}\\ = \left\{\lambda\in\Lambda : \|\lambda-\lambda_{\rm true}\|_\Sigma \leq \sigma \sqrt{(1+\gamma^2)(1+\phi)^2 - 1}\right\}.\end{multline*}

Denote $D_{J} \in \mathbb{R}^{d \times d}$ as the diagonal matrix where $D_{jj} = \mathbf{1}\{j \in J\}$ for any $j \in [d]$. Denote $\textnormal{\texttt{supp}}(v) = \{j \in [d]: v_j \ne 0\}$.
Since for any $\lambda,\lambda'\in\Lambda_{\ast}(\Delta) \subseteq \Lambda$, and for any $x$,
\begin{multline*}
    |f_{\lambda}(x) - f_{\lambda'}(x)| = \left|(\lambda - \lambda')^\top x\right| = \left|(\lambda - \lambda')^\top D_{\textnormal{\texttt{supp}}(\lambda) \cup \textnormal{\texttt{supp}}(\lambda')} x\right| \\ 
    = \left|(\lambda - \lambda')^\top \Sigma^{1/2} \cdot \Sigma^{-1/2} D_{\textnormal{\texttt{supp}}(\lambda) \cup \textnormal{\texttt{supp}}(\lambda')}x\right| \leq \|\lambda - \lambda'\|_\Sigma \cdot \|\Sigma^{-1/2} D_{\textnormal{\texttt{supp}}(\lambda) \cup \textnormal{\texttt{supp}}(\lambda')} x\|,
\end{multline*}
therefore,
\begin{multline*}
    \sup_{\lambda,\lambda'\in\Lambda_{\ast}(\Delta)} |f_{\lambda}(x)-f_{\lambda'}(x)| \leq \sup_{\lambda,\lambda'\in\Lambda_{\ast}(\Delta)}\|\lambda - \lambda'\|_\Sigma \cdot \|\Sigma^{-1/2} D_{\textnormal{\texttt{supp}}(\lambda) \cup \textnormal{\texttt{supp}}(\lambda')} x\| \\
    \leq 2\sigma \sqrt{(1+\gamma^2)(1+\phi)^2 - 1} \cdot \sup \limits_{\substack{J\subseteq [d] \\ |J|=2s}} \|\Sigma^{-1/2} D_{J} x\| .
\end{multline*}
Therefore,
\[\kappa(\Delta) \leq 2\sigma \sqrt{(1+\gamma^2)(1+\phi)^2 - 1}\cdot  \sup \limits_{\substack{J\subseteq [d] \\ |J|=2s}} \|\Sigma^{-1/2} D_{J} X_{n+1}\|, \]
and similarly
\[\kappa_{\texttt{LOO}}(\Delta) \leq 2\sigma \sqrt{(1+\gamma^2)(1+\phi)^2 - 1} \cdot \max_{i=1,\dots,n} \sup \limits_{\substack{J\subseteq [d] \\ |J|=2s}} \|\Sigma^{-1/2} D_{J} X_{i}\| .\]

Note that for any $J \subseteq [d]$, $\mathrm{rank}(\Sigma^{-1/2} D_{J} \Sigma^{1/2}) = |J|$ and the largest singular value of $\Sigma^{-1/2} D_{J} \Sigma^{1/2}$ is upper bounded by $\sqrt{\frac{\eta_{\max}(\Sigma)}{\eta_{\min}(\Sigma)}}$, where $\eta(\Sigma)$ denotes the eigenvalue of $\Sigma$. Thus, write $X = \Sigma^{1/2}Z$ where $Z \sim \mathcal{N}(0, I_d)$, and we have that
\[\sup \limits_{\substack{J\subseteq [d] \\ |J|=2s}} \|\Sigma^{-1/2} D_{J} X\| = \sup \limits_{\substack{J\subseteq [d] \\ |J|=2s}} \|\Sigma^{-1/2} D_{J} \Sigma^{1/2}Z\|\le_{st.} \max \limits_{\substack{J\subseteq [d] \\ |J|=2s}} \sqrt{\frac{\eta_{\max}(\Sigma)}{\eta_{\min}(\Sigma)}}\cdot\sqrt{\sum_{j \in J} Z_j^2}.\]
Since $\|Z_J\| \sim \chi_{2s}$ for any $J$, we have that
\[\mathbb{P} \left(\max \limits_{\substack{J\subseteq [d] \\ |J|=2s}} \|Z_J\| \ge \sqrt{2s} + \sqrt{2u} \right) \le {d \choose 2s} \cdot \mathbb{P} \left( \|Z_{[2s]}\| \ge \sqrt{2s} + \sqrt{2u} \right) \le {d \choose 2s} \cdot e^{-u}.\]
Set $u = \mathrm{log}{d \choose 2s} + \mathrm{log}(1/\delta)$, we have that
\begin{multline*}
    \mathbb{P} \left(\max \limits_{\substack{J\subseteq [d] \\ |J|=2s}} \|Z_J\| \le \sqrt{2s} + \sqrt{2 \left(2s\mathrm{log}\left(\frac{ed}{2s} \right) + \mathrm{log}(1/\delta) \right)} \right) \\ \ge \mathbb{P} \left(\max \limits_{\substack{J\subseteq [d] \\ |J|=2s}} \|Z_J\| \le \sqrt{2s} + \sqrt{2 \left(\mathrm{log}{d \choose 2s} + \mathrm{log}(1/\delta) \right)} \right) \ge 1-\delta.
\end{multline*}
Thus, we have that
\[\mathbb{P}  \left\{\kappa(\Delta) \le 8\sigma \sqrt{\frac{\eta_{\max}(\Sigma)}{\eta_{\min}(\Sigma)}}\cdot\sqrt{(1+\gamma^2)(1+\phi)^2 - 1}\cdot \sqrt{ s\mathrm{log}\left(\frac{ed}{s} \right) + \mathrm{log}(1/\delta)}\right\} \ge 1-\delta.\]
Similarly, we have that
\[\mathbb{P}  \left\{\kappa_{\textnormal{\texttt{LOO}}}(\Delta) \le 8\sigma\sqrt{\frac{\eta_{\max}(\Sigma)}{\eta_{\min}(\Sigma)}}\cdot \sqrt{(1+\gamma^2)(1+\phi)^2 - 1}\cdot \sqrt{ s\mathrm{log}\left(\frac{ed}{s} \right) + \mathrm{log}(n/\delta)}\right\} \ge 1-\delta.\]

\section{Comparing the construction of \texttt{ModSel-CP} and \texttt{ModSel-CP-LOO}}\label{sec:ENSvsLOO-additional}
To better understand our methods, in this section, we take a closer look at the differences between the construction of our two proposed methods, \texttt{ModSel-CP} and \texttt{ModSel-CP-LOO}.

As mentioned in Section~\ref{sec:ModSelvsLOO}, although \texttt{ModSel-CP} and \texttt{ModSel-CP-LOO} share the property that both are strictly more conservative than \texttt{YK-baseline} (under a mild assumption), they differ in the exact way that they enlarge the set from \texttt{YK-baseline} to achieve validity.
Recalling the definition~\eqref{eqn:YK-baseline_define} of \texttt{YK-baseline}, we can equivalently rewrite that method as
\begin{multline}\label{defn:YK-baseline_altdef}\widehat{C}_{\texttt{YK-baseline}}(X_{n+1}) = \left\{y \in \mathcal{Y}: S^{\hat\lambda}(X_{n+1},y)\leq \textnormal{Quantile}_{(1-\alpha)(1+1/n)}\left(S^{\hat\lambda}(X_1,Y_1),\dots,S^{\hat\lambda}(X_n,Y_n)\right)\right\}.\end{multline}
This procedure is not treating the $n+1$ data points (the calibration points $(X_1,Y_1),\dots,(X_n,Y_n)$, and a hypothesized test point $(X_{n+1},y)$) symmetrically, because $\hat\lambda$ is selected based on the first $n$ points only.
Rewriting the definition~\eqref{defn:ModSel} of our \texttt{ModSel-CP} method as
\begin{multline*}    
\widehat{C}_{\texttt{ModSel-CP}} (X_{n+1}) =\left\{y \in \mathcal{Y}: S^{\hat{\lambda}(y)}(X_{n+1},y) \le \mathrm{Quantile}_{(1-\alpha)(1+1/n)} \left(S^{\hat{\lambda}(y)}(X_{1},Y_{1}), \dots, S^{\hat{\lambda}(y)}(X_{n},Y_{n}) \right) \right\},\end{multline*}
we can therefore see that the prediction set $\widehat{C}_{\texttt{ModSel-CP}}(X_{n+1})$ is constructed in exactly the same way, except that we symmetrise the way that the $n+1$ data points are treated, since $\hat{\lambda}(y)$ is now the model selected based on all $n+1$ data points. 

Turning to \texttt{ModSel-CP-LOO}, this variant of our method restores symmetry in a different way. Taking the case of the residual score, $S^\lambda(x,y) = |y-f_\lambda(x)|$, to allow for a simpler and more intuitive comparison, in this special case we have $\mathcal{L}(\lambda,q;x_1,\dots,x_m) =2q$ (i.e., the loss is independent of the model $S^\lambda$ and the evaluation points $x_1,\dots,x_m$), and therefore the \texttt{ModSel-CP-LOO} prediction set defined in~\eqref{defn:LOO} can be simplified to
\begin{multline*}   
\widehat{C}_{\texttt{ModSel-CP-LOO}} (X_{n+1}) \\=\left\{y \in \mathcal{Y}: S^{\hat{\lambda}}(X_{n+1},y) \le \mathrm{Quantile}_{(1-\alpha)(1+1/n)} \left(S^{\hat{\lambda}_{-1}(y)}(X_{1},Y_{1}), \dots, S^{\hat{\lambda}_{-n}(y)}(X_{n},Y_{n}) \right) \right\}.\end{multline*}
In other words, this is again constructed in the same way as \texttt{YK-baseline}~\eqref{defn:YK-baseline_altdef}, but now symmetry is enforced by using the leave-one-out model selection for each data point's score. That is, the selected model $\hat\lambda$ uses the $n$ calibration points $\{(X_j,Y_j)\}_{j\in[n]}$ but not the hypothesized test point $(X_{n+1},y)$; so, analogously, in the quantile calculation, for each data point $(X_i,Y_i)$ we use the selected model $\hat\lambda_{-i}(y)$, which uses all data points $\{(X_j,Y_j)\}_{j\in[n]\backslash \{i\}}\cup\{(X_{n+1},y)\}$.

To take another perspective, we can again examine the \texttt{YK-baseline} construction, and identify which part of this construction is modified by the proposed methods to correct for selection bias. Returning to our expression~\eqref{defn:YK-baseline_altdef}, let us consider the left-hand-side (LHS) and right-hand-side (RHS) terms separately:
\begin{multline*}
    \widehat{C}_{\texttt{YK-baseline}}(X_{n+1}) = \bigg\{y \in \mathcal{Y}: \underbrace{S^{\hat\lambda}(X_{n+1},y)}_{\textnormal{LHS term}}\leq \underbrace{\textnormal{Quantile}_{(1-\alpha)(1+1/n)}\left(S^{\hat\lambda}(X_1,Y_1),\dots,S^{\hat\lambda}(X_n,Y_n)\right)}_{\textnormal{RHS term}}\bigg\}.
\end{multline*}
Since \texttt{YK-baseline} does not correct for selection bias, any valid method will need to be more conservative in terms of the criterion for including $y$ into the prediction set. Again, for simplicity, consider the case of the residual score, $S^\lambda(x,y) = |y-f_\lambda(x)|$, to allow for a more intuitive comparison.

Under this special score, the \texttt{ModSel-CP-LOO} prediction set is simplified to,
\begin{multline*}   \widehat{C}_{\texttt{ModSel-CP-LOO}} (X_{n+1}) \\=\left\{y \in \mathcal{Y}: S^{\hat{\lambda}}(X_{n+1},y) \le \mathrm{Quantile}_{(1-\alpha)(1+1/n)} \left(S^{\hat{\lambda}_{-1}(y)}(X_{1},Y_{1}), \dots, S^{\hat{\lambda}_{-n}(y)}(X_{n},Y_{n}) \right) \right\}.\end{multline*}
We see that \texttt{ModSel-CP-LOO} prediction set is constructed by adjusting the value of the RHS term (and keeping the LHS term the same). Conversely, comparing to the
expression~\eqref{eqn:ModSel_simple} in Theorem~\ref{thm:ModSel_simplify},  we see that the \texttt{ModSel-CP} prediction set is constructed by adjusting the value of the LHS term and keeping the RHS term the same.
Since two methods adjust terms differently to be more conservative so that selection bias is accounted for, we would expect the width of their resulting sets may differ under different settings.

For \texttt{ModSel-CP-LOO}, since it enlarges the set by enlarging the RHS term, namely, the adjusted empirical quantile (ignoring variation of model sources), then under settings where the competing models make similar predictions and yet $S^\lambda(X,Y)$ exhibits heavy tails, it is likely that \texttt{ModSel-CP-LOO} is more conservative than \texttt{ModSel-CP}. Evidence for this speculation can be found in the simulation results in Section~\ref{sec:simulation}, by comparing residual score results between ``NormalX + sparse structure + Gaussian noise'' and ``tX + sparse structure + Gaussian noise''.

For \texttt{ModSel-CP}, on the other hand, its major source of conservativeness comes from model multiplicity, i.e., when competing models are indistinguishable via score quantile, it could possibly result in \texttt{ModSel-CP} returning a large set. To be more specific, by \eqref{eqn:ModSel_simple}, if there are many competing models ($|\mathcal{M}|$ large) while each model can predict differently, the \texttt{ModSel-CP} set would be wide. This could be the case when both the data distribution and the model class are highly symmetrical. To demonstrate this point, we provide an additional simulation below.

\begin{itemize}
    \item \textbf{Distribution of $(X,Y)$:} $Y = X + \epsilon$, where $X \sim \mathcal{N}(0,1)$ and $\epsilon \sim \mathcal{N}(\mu, 1)$. $X$ is independent of $\epsilon$.

    \item \textbf{Models training:} For any $C>0$, the model class $\Lambda_C$ consists of two models:
    \[\Lambda_{C} = \left\{S^{+}(x,y) = |y-x-C|,  S^{-}(x,y) = |y-x+C|\right\}.\]

    \item As in Section~\ref{sec:simulation}, we fix $n=200$, $\alpha=0.1$. For each $C \in \{0.5, 1, 5\}$ and each $\mu \in \{-1, -0.5, 0, 0.5, 1 \}$, we run $5000$ independent trials of the simulation. The results are displayed below. Since \texttt{YK-adjust} always outputs prediction sets that are much larger than other methods, and we are comparing the performance difference between \texttt{ModSel-CP} and \texttt{ModSel-CP-LOO} under different settings, we omit displaying \texttt{YK-adjust} in Figure~\ref{fig:two_model_sim} to make the performance distinctions more prominent. The full results are reported in the table summary in Table~\ref{tab:two_model_sim} with standard errors in parentheses.
\end{itemize}

\begin{figure}
    \centering
    \includegraphics[width=\linewidth]{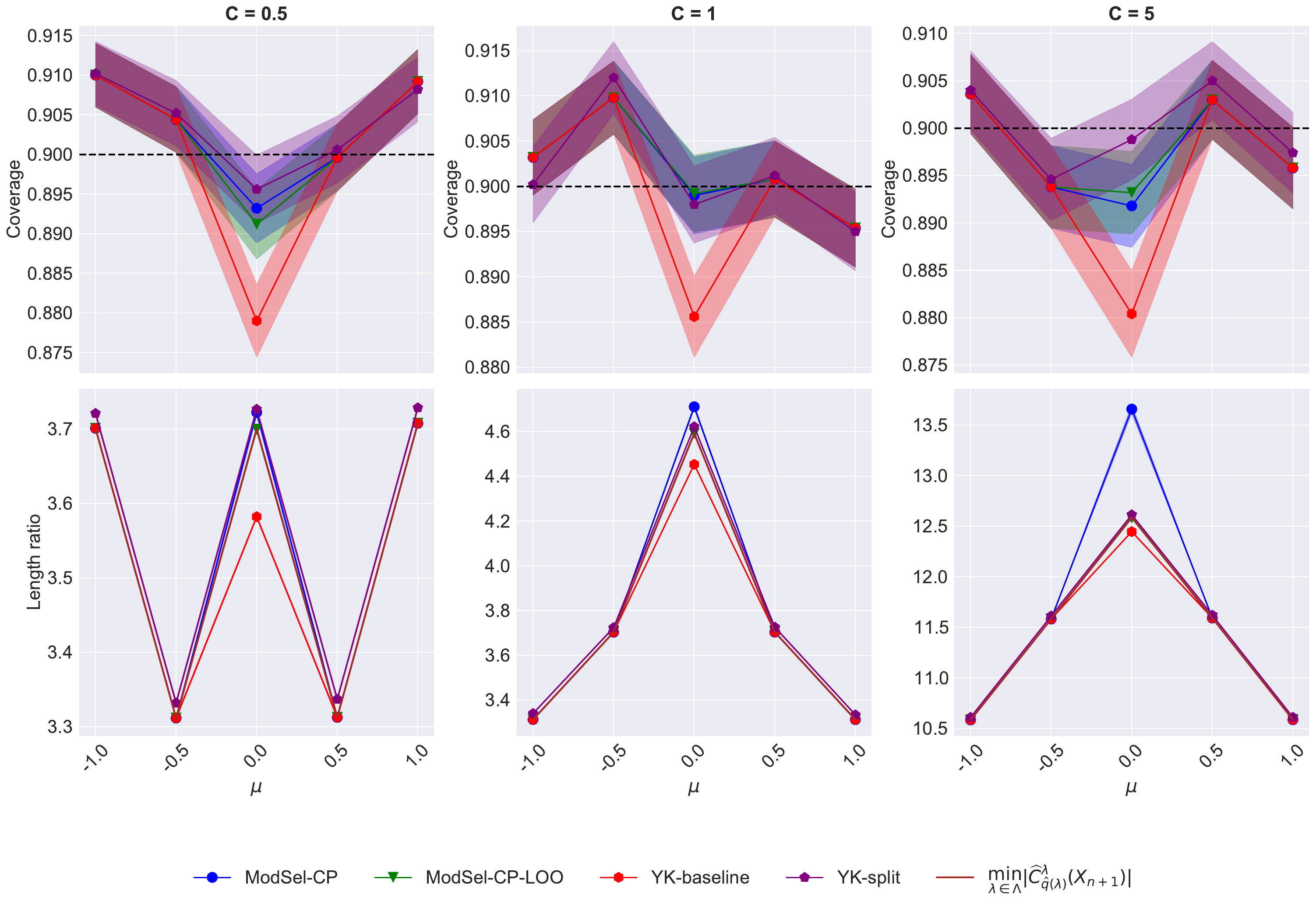}
    \caption{Results of additional simulation in Appendix~\ref{sec:ENSvsLOO-additional}.}
    \label{fig:two_model_sim}
\end{figure}

\begin{table}[h]
\caption{Results of simulation in Appendix~\ref{sec:ENSvsLOO-additional}}
    \centering
\begin{tabular}{cccccc}
\multicolumn{6}{c}{\textbf{$C = 0.5$, $n = 200$}}\\
\multicolumn{6}{l}{\textbf{Coverage}} \\
\bottomrule
{$\mu$} & ModSel-CP & ModSel-CP-LOO & YK-baseline & YK-split & YK-adjust \\
\hline
$-1.0$ & 0.910 (0.004) & 0.910 (0.004) & 0.910 (0.004) & 0.910 (0.004) & 0.981 (0.002) \\
$-0.5$ & 0.904 (0.004) & 0.904 (0.004) & 0.904 (0.004) & 0.905 (0.004) & 0.983 (0.002) \\
$0.0$ & 0.893 (0.004) & 0.891 (0.004) & 0.879 (0.005) & 0.896 (0.004) & 0.976 (0.002) \\
$0.5$ & 0.900 (0.004) & 0.900 (0.004) & 0.900 (0.004) & 0.901 (0.004) & 0.980 (0.002) \\
$1.0$ & 0.909 (0.004) & 0.909 (0.004) & 0.909 (0.004) & 0.908 (0.004) & 0.980 (0.002) \\
\multicolumn{6}{l}{\textbf{Width of prediction interval}} \\
\bottomrule
{$\mu$} & ModSel-CP & ModSel-CP-LOO & YK-baseline & YK-split & YK-adjust \\
\hline
$-1.0$ & 3.701 (0.003) & 3.701 (0.003) & 3.701 (0.003) & 3.721 (0.004) & 5.240 (0.006) \\
$-0.5$ & 3.312 (0.003) & 3.312 (0.003) & 3.312 (0.003) & 3.332 (0.004) & 4.729 (0.005) \\
$0.0$ & 3.722 (0.006) & 3.699 (0.004) & 3.582 (0.003) & 3.726 (0.005) & 5.023 (0.005) \\
$0.5$ & 3.312 (0.003) & 3.312 (0.003) & 3.312 (0.003) & 3.337 (0.004) & 4.732 (0.006) \\
$1.0$ & 3.707 (0.003) & 3.707 (0.003) & 3.707 (0.003) & 3.728 (0.004) & 5.250 (0.006) \\
\bottomrule
\end{tabular}
\label{tab:two_model_sim}
\end{table}

\begin{table}[h]
    \centering
\begin{tabular}{cccccc}
\multicolumn{6}{c}{\textbf{$C = 1$, $n = 200$}}\\
\multicolumn{6}{l}{\textbf{Coverage}} \\
\bottomrule
{$\mu$} & ModSel-CP & ModSel-CP-LOO & YK-baseline & YK-split & YK-adjust \\
\hline
$-1.0$ & 0.903 (0.004) & 0.903 (0.004) & 0.903 (0.004) & 0.900 (0.004) & 0.978 (0.002) \\
$-0.5$ & 0.910 (0.004) & 0.910 (0.004) & 0.910 (0.004) & 0.912 (0.004) & 0.985 (0.002) \\
$0.0$ & 0.899 (0.004) & 0.899 (0.004) & 0.886 (0.005) & 0.898 (0.004) & 0.973 (0.002) \\
$0.5$ & 0.901 (0.004) & 0.901 (0.004) & 0.901 (0.004) & 0.901 (0.004) & 0.980 (0.002) \\
$1.0$ & 0.895 (0.004) & 0.895 (0.004) & 0.895 (0.004) & 0.895 (0.004) & 0.980 (0.002) \\
\multicolumn{6}{l}{\textbf{Width of prediction interval}} \\
\bottomrule
{$\mu$} & ModSel-CP & ModSel-CP-LOO & YK-baseline & YK-split & YK-adjust \\
\hline
$-1.0$ & 3.313 (0.003) & 3.313 (0.003) & 3.313 (0.003) & 3.340 (0.004) & 4.736 (0.005) \\
$-0.5$ & 3.703 (0.003) & 3.703 (0.003) & 3.703 (0.003) & 3.724 (0.005) & 5.244 (0.006) \\
$0.0$ & 4.711 (0.011) & 4.592 (0.005) & 4.453 (0.003) & 4.621 (0.005) & 5.964 (0.005) \\
$0.5$ & 3.703 (0.003) & 3.703 (0.003) & 3.703 (0.003) & 3.726 (0.005) & 5.245 (0.006) \\
$1.0$ & 3.312 (0.003) & 3.312 (0.003) & 3.312 (0.003) & 3.334 (0.004) & 4.744 (0.005) \\
\bottomrule
\end{tabular}
\end{table}

\begin{table}[h]
    \centering
\begin{tabular}{cccccc}
\multicolumn{6}{c}{\textbf{$C = 5$, $n = 200$}}\\
\multicolumn{6}{l}{\textbf{Coverage}} \\
\bottomrule
{$\mu$} & ModSel-CP & ModSel-CP-LOO & YK-baseline & YK-split & YK-adjust \\
\hline
$-1.0$ & 0.904 (0.004) & 0.904 (0.004) & 0.904 (0.004) & 0.904 (0.004) & 0.975 (0.002) \\
$-0.5$ & 0.894 (0.004) & 0.894 (0.004) & 0.894 (0.004) & 0.895 (0.004) & 0.980 (0.002) \\
$0.0$ & 0.892 (0.004) & 0.893 (0.004) & 0.880 (0.005) & 0.899 (0.004) & 0.970 (0.002) \\
$0.5$ & 0.903 (0.004) & 0.903 (0.004) & 0.903 (0.004) & 0.905 (0.004) & 0.981 (0.002) \\
$1.0$ & 0.896 (0.004) & 0.896 (0.004) & 0.896 (0.004) & 0.897 (0.004) & 0.981 (0.002) \\
\multicolumn{6}{l}{\textbf{Width of prediction interval}} \\
\bottomrule
{$\mu$} & ModSel-CP & ModSel-CP-LOO & YK-baseline & YK-split & YK-adjust \\
\hline
$-1.0$ & 10.584 (0.003) & 10.584 (0.003) & 10.584 (0.003) & 10.609 (0.005) & 12.197 (0.006) \\
$-0.5$ & 11.582 (0.003) & 11.582 (0.003) & 11.582 (0.003) & 11.612 (0.005) & 13.196 (0.006) \\
$0.0$ & 13.658 (0.047) & 12.581 (0.005) & 12.446 (0.003) & 12.613 (0.005) & 13.964 (0.005) \\
$0.5$ & 11.591 (0.003) & 11.591 (0.003) & 11.591 (0.003) & 11.620 (0.005) & 13.199 (0.006) \\
$1.0$ & 10.586 (0.003) & 10.586 (0.003) & 10.586 (0.003) & 10.610 (0.005) & 12.200 (0.006) \\
\bottomrule
\end{tabular}
\end{table}

\clearpage

\section{Additional simulations}\label{sec:appendxi_addition_sim}
\subsection{Regression: rescaled residual score}\label{sec:appendix_sim_rescaled_residual}
Under the same setting as in Section~\ref{sec:experiment}, we also conducted simulations with a rescaled residual score model class,
\[\left\{\left|\frac{y-x^\top\hat\theta_\lambda}{\hat\sigma_\lambda(x)}\right| : \lambda\in\Lambda\right\},\]
where $\hat\theta_\lambda$ is fitted as above but using only $n_{\textnormal{train}}/2$ many data points, and then use the remaining half of the pretraining data to fit a random forest to data pairs $(X_i,|Y_i - X_i^\top\hat\theta_\lambda|)$ in order to estimate $\hat\sigma_\lambda(x)$.  We train the random forest by applying \texttt{RandomForestRegressor} function in Python's \texttt{sklearn} package, leaving all parameters at their default values. If needed, we break ties in model selection deterministically by selecting the model with the smallest index. 

The results of the simulations are displayed in Figure~\ref{fig:resresidual_n100} for varying $|\Lambda|$ with $n=100$, and Figure~\ref{fig:resresidual_m200} for varying $n$ with $|\Lambda|=200$. 

\begin{figure}
\centering
\includegraphics[scale = 0.3]{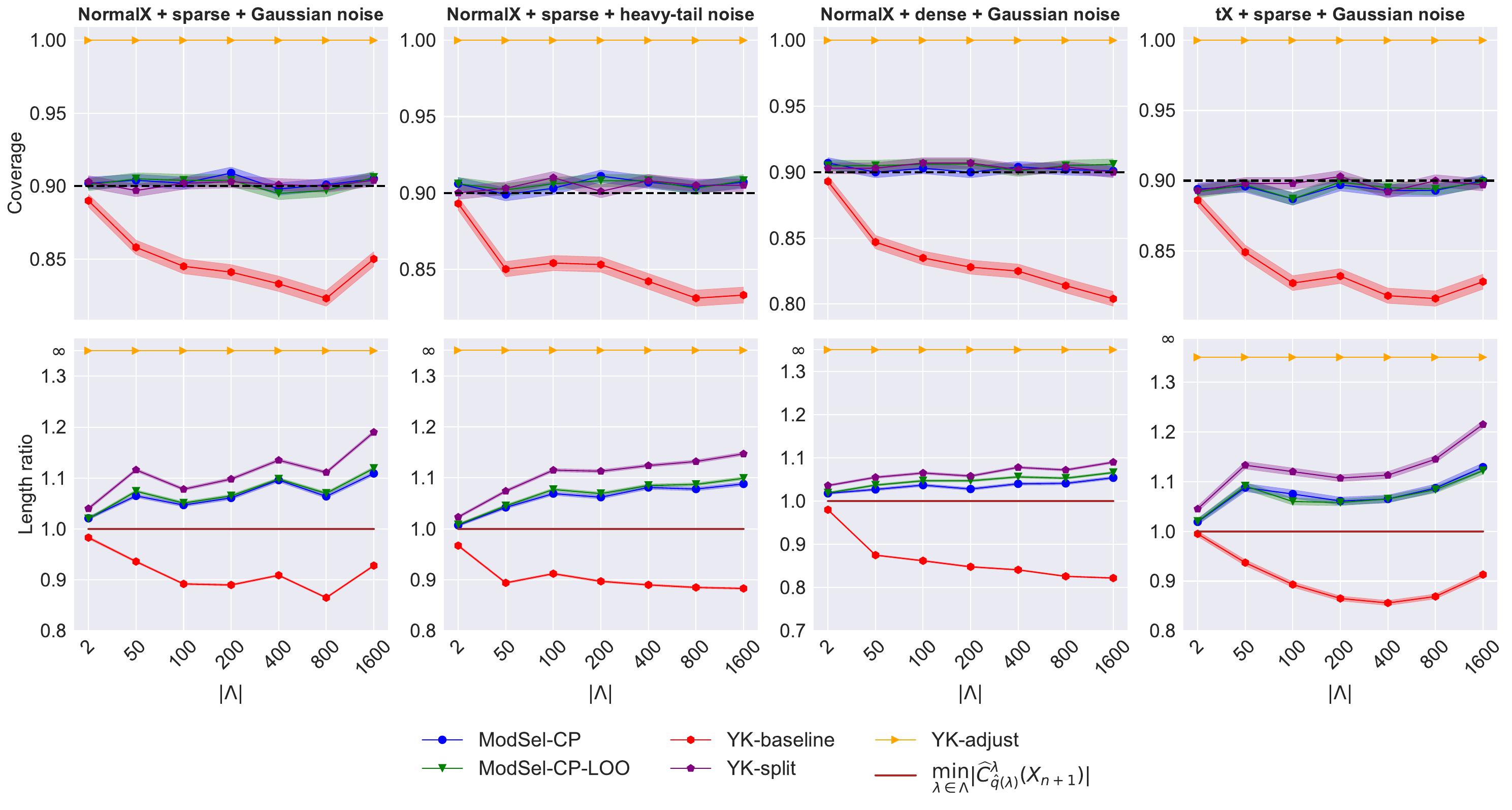}
\caption{Rescaled residual score, $n = 100$, varying $|\Lambda|$.} 
\label{fig:resresidual_n100}

\vspace{0.8cm}
\includegraphics[scale = 0.3]{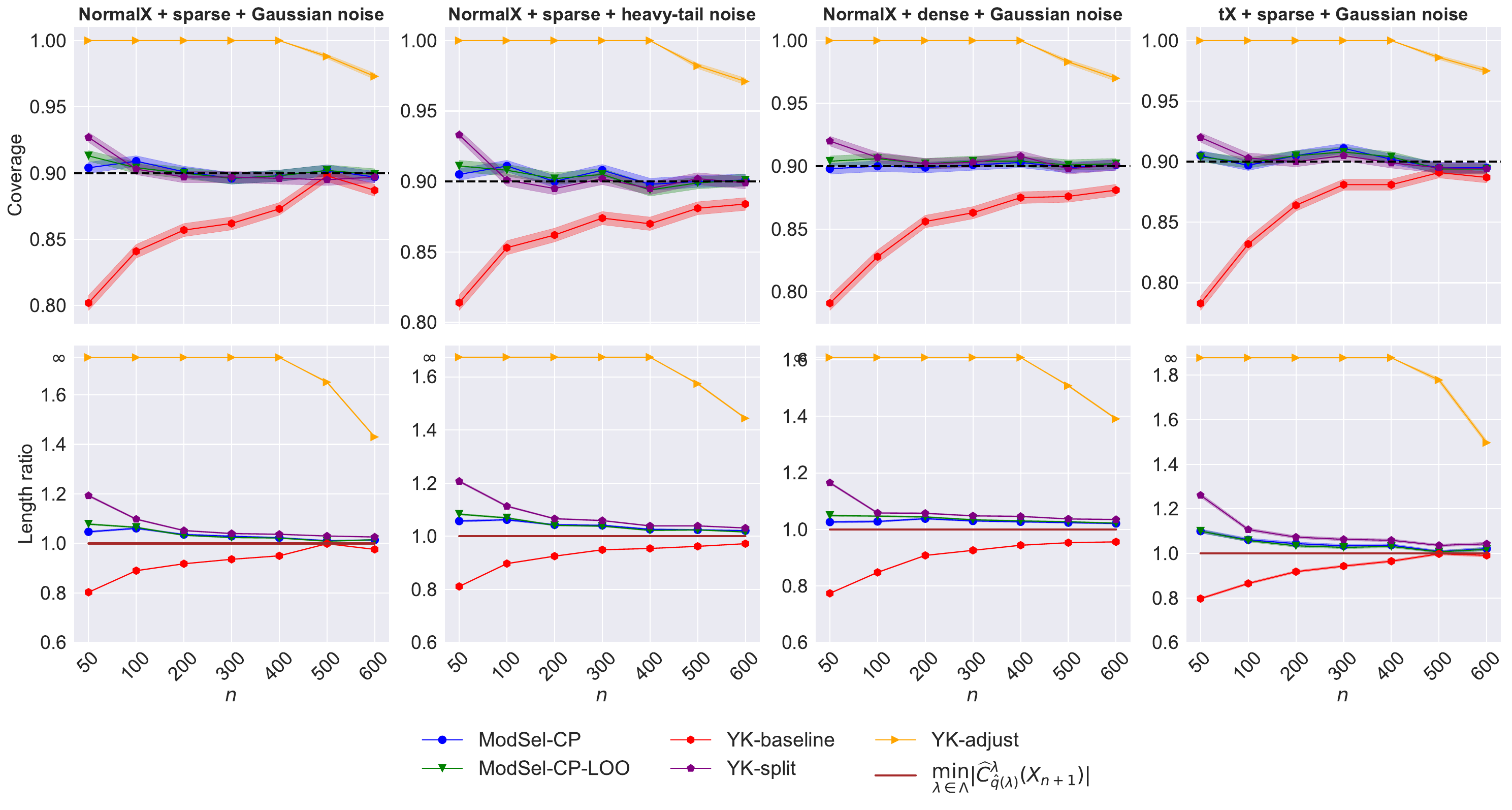}
\caption{Rescaled residual score, $ |\Lambda| = 200$, varying $n$.} 
\label{fig:resresidual_m200}

\end{figure}

\subsection{Classification}\label{sec:appendix_sims_classification}
In this section, we perform additional simulated data experiments to see how our methods behave in the setting of classification (rather than regression, as examined earlier in Section~\ref{sec:experiment}). In this setting, as described earlier in~\eqref{eqn:loss_classification}, we use the conditional probability score, and the loss $\mathcal{L}$ is defined as the average cardinality of the resulting prediction sets. The details of our experiment are as follows. (This simulation setting is a slight modification of the experiment setting in \citet{romano2020classification}'s work.)
\begin{itemize}
    \item \textbf{Distribution of $(X,Y)$:} for feature dimension $d = 50$, the distribution $P$ for data points $(X,Y)$ is defined as follows. First the feature $X\in\mathbb{R}^d$ is generated as  \[X_{1} =1, \textrm{ w.p. } 1/5; X_{1} = -8, \textrm{ otherwise }\] and \[X_{2},\dots,X_{d} \overset{\textnormal{i.i.d.}}{\sim} \mathcal{N}(0,1).\]
    Then, the conditional distribution of $Y\mid X$ follows the multinomial distribution with weights \[w_{j}(X) = \frac{\mathrm{exp}(X^{\top}\beta_{j})}{\sum_{j'=1}^{K}\mathrm{exp}(X^{\top}\beta_{j'})}, \ j = 1,\dots,K.\]
    where $\beta_{j} \sim \mathcal{N}(0, \mathbf{I}_{d})$, and where $K=10$ is the number of classes.

    \item \textbf{Models training (conditional density score class):} our model class $\Lambda$ is given by \[\Lambda =  \left\{S^\lambda(x,y) =  -p_{\lambda}(y\mid x): \lambda \in \Lambda\right\},\] where each $p_\lambda(y\mid x)$ is a pretrained estimate of the conditional probability of $Y=y$ given $X=x$, trained by applying \texttt{ensemble.RandomForestClassifier} function in Python's package \texttt{sklearn}. To construct our set of candidate models $\Lambda$, for each desired size $|\Lambda|$ of this class we first vary the number of base estimators \texttt{n$\_$estimators} to $\frac{|\Lambda|}{2}$ different values varying from $10$ to $100$ as an equispaced sequence during the training, and then set the parameter \texttt{criterion} to be \texttt{`gini'} or \texttt{`entropy'}, leading to $\frac{|\Lambda|}{2} \cdot 2 = |\Lambda|$ many candidate models. All other parameters at their default values.
    \item Pretraining is carried out on a sample of size $n_{\textnormal{train}}=300$, and then the prediction sets are constructed using an independent sample of $n=150$ data points. If needed, we break ties in model selection deterministically by selecting the model with the smallest index.
\end{itemize}
The results of the experiment are displayed with standard error bars in Figure~\ref{fig:classification} and Table~\ref{table:classification}, averaged over 5000 independent trials of the simulation. As expected, we see that \texttt{YK-baseline} has some loss of coverage, but its corrected version \texttt{YK-adjust} is overly conservative. The remaining methods, \texttt{YK-split}, \texttt{ModSel-CP}, and \texttt{ModSel-CP-LOO}, all achieve the desired coverage level, and return prediction sets of similar widths.

\begin{figure}[t]
     \centering
     \includegraphics[scale = 0.22]{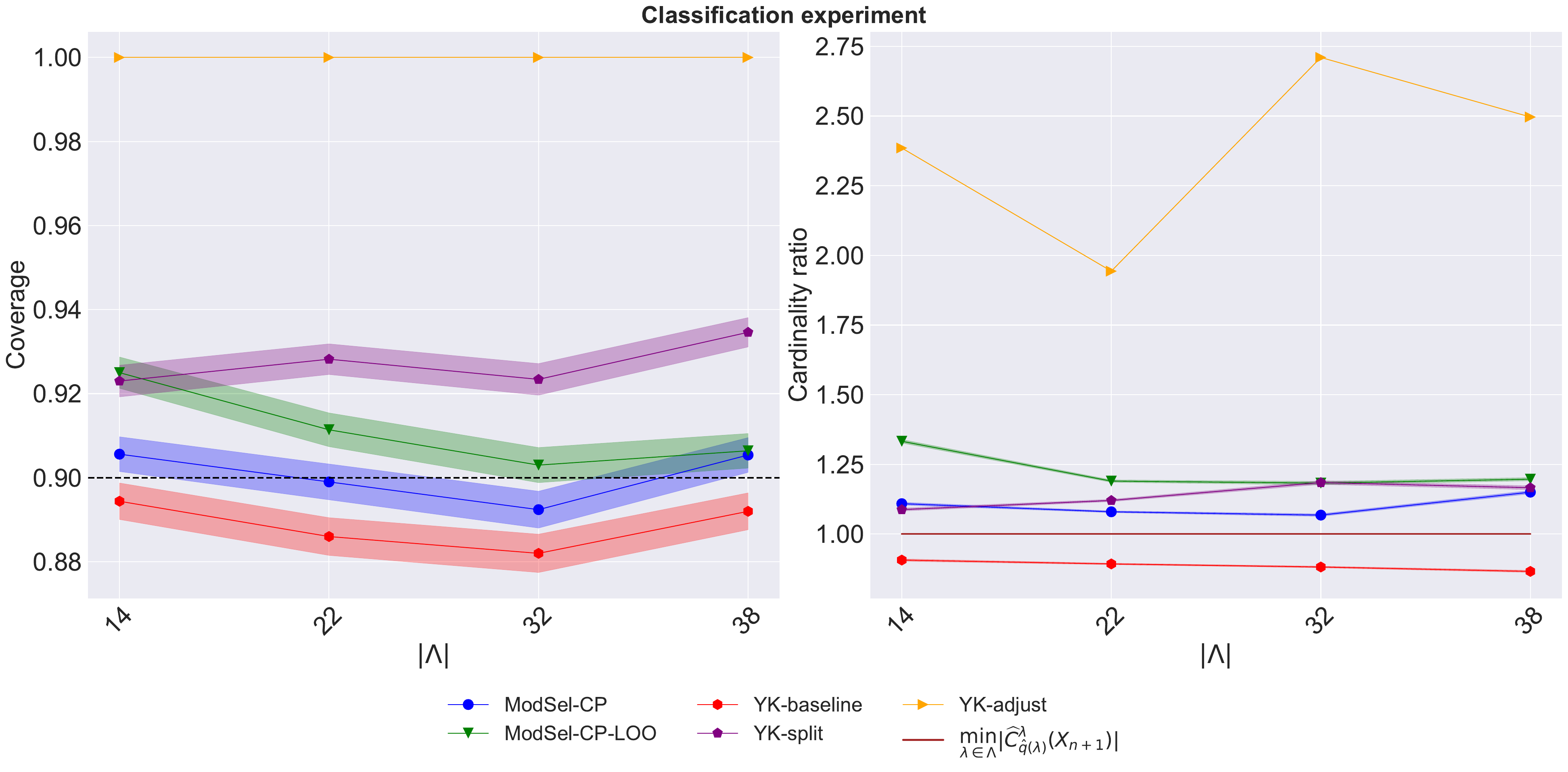}
     \caption{Results for the classification experiment (see Section~\ref{sec:appendix_sims_classification} for details).}
     \label{fig:classification}
 \end{figure}

\begin{table}[ht]
    \caption{Results of the classification experiment}
    \label{table:classification}
    \centering
\begin{tabular}{lccccc}
\multicolumn{6}{c}{\textbf{$n = 150$, varying $|\Lambda|$}}\\
\multicolumn{6}{l}{\textbf{Coverage}} \\
\bottomrule
{$|\Lambda|$} & ModSel-CP & ModSel-CP-LOO & YK-baseline & YK-split & YK-adjust  \\
\hline
$14$ & 0.906 (0.004) & 0.925 (0.004) & 0.894 (0.004) & 0.923 (0.004) & 1.000 (0.000)  \\
$22$ & 0.899 (0.004) & 0.911 (0.004) & 0.886 (0.004) & 0.928 (0.004) & 1.000 (0.000)  \\
$32$ & 0.892 (0.004) & 0.903 (0.004) & 0.882 (0.005) & 0.923 (0.004) & 1.000 (0.000)  \\
$38$ & 0.905 (0.004) & 0.906 (0.004) & 0.892 (0.004) & 0.935 (0.003) & 1.000 (0.000)  \\
\multicolumn{6}{l}{\textbf{Width ratio}} \\
\bottomrule
{$|\Lambda|$} & ModSel-CP & ModSel-CP-LOO & YK-baseline & YK-split & YK-adjust  \\
\hline
$14$ & 1.109 (0.005) & 1.333 (0.006) & 0.907 (0.004) & 1.088 (0.005) & 2.385 (0.000) \\
$22$ & 1.080 (0.003) & 1.190 (0.004) & 0.893 (0.003) & 1.121 (0.004) & 1.943 (0.000)  \\
$32$ & 1.068 (0.005) & 1.183 (0.006) & 0.882 (0.003) & 1.185 (0.007) & 2.710 (0.000)  \\
$38$ & 1.150 (0.006) & 1.197 (0.006) & 0.866 (0.004) & 1.166 (0.007) & 2.496 (0.000)  \\
\bottomrule
\end{tabular}
\end{table}

\clearpage

\section{Table summary of simulation results (regression)}\label{app:table_summary_simulation}

\begin{table}[H]
    \caption{ Residual score,  NormalX + sparse + Gaussian noise.}
    \centering
\begin{tabular}{lccccc}
\multicolumn{6}{c}{\textbf{$n = 100$, varying $|\Lambda|$}}\\
\multicolumn{6}{l}{\textbf{Coverage}} \\
\bottomrule
{$|\Lambda|$} & ModSel-CP & ModSel-CP-LOO & YK-baseline & YK-split & YK-adjust  \\
\hline
$2$ & 0.907 (0.004) & 0.909 (0.004) & 0.899 (0.004) & 0.904 (0.004) & 1.000 (0.000)  \\
$50$ & 0.902 (0.004) & 0.895 (0.004) & 0.859 (0.005) & 0.908 (0.004) & 1.000 (0.000)  \\
$100$ & 0.902 (0.004) & 0.902 (0.004) & 0.849 (0.005) & 0.909 (0.004) & 1.000 (0.000)  \\
$200$ & 0.909 (0.004) & 0.914 (0.004) & 0.865 (0.005) & 0.908 (0.004) & 1.000 (0.000) \\
$400$ & 0.899 (0.004) & 0.898 (0.004) & 0.839 (0.005) & 0.903 (0.004) & 1.000 (0.000) \\
$800$ & 0.899 (0.004) & 0.900 (0.004) & 0.828 (0.005) & 0.896 (0.004) & 1.000 (0.000)  \\
$1600$ & 0.909 (0.004) & 0.908 (0.004) & 0.837 (0.005) & 0.908 (0.004) & 1.000 (0.000)  \\
\multicolumn{6}{l}{\textbf{Width ratio}} \\
\bottomrule
{$|\Lambda|$} & ModSel-CP & ModSel-CP-LOO & YK-baseline & YK-split & YK-adjust  \\
\hline
$2$ & 1.019 (0.002) & 1.017 (0.001) & 0.987 (0.001) & 1.035 (0.002) & $\infty$  \\
$50$ & 1.051 (0.003) & 1.049 (0.002) & 0.931 (0.001) & 1.082 (0.002) & $\infty$  \\
$100$ & 1.065 (0.003) & 1.053 (0.002) & 0.915 (0.001) & 1.088 (0.002) & $\infty$  \\
$200$ & 1.063 (0.003) & 1.060 (0.002) & 0.931 (0.001) & 1.108 (0.002) & $\infty$  \\
$400$ & 1.056 (0.003) & 1.050 (0.002) & 0.888 (0.001) & 1.088 (0.002) & $\infty$ \\
$800$ & 1.051 (0.003) & 1.043 (0.002) & 0.877 (0.001) & 1.079 (0.002) & $\infty$  \\
$1600$ & 1.104 (0.003) & 1.099 (0.002) & 0.908 (0.001) & 1.142 (0.002) & $\infty$  \\
\bottomrule
\end{tabular}
\end{table}

\begin{table}[ht]
    \centering
\begin{tabular}{lccccc}
\multicolumn{6}{c}{\textbf{$|\Lambda| = 200$, varying $n$}}\\
\multicolumn{6}{l}{\textbf{Coverage}} \\
\bottomrule
{$n$} & ModSel-CP & ModSel-CP-LOO & YK-baseline & YK-split & YK-adjust \\
\hline
$50$ & 0.899 (0.004) & 0.898 (0.004) & 0.811 (0.006) & 0.920 (0.004) & 1.000 (0.000)  \\
$100$ & 0.904 (0.004) & 0.902 (0.004) & 0.855 (0.005) & 0.904 (0.004) & 1.000 (0.000)  \\
$200$ & 0.893 (0.004) & 0.894 (0.004) & 0.867 (0.005) & 0.898 (0.004) & 1.000 (0.000)  \\
$300$ & 0.900 (0.004) & 0.899 (0.004) & 0.878 (0.005) & 0.897 (0.004) & 1.000 (0.000)  \\
$400$ & 0.903 (0.004) & 0.906 (0.004) & 0.897 (0.004) & 0.904 (0.004) & 1.000 (0.000)  \\
$500$ & 0.907 (0.004) & 0.907 (0.004) & 0.889 (0.004) & 0.903 (0.004) & 0.982 (0.002)  \\
$600$ & 0.901 (0.004) & 0.900 (0.004) & 0.898 (0.004) & 0.899 (0.004) & 0.985 (0.002)  \\
\multicolumn{6}{l}{\textbf{Width ratio}} \\
\bottomrule
{$n$} & ModSel-CP & ModSel-CP-LOO & YK-baseline & YK-split & YK-adjust \\
\hline
$50$ & 1.068 (0.003) & 1.068 (0.002) & 0.845 (0.001) & 1.194 (0.003) & $\infty$  \\
$100$ & 1.074 (0.003) & 1.072 (0.002) & 0.939 (0.001) & 1.121 (0.002) & $\infty$  \\
$200$ & 1.048 (0.002) & 1.045 (0.001) & 0.964 (0.001) & 1.081 (0.002) & $\infty$  \\
$300$ & 1.038 (0.002) & 1.034 (0.001) & 0.971 (0.001) & 1.060 (0.001) & $\infty$  \\
$400$ & 1.017 (0.001) & 1.018 (0.001) & 0.996 (0.001) & 1.050 (0.001) & $\infty$  \\
$500$ & 1.021 (0.002) & 1.017 (0.001) & 0.969 (0.000) & 1.031 (0.001) & 1.521 (0.001)\\
$600$ & 1.008 (0.001) & 1.006 (0.001) & 0.998 (0.000) & 1.018 (0.001) & 1.453 (0.001)  \\
\bottomrule
\end{tabular}
\end{table}


\begin{table}[ht]
    \caption{  Residual score,  NormalX + sparse + heavy-tail noise.}
    \centering
\begin{tabular}{lccccc}
\multicolumn{6}{c}{\textbf{$n = 100$, varying $|\Lambda|$}}\\
\multicolumn{6}{l}{\textbf{Coverage}} \\
\bottomrule
{$|\Lambda|$} & ModSel-CP & ModSel-CP-LOO & YK-baseline & YK-split & YK-adjust   \\
\hline
$2$ & 0.898 (0.004) & 0.901 (0.004) & 0.886 (0.004) & 0.906 (0.004) & 1.000 (0.000)  \\
$50$ & 0.905 (0.004) & 0.906 (0.004) & 0.859 (0.005) & 0.906 (0.004) & 1.000 (0.000)  \\
$100$ & 0.901 (0.004) & 0.904 (0.004) & 0.851 (0.005) & 0.912 (0.004) & 1.000 (0.000)  \\
$200$ & 0.902 (0.004) & 0.905 (0.004) & 0.846 (0.005) & 0.898 (0.004) & 1.000 (0.000)  \\
$400$ & 0.897 (0.004) & 0.898 (0.004) & 0.834 (0.005) & 0.894 (0.004) & 1.000 (0.000) \\
$800$ & 0.899 (0.004) & 0.902 (0.004) & 0.826 (0.005) & 0.908 (0.004) & 1.000 (0.000)  \\
$1600$ & 0.899 (0.004) & 0.899 (0.004) & 0.843 (0.005) & 0.898 (0.004) & 1.000 (0.000)  \\
\multicolumn{6}{l}{\textbf{Width ratio}} \\
\bottomrule
{$|\Lambda|$} & ModSel-CP & ModSel-CP-LOO & YK-baseline & YK-split & YK-adjust   \\
\hline
$2$ & 1.010 (0.002) & 1.010 (0.001) & 0.976 (0.001) & 1.030 (0.002) & $\infty$  \\
$50$ & 1.040 (0.003) & 1.036 (0.002) & 0.917 (0.001) & 1.069 (0.002) & $\infty$  \\
$100$ & 1.041 (0.003) & 1.038 (0.002) & 0.900 (0.001) & 1.072 (0.002) & $\infty$ \\
$200$ & 1.037 (0.003) & 1.035 (0.002) & 0.885 (0.001) & 1.070 (0.002) & $\infty$ \\
$400$ & 1.073 (0.003) & 1.070 (0.002) & 0.909 (0.001) & 1.107 (0.002) & $\infty$  \\
$800$ & 1.066 (0.003) & 1.059 (0.002) & 0.874 (0.001) & 1.100 (0.002) & $\infty$  \\
$1600$ & 1.110 (0.003) & 1.109 (0.002) & 0.936 (0.001) & 1.164 (0.002) & $\infty$  \\
\bottomrule
\end{tabular}
\end{table}

\begin{table}[ht]
    \centering
\begin{tabular}{lccccc}
\multicolumn{6}{c}{\textbf{$|\Lambda| = 200$, varying $n$}}\\
\multicolumn{6}{l}{\textbf{Coverage}} \\
\bottomrule
{$n$} & ModSel-CP & ModSel-CP-LOO & YK-baseline & YK-split & YK-adjust   \\
\hline
$50$ & 0.905 (0.004) & 0.904 (0.004) & 0.814 (0.006) & 0.926 (0.004) & 1.000 (0.000)  \\
$100$ & 0.906 (0.004) & 0.903 (0.004) & 0.846 (0.005) & 0.900 (0.004) & 1.000 (0.000)  \\
$200$ & 0.900 (0.004) & 0.900 (0.004) & 0.863 (0.005) & 0.904 (0.004) & 1.000 (0.000)  \\
$300$ & 0.898 (0.004) & 0.896 (0.004) & 0.880 (0.005) & 0.899 (0.004) & 1.000 (0.000)  \\
$400$ & 0.905 (0.004) & 0.906 (0.004) & 0.887 (0.004) & 0.906 (0.004) & 1.000 (0.000)  \\
$500$ & 0.896 (0.004) & 0.897 (0.004) & 0.880 (0.005) & 0.901 (0.004) & 0.982 (0.002)  \\
$600$ & 0.911 (0.004) & 0.910 (0.004) & 0.893 (0.004) & 0.907 (0.004) & 0.971 (0.002) \\
\multicolumn{6}{l}{\textbf{Width ratio}} \\
\bottomrule
{$n$} & ModSel-CP & ModSel-CP-LOO & YK-baseline & YK-split & YK-adjust   \\
\hline
$50$ & 1.075 (0.003) & 1.071 (0.002) & 0.842 (0.001) & 1.202 (0.003) & $\infty$  \\
$100$ & 1.065 (0.003) & 1.057 (0.002) & 0.909 (0.001) & 1.093 (0.002) & $\infty$  \\
$200$ & 1.038 (0.002) & 1.036 (0.001) & 0.940 (0.001) & 1.060 (0.001) & $\infty$  \\
$300$ & 1.034 (0.002) & 1.032 (0.001) & 0.982 (0.001) & 1.061 (0.001) & $\infty$  \\
$400$ & 1.027 (0.002) & 1.027 (0.001) & 0.972 (0.000) & 1.048 (0.001) & $\infty$  \\
$500$ & 1.017 (0.002) & 1.017 (0.001) & 0.963 (0.000) & 1.029 (0.001) & 1.565 (0.001)  \\
$600$ & 1.021 (0.002) & 1.016 (0.001) & 0.964 (0.000) & 1.027 (0.001) & 1.410 (0.001)  \\
\bottomrule
\end{tabular}
\end{table}

\begin{table}[ht]
    \caption{  Residual score,  NormalX + dense + Gaussian noise.}
    \centering
\begin{tabular}{lccccc}
\multicolumn{6}{c}{\textbf{$n = 100$, varying $|\Lambda|$}}\\
\multicolumn{6}{l}{\textbf{Coverage}} \\
\bottomrule
{$|\Lambda|$} & ModSel-CP & ModSel-CP-LOO & YK-baseline & YK-split & YK-adjust   \\
\hline
$2$ & 0.896 (0.004) & 0.896 (0.004) & 0.883 (0.005) & 0.894 (0.004) & 1.000 (0.000)  \\
$50$ & 0.901 (0.004) & 0.900 (0.004) & 0.850 (0.005) & 0.903 (0.004) & 1.000 (0.000) \\
$100$ & 0.903 (0.004) & 0.901 (0.004) & 0.845 (0.005) & 0.898 (0.004) & 1.000 (0.000)  \\
$200$ & 0.904 (0.004) & 0.896 (0.004) & 0.829 (0.005) & 0.902 (0.004) & 1.000 (0.000)  \\
$400$ & 0.905 (0.004) & 0.901 (0.004) & 0.833 (0.005) & 0.905 (0.004) & 1.000 (0.000)  \\
$800$ & 0.898 (0.004) & 0.898 (0.004) & 0.818 (0.005) & 0.893 (0.004) & 1.000 (0.000)  \\
$1600$ & 0.899 (0.004) & 0.898 (0.004) & 0.815 (0.005) & 0.899 (0.004) & 1.000 (0.000)  \\
\multicolumn{6}{l}{\textbf{Width ratio}} \\
\bottomrule
{$|\Lambda|$} & ModSel-CP & ModSel-CP-LOO & YK-baseline & YK-split & YK-adjust   \\
\hline
$2$ & 1.012 (0.002) & 1.009 (0.001) & 0.974 (0.001) & 1.023 (0.002) & $\infty$  \\
$50$ & 1.021 (0.002) & 1.021 (0.002) & 0.884 (0.001) & 1.038 (0.002) & $\infty$ \\
$100$ & 1.019 (0.003) & 1.016 (0.002) & 0.864 (0.001) & 1.037 (0.002) & $\infty$  \\
$200$ & 1.025 (0.003) & 1.022 (0.002) & 0.856 (0.001) & 1.044 (0.002) & $\infty$  \\
$400$ & 1.016 (0.003) & 1.018 (0.002) & 0.843 (0.001) & 1.036 (0.002) & $\infty$  \\
$800$ & 1.022 (0.003) & 1.020 (0.001) & 0.839 (0.001) & 1.042 (0.002) & $\infty$  \\
$1600$ & 1.025 (0.003) & 1.021 (0.002) & 0.822 (0.001) & 1.048 (0.002) & $\infty$  \\
\bottomrule
\end{tabular}
\end{table}

\begin{table}[ht]
    \centering
\begin{tabular}{lccccc}
\multicolumn{6}{c}{\textbf{$|\Lambda| = 200$, varying $n$}}\\
\multicolumn{6}{l}{\textbf{Coverage}} \\
\bottomrule
{$n$} & ModSel-CP & ModSel-CP-LOO & YK-baseline & YK-split & YK-adjust   \\
\hline
$50$ & 0.897 (0.004) & 0.904 (0.004) & 0.800 (0.006) & 0.919 (0.004) & 1.000 (0.000) \\
$100$ & 0.894 (0.004) & 0.895 (0.004) & 0.823 (0.005) & 0.900 (0.004) & 1.000 (0.000) \\
$200$ & 0.895 (0.004) & 0.895 (0.004) & 0.851 (0.005) & 0.900 (0.004) & 1.000 (0.000)  \\
$300$ & 0.905 (0.004) & 0.904 (0.004) & 0.866 (0.005) & 0.897 (0.004) & 1.000 (0.000)  \\
$400$ & 0.904 (0.004) & 0.902 (0.004) & 0.872 (0.005) & 0.903 (0.004) & 1.000 (0.000)  \\
$500$ & 0.903 (0.004) & 0.904 (0.004) & 0.875 (0.005) & 0.900 (0.004) & 0.982 (0.002)  \\
$600$ & 0.904 (0.004) & 0.901 (0.004) & 0.879 (0.005) & 0.908 (0.004) & 0.975 (0.002)  \\
\multicolumn{6}{l}{\textbf{Width ratio}} \\
\bottomrule
{$n$} & ModSel-CP & ModSel-CP-LOO & YK-baseline & YK-split & YK-adjust   \\
\hline
$50$ & 1.021 (0.003) & 1.025 (0.002) & 0.801 (0.001) & 1.138 (0.003) & $\infty$  \\
$100$ & 1.015 (0.003) & 1.015 (0.002) & 0.849 (0.001) & 1.035 (0.002) & $\infty$  \\
$200$ & 1.016 (0.002) & 1.016 (0.001) & 0.899 (0.001) & 1.028 (0.001) & $\infty$ \\
$300$ & 1.019 (0.002) & 1.015 (0.001) & 0.921 (0.001) & 1.023 (0.001) & $\infty$  \\
$400$ & 1.021 (0.002) & 1.017 (0.001) & 0.937 (0.000) & 1.024 (0.001) & $\infty$  \\
$500$ & 1.028 (0.002) & 1.023 (0.001) & 0.949 (0.000) & 1.030 (0.001) & 1.467 (0.001) \\
$600$ & 1.023 (0.002) & 1.021 (0.001) & 0.955 (0.000) & 1.027 (0.001) & 1.360 (0.001)  \\
\bottomrule
\end{tabular}
\end{table}

\begin{table}[ht]
    \caption{  Residual score,  tX + sparse + Gaussian noise.}
    \centering
\begin{tabular}{lccccc}
\multicolumn{6}{c}{\textbf{$n = 100$, varying $|\Lambda|$}}\\
\multicolumn{6}{l}{\textbf{Coverage}} \\
\bottomrule
{$|\Lambda|$} & ModSel-CP & ModSel-CP-LOO & YK-baseline & YK-split & YK-adjust   \\
\hline
$2$ & 0.904 (0.004) & 0.897 (0.004) & 0.889 (0.004) & 0.902 (0.004) & 1.000 (0.000)  \\
$50$ & 0.897 (0.004) & 0.897 (0.004) & 0.848 (0.005) & 0.901 (0.004) & 1.000 (0.000)  \\
$100$ & 0.903 (0.004) & 0.905 (0.004) & 0.844 (0.005) & 0.904 (0.004) & 1.000 (0.000)  \\
$200$ & 0.903 (0.004) & 0.906 (0.004) & 0.847 (0.005) & 0.907 (0.004) & 1.000 (0.000)  \\
$400$ & 0.906 (0.004) & 0.903 (0.004) & 0.847 (0.005) & 0.902 (0.004) & 1.000 (0.000) \\
$800$ & 0.900 (0.004) & 0.897 (0.004) & 0.829 (0.005) & 0.912 (0.004) & 1.000 (0.000)  \\
$1600$ & 0.894 (0.004) & 0.899 (0.004) & 0.825 (0.005) & 0.901 (0.004) & 1.000 (0.000)  \\
\multicolumn{6}{l}{\textbf{Width ratio}} \\
\bottomrule
{$|\Lambda|$} & ModSel-CP & ModSel-CP-LOO & YK-baseline & YK-split & YK-adjust   \\
\hline
$2$ & 1.006 (0.003) & 1.017 (0.002) & 0.961 (0.002) & 1.050 (0.003) & $\infty$  \\
$50$ & 1.007 (0.003) & 1.052 (0.003) & 0.863 (0.001) & 1.101 (0.003) & $\infty$  \\
$100$ & 1.006 (0.004) & 1.058 (0.003) & 0.842 (0.001) & 1.104 (0.003) & $\infty$  \\
$200$ & 1.051 (0.004) & 1.115 (0.003) & 0.878 (0.001) & 1.176 (0.004) & $\infty$  \\
$400$ & 1.091 (0.004) & 1.150 (0.003) & 0.900 (0.001) & 1.233 (0.004) & $\infty$  \\
$800$ & 1.059 (0.004) & 1.122 (0.003) & 0.842 (0.001) & 1.196 (0.004) & $\infty$  \\
$1600$ & 1.061 (0.004) & 1.129 (0.003) & 0.851 (0.001) & 1.205 (0.004) & $\infty$  \\
\bottomrule
\end{tabular}
\end{table}
\begin{table}[ht]
    \centering
\begin{tabular}{lccccc}
\multicolumn{6}{c}{\textbf{$|\Lambda| = 200$, varying $n$}}\\
\multicolumn{6}{l}{\textbf{Coverage}} \\
\bottomrule
{$n$} & ModSel-CP & ModSel-CP-LOO & YK-baseline & YK-split & YK-adjust   \\
\hline
$50$ & 0.904 (0.004) & 0.911 (0.004) & 0.806 (0.006) & 0.925 (0.004) & 1.000 (0.000)  \\
$100$ & 0.903 (0.004) & 0.901 (0.004) & 0.849 (0.005) & 0.907 (0.004) & 1.000 (0.000)  \\
$200$ & 0.901 (0.004) & 0.901 (0.004) & 0.866 (0.005) & 0.902 (0.004) & 1.000 (0.000)  \\
$300$ & 0.900 (0.004) & 0.897 (0.004) & 0.873 (0.005) & 0.898 (0.004) & 1.000 (0.000)  \\
$400$ & 0.901 (0.004) & 0.903 (0.004) & 0.881 (0.005) & 0.902 (0.004) & 1.000 (0.000)  \\
$500$ & 0.897 (0.004) & 0.894 (0.004) & 0.875 (0.005) & 0.898 (0.004) & 0.982 (0.002)  \\
$600$ & 0.897 (0.004) & 0.895 (0.004) & 0.880 (0.005) & 0.901 (0.004) & 0.970 (0.002)  \\
\multicolumn{6}{l}{\textbf{Width ratio}} \\
\bottomrule
{$n$} & ModSel-CP & ModSel-CP-LOO & YK-baseline & YK-split & YK-adjust   \\
\hline
$50$ & 0.970 (0.004) & 1.069 (0.003) & 0.736 (0.001) & 1.316 (0.007) & $\infty$  \\
$100$ & 1.084 (0.004) & 1.139 (0.003) & 0.911 (0.001) & 1.216 (0.004) & $\infty$  \\
$200$ & 1.042 (0.004) & 1.068 (0.002) & 0.927 (0.001) & 1.117 (0.002) & $\infty$  \\
$300$ & 1.015 (0.003) & 1.040 (0.002) & 0.925 (0.001) & 1.075 (0.002) & $\infty$  \\
$400$ & 1.029 (0.003) & 1.048 (0.002) & 0.961 (0.001) & 1.079 (0.002) & $\infty$  \\
$500$ & 1.011 (0.003) & 1.026 (0.002) & 0.942 (0.001) & 1.044 (0.001) & 2.206 (0.004) \\
$600$ & 1.006 (0.002) & 1.023 (0.002) & 0.943 (0.001) & 1.039 (0.001) & 1.798 (0.002) \\
\bottomrule
\end{tabular}
\end{table}

\begin{table}[ht]
\caption{  Rescaled residual score,  NormalX + sparse + Gaussian noise.}
\centering
    \centering
\begin{tabular}{lccccc}
\multicolumn{6}{c}{\textbf{$n = 100$, varying $|\Lambda|$}}\\
\multicolumn{6}{l}{\textbf{Coverage}} \\
\bottomrule
{$|\Lambda|$} & ModSel-CP & ModSel-CP-LOO & YK-baseline & YK-split & YK-adjust   \\
\hline
$2$ & 0.902 (0.004) & 0.901 (0.004) & 0.890 (0.004) & 0.903 (0.004) & 1.000 (0.000)  \\
$50$ & 0.904 (0.004) & 0.905 (0.004) & 0.858 (0.005) & 0.897 (0.004) & 1.000 (0.000)  \\
$100$ & 0.902 (0.004) & 0.904 (0.004) & 0.845 (0.005) & 0.902 (0.004) & 1.000 (0.000)  \\
$200$ & 0.909 (0.004) & 0.904 (0.004) & 0.841 (0.005) & 0.903 (0.004) & 1.000 (0.000)  \\
$400$ & 0.898 (0.004) & 0.895 (0.004) & 0.833 (0.005) & 0.901 (0.004) & 1.000 (0.000)  \\
$800$ & 0.901 (0.004) & 0.897 (0.004) & 0.823 (0.005) & 0.900 (0.004) & 1.000 (0.000)  \\
$1600$ & 0.905 (0.004) & 0.906 (0.004) & 0.850 (0.005) & 0.904 (0.004) & 1.000 (0.000)  \\
\multicolumn{6}{l}{\textbf{Width ratio}} \\
\bottomrule
{$|\Lambda|$} & ModSel-CP & ModSel-CP-LOO & YK-baseline & YK-split & YK-adjust   \\
\hline
$2$ & 1.021 (0.003) & 1.020 (0.003) & 0.983 (0.003) & 1.040 (0.003) & $\infty$  \\
$50$ & 1.065 (0.003) & 1.074 (0.003) & 0.936 (0.002) & 1.116 (0.003) & $\infty$  \\
$100$ & 1.047 (0.003) & 1.051 (0.003) & 0.892 (0.002) & 1.078 (0.003) & $\infty$  \\
$200$ & 1.061 (0.004) & 1.065 (0.003) & 0.890 (0.002) & 1.098 (0.003) & $\infty$  \\
$400$ & 1.096 (0.004) & 1.098 (0.003) & 0.909 (0.002) & 1.135 (0.003) & $\infty$  \\
$800$ & 1.064 (0.004) & 1.070 (0.003) & 0.865 (0.002) & 1.111 (0.003) & $\infty$  \\
$1600$ & 1.109 (0.004) & 1.119 (0.003) & 0.928 (0.002) & 1.190 (0.004) & $\infty$  \\

\bottomrule
\end{tabular}
\end{table}

\begin{table}[ht]
\centering
    \centering
\begin{tabular}{lccccc}
\multicolumn{6}{c}{\textbf{$|\Lambda| = 200$, varying $n$}}\\
\multicolumn{6}{l}{\textbf{Coverage}} \\
\bottomrule
{$n$} & ModSel-CP & ModSel-CP-LOO & YK-baseline & YK-split & YK-adjust   \\
\hline
$50$ & 0.904 (0.004) & 0.913 (0.004) & 0.802 (0.006) & 0.927 (0.004) & 1.000 (0.000)  \\
$100$ & 0.909 (0.004) & 0.904 (0.004) & 0.841 (0.005) & 0.903 (0.004) & 1.000 (0.000)  \\
$200$ & 0.901 (0.004) & 0.900 (0.004) & 0.857 (0.005) & 0.897 (0.004) & 1.000 (0.000)  \\
$300$ & 0.896 (0.004) & 0.896 (0.004) & 0.862 (0.005) & 0.897 (0.004) & 1.000 (0.000)  \\
$400$ & 0.898 (0.004) & 0.898 (0.004) & 0.873 (0.005) & 0.896 (0.004) & 1.000 (0.000)  \\
$500$ & 0.902 (0.004) & 0.902 (0.004) & 0.898 (0.004) & 0.895 (0.004) & 0.988 (0.002)  \\
$600$ & 0.897 (0.004) & 0.899 (0.004) & 0.887 (0.004) & 0.897 (0.004) & 0.973 (0.002)  \\
\multicolumn{6}{l}{\textbf{Width ratio}} \\
\bottomrule
{$n$} & ModSel-CP & ModSel-CP-LOO & YK-baseline & YK-split & YK-adjust   \\
\hline
$50$ & 1.047 (0.004) & 1.078 (0.003) & 0.803 (0.002) & 1.193 (0.004) & $\infty$  \\
$100$ & 1.061 (0.004) & 1.065 (0.003) & 0.890 (0.002) & 1.098 (0.003) & $\infty$  \\
$200$ & 1.036 (0.003) & 1.033 (0.003) & 0.918 (0.002) & 1.052 (0.003) & $\infty$  \\
$300$ & 1.029 (0.003) & 1.024 (0.003) & 0.936 (0.002) & 1.040 (0.002) & $\infty$  \\
$400$ & 1.023 (0.003) & 1.022 (0.003) & 0.950 (0.002) & 1.037 (0.002) & $\infty$  \\
$500$ & 1.011 (0.002) & 1.009 (0.002) & 0.999 (0.002) & 1.030 (0.002) & 1.651 (0.004)  \\
$600$ & 1.014 (0.002) & 1.015 (0.002) & 0.976 (0.002) & 1.026 (0.002) & 1.430 (0.003)  \\
\bottomrule
\end{tabular}
\end{table}

\begin{table}[ht]
\caption{  Rescaled residual score, NormalX + sparse + heavy-tail noise.}
\centering
\begin{tabular}{lccccc}
\multicolumn{6}{c}{\textbf{$n = 100$, varying $|\Lambda|$}}\\
\multicolumn{6}{l}{\textbf{Coverage}} \\
\bottomrule
{$|\Lambda|$} & ModSel-CP & ModSel-CP-LOO & YK-baseline & YK-split & YK-adjust   \\
\hline
$2$ & 0.906 (0.004) & 0.906 (0.004) & 0.893 (0.004) & 0.900 (0.004) & 1.000 (0.000)  \\
$50$ & 0.899 (0.004) & 0.902 (0.004) & 0.850 (0.005) & 0.903 (0.004) & 1.000 (0.000)  \\
$100$ & 0.903 (0.004) & 0.906 (0.004) & 0.854 (0.005) & 0.910 (0.004) & 1.000 (0.000)  \\
$200$ & 0.911 (0.004) & 0.908 (0.004) & 0.853 (0.005) & 0.901 (0.004) & 1.000 (0.000)  \\
$400$ & 0.907 (0.004) & 0.908 (0.004) & 0.842 (0.005) & 0.908 (0.004) & 1.000 (0.000)  \\
$800$ & 0.904 (0.004) & 0.903 (0.004) & 0.831 (0.005) & 0.905 (0.004) & 1.000 (0.000)  \\
$1600$ & 0.907 (0.004) & 0.908 (0.004) & 0.833 (0.005) & 0.905 (0.004) & 1.000 (0.000)  \\
\multicolumn{6}{l}{\textbf{Width ratio}} \\
\bottomrule
{$|\Lambda|$} & ModSel-CP & ModSel-CP-LOO & YK-baseline & YK-split & YK-adjust   \\
\hline
$2$ & 1.007 (0.003) & 1.008 (0.003) & 0.967 (0.002) & 1.023 (0.003) & $\infty$  \\
$50$ & 1.042 (0.003) & 1.045 (0.003) & 0.894 (0.002) & 1.074 (0.003) & $\infty$ \\
$100$ & 1.069 (0.003) & 1.077 (0.003) & 0.912 (0.002) & 1.115 (0.003) & $\infty$  \\
$200$ & 1.062 (0.004) & 1.069 (0.003) & 0.897 (0.002) & 1.113 (0.003) & $\infty$  \\
$400$ & 1.081 (0.004) & 1.085 (0.003) & 0.890 (0.002) & 1.124 (0.003) & $\infty$  \\
$800$ & 1.078 (0.004) & 1.087 (0.003) & 0.885 (0.002) & 1.132 (0.003) & $\infty$  \\
$1600$ & 1.088 (0.004) & 1.099 (0.003) & 0.883 (0.002) & 1.147 (0.003) & $\infty$  \\

\bottomrule
\end{tabular}
\end{table}

\begin{table}[ht]
\centering
\begin{tabular}{lccccc}
\multicolumn{6}{c}{\textbf{$|\Lambda| = 200$, varying $n$}}\\
\multicolumn{6}{l}{\textbf{Coverage}} \\
\bottomrule
{$n$} & ModSel-CP & ModSel-CP-LOO & YK-baseline & YK-split & YK-adjust   \\
\hline
$50$ & 0.905 (0.004) & 0.911 (0.004) & 0.814 (0.006) & 0.933 (0.004) & 1.000 (0.000)  \\
$100$ & 0.911 (0.004) & 0.908 (0.004) & 0.853 (0.005) & 0.901 (0.004) & 1.000 (0.000)  \\
$200$ & 0.900 (0.004) & 0.902 (0.004) & 0.862 (0.005) & 0.895 (0.004) & 1.000 (0.000)  \\
$300$ & 0.908 (0.004) & 0.905 (0.004) & 0.874 (0.005) & 0.902 (0.004) & 1.000 (0.000)\\
$400$ & 0.898 (0.004) & 0.894 (0.004) & 0.870 (0.005) & 0.895 (0.004) & 1.000 (0.000)  \\
$500$ & 0.900 (0.004) & 0.899 (0.004) & 0.881 (0.005) & 0.902 (0.004) & 0.982 (0.002)  \\
$600$ & 0.901 (0.004) & 0.901 (0.004) & 0.884 (0.005) & 0.899 (0.004) & 0.971 (0.002)  \\
\multicolumn{6}{l}{\textbf{Width ratio}} \\
\bottomrule
{$n$} & ModSel-CP & ModSel-CP-LOO & YK-baseline & YK-split & YK-adjust   \\
\hline
$50$ & 1.057 (0.004) & 1.083 (0.003) & 0.811 (0.002) & 1.207 (0.004) & $\infty$  \\
$100$ & 1.062 (0.004) & 1.069 (0.003) & 0.897 (0.002) & 1.113 (0.003) & $\infty$  \\
$200$ & 1.043 (0.003) & 1.041 (0.003) & 0.925 (0.002) & 1.066 (0.003) & $\infty$  \\
$300$ & 1.041 (0.003) & 1.038 (0.003) & 0.949 (0.002) & 1.059 (0.002) & $\infty$  \\
$400$ & 1.026 (0.003) & 1.022 (0.002) & 0.954 (0.002) & 1.039 (0.002) & $\infty$  \\
$500$ & 1.024 (0.003) & 1.024 (0.002) & 0.962 (0.002) & 1.039 (0.002) & 1.574 (0.003)  \\
$600$ & 1.021 (0.002) & 1.016 (0.002) & 0.972 (0.002) & 1.031 (0.002) & 1.444 (0.003)  \\
\bottomrule
\end{tabular}
\end{table}

\begin{table}[ht]
\caption{ Rescaled residual score, NormalX + dense + Gaussian noise.}
\centering
\begin{tabular}{lccccc}
\multicolumn{6}{c}{\textbf{$n = 100$, varying $|\Lambda|$}}\\
\multicolumn{6}{l}{\textbf{Coverage}} \\
\bottomrule
{$|\Lambda|$} & ModSel-CP & ModSel-CP-LOO & YK-baseline & YK-split & YK-adjust   \\
\hline
$2$ & 0.907 (0.004) & 0.905 (0.004) & 0.893 (0.004) & 0.903 (0.004) & 1.000 (0.000)  \\
$50$ & 0.900 (0.004) & 0.905 (0.004) & 0.847 (0.005) & 0.903 (0.004) & 1.000 (0.000)  \\
$100$ & 0.903 (0.004) & 0.906 (0.004) & 0.835 (0.005) & 0.907 (0.004) & 1.000 (0.000) \\
$200$ & 0.900 (0.004) & 0.906 (0.004) & 0.828 (0.005) & 0.907 (0.004) & 1.000 (0.000) \\
$400$ & 0.904 (0.004) & 0.901 (0.004) & 0.825 (0.005) & 0.902 (0.004) & 1.000 (0.000)  \\
$800$ & 0.902 (0.004) & 0.905 (0.004) & 0.814 (0.006) & 0.904 (0.004) & 1.000 (0.000)  \\
$1600$ & 0.901 (0.004) & 0.906 (0.004) & 0.804 (0.006) & 0.900 (0.004) & 1.000 (0.000)  \\
\multicolumn{6}{l}{\textbf{Width ratio}} \\
\bottomrule
{$|\Lambda|$} & ModSel-CP & ModSel-CP-LOO & YK-baseline & YK-split & YK-adjust   \\
\hline
$2$ & 1.018 (0.003) & 1.019 (0.002) & 0.980 (0.002) & 1.036 (0.003) & $\infty$  \\
$50$ & 1.027 (0.003) & 1.037 (0.003) & 0.875 (0.002) & 1.055 (0.003) & $\infty$  \\
$100$ & 1.037 (0.003) & 1.047 (0.003) & 0.862 (0.002) & 1.065 (0.003) & $\infty$  \\
$200$ & 1.028 (0.003) & 1.047 (0.003) & 0.848 (0.002) & 1.058 (0.003) & $\infty$  \\
$400$ & 1.040 (0.003) & 1.056 (0.003) & 0.841 (0.002) & 1.078 (0.003) & $\infty$ \\
$800$ & 1.041 (0.003) & 1.053 (0.003) & 0.826 (0.002) & 1.072 (0.003) & $\infty$  \\
$1600$ & 1.054 (0.003) & 1.066 (0.003) & 0.822 (0.002) & 1.090 (0.003) & $\infty$  \\

\bottomrule
\end{tabular}
\end{table}

\begin{table}[ht]
\centering
\begin{tabular}{lccccc}
\multicolumn{6}{c}{\textbf{$|\Lambda| = 200$, varying $n$}}\\
\multicolumn{6}{l}{\textbf{Coverage}} \\
\bottomrule
{$n$} & ModSel-CP & ModSel-CP-LOO & YK-baseline & YK-split & YK-adjust   \\
\hline
$50$ & 0.898 (0.004) & 0.904 (0.004) & 0.791 (0.006) & 0.920 (0.004) & 1.000 (0.000)  \\
$100$ & 0.900 (0.004) & 0.906 (0.004) & 0.828 (0.005) & 0.907 (0.004) & 1.000 (0.000) \\
$200$ & 0.899 (0.004) & 0.902 (0.004) & 0.856 (0.005) & 0.902 (0.004) & 1.000 (0.000) \\
$300$ & 0.901 (0.004) & 0.904 (0.004) & 0.863 (0.005) & 0.903 (0.004) & 1.000 (0.000) \\
$400$ & 0.903 (0.004) & 0.904 (0.004) & 0.875 (0.005) & 0.908 (0.004) & 1.000 (0.000) \\
$500$ & 0.899 (0.004) & 0.901 (0.004) & 0.876 (0.005) & 0.898 (0.004) & 0.983 (0.002)  \\
$600$ & 0.901 (0.004) & 0.902 (0.004) & 0.881 (0.005) & 0.901 (0.004) & 0.970 (0.002) \\
\multicolumn{6}{l}{\textbf{Width ratio}} \\
\bottomrule
{$n$} & ModSel-CP & ModSel-CP-LOO & YK-baseline & YK-split & YK-adjust   \\
\hline
$50$ & 1.026 (0.004) & 1.049 (0.003) & 0.774 (0.002) & 1.165 (0.004) & $\infty$  \\
$100$ & 1.028 (0.003) & 1.047 (0.003) & 0.848 (0.002) & 1.058 (0.003) & $\infty$  \\
$200$ & 1.038 (0.003) & 1.044 (0.003) & 0.908 (0.002) & 1.057 (0.003) & $\infty$  \\
$300$ & 1.030 (0.003) & 1.034 (0.002) & 0.926 (0.002) & 1.048 (0.002) & $\infty$  \\
$400$ & 1.027 (0.003) & 1.030 (0.002) & 0.944 (0.002) & 1.046 (0.002) & $\infty$ \\
$500$ & 1.024 (0.003) & 1.027 (0.002) & 0.953 (0.002) & 1.037 (0.002) & 1.508 (0.003)  \\
$600$ & 1.021 (0.003) & 1.022 (0.002) & 0.956 (0.002) & 1.035 (0.002) & 1.391 (0.003) \\
\bottomrule
\end{tabular}
\end{table}
 
\begin{table}[ht]
\caption{ Rescaled residual score,  tX + sparse + Gaussian noise.}
\centering
\begin{tabular}{lccccc}
\multicolumn{6}{c}{\textbf{$n = 100$, varying $|\Lambda|$}}\\
\multicolumn{6}{l}{\textbf{Coverage}} \\
\bottomrule
{$|\Lambda|$} & ModSel-CP & ModSel-CP-LOO & YK-baseline & YK-split & YK-adjust   \\
\hline
$2$ & 0.894 (0.004) & 0.892 (0.004) & 0.886 (0.004) & 0.893 (0.004) & 1.000 (0.000) \\
$50$ & 0.896 (0.004) & 0.897 (0.004) & 0.849 (0.005) & 0.898 (0.004) & 1.000 (0.000)  \\
$100$ & 0.887 (0.004) & 0.887 (0.004) & 0.827 (0.005) & 0.898 (0.004) & 1.000 (0.000) \\
$200$ & 0.897 (0.004) & 0.899 (0.004) & 0.832 (0.005) & 0.903 (0.004) & 1.000 (0.000)  \\
$400$ & 0.893 (0.004) & 0.895 (0.004) & 0.818 (0.005) & 0.892 (0.004) & 1.000 (0.000)  \\
$800$ & 0.893 (0.004) & 0.894 (0.004) & 0.816 (0.005) & 0.900 (0.004) & 1.000 (0.000)\\
$1600$ & 0.900 (0.004) & 0.899 (0.004) & 0.828 (0.005) & 0.897 (0.004) & 1.000 (0.000)  \\
\multicolumn{6}{l}{\textbf{Width ratio}} \\
\bottomrule
{$|\Lambda|$} & ModSel-CP & ModSel-CP-LOO & YK-baseline & YK-split & YK-adjust   \\
\hline
$2$ & 1.019 (0.006) & 1.020 (0.006) & 0.995 (0.006) & 1.045 (0.006) & $\infty$ \\
$50$ & 1.088 (0.008) & 1.091 (0.007) & 0.937 (0.006) & 1.133 (0.008) & $\infty$  \\
$100$ & 1.075 (0.009) & 1.060 (0.007) & 0.893 (0.006) & 1.120 (0.007) & $\infty$  \\
$200$ & 1.061 (0.008) & 1.058 (0.006) & 0.865 (0.005) & 1.107 (0.007) & $\infty$ \\
$400$ & 1.065 (0.008) & 1.065 (0.007) & 0.856 (0.005) & 1.113 (0.007) & $\infty$  \\
$800$ & 1.087 (0.008) & 1.084 (0.007) & 0.869 (0.005) & 1.145 (0.007) & $\infty$ \\
$1600$ & 1.129 (0.008) & 1.122 (0.007) & 0.913 (0.005) & 1.215 (0.008) & $\infty$  \\
\bottomrule
\end{tabular}
\end{table}

\begin{table}[ht]
\centering
\begin{tabular}{lccccc}
\multicolumn{6}{c}{\textbf{$|\Lambda| = 200$, varying $n$}}\\
\multicolumn{6}{l}{\textbf{Coverage}} \\
\bottomrule
{$n$} & ModSel-CP & ModSel-CP-LOO & YK-baseline & YK-split & YK-adjust   \\
\hline
$50$ & 0.905 (0.004) & 0.904 (0.004) & 0.783 (0.006) & 0.920 (0.004) & 1.000 (0.000)  \\
$100$ & 0.897 (0.004) & 0.899 (0.004) & 0.832 (0.005) & 0.903 (0.004) & 1.000 (0.000)  \\
$200$ & 0.905 (0.004) & 0.905 (0.004) & 0.864 (0.005) & 0.900 (0.004) & 1.000 (0.000)  \\
$300$ & 0.911 (0.004) & 0.908 (0.004) & 0.881 (0.005) & 0.905 (0.004) & 1.000 (0.000)  \\
$400$ & 0.902 (0.004) & 0.904 (0.004) & 0.881 (0.005) & 0.899 (0.004) & 1.000 (0.000) \\
$500$ & 0.895 (0.004) & 0.894 (0.004) & 0.891 (0.004) & 0.895 (0.004) & 0.986 (0.002) \\
$600$ & 0.895 (0.004) & 0.894 (0.004) & 0.887 (0.004) & 0.894 (0.004) & 0.975 (0.002)  \\
\multicolumn{6}{l}{\textbf{Width ratio}} \\
\bottomrule
{$n$} & ModSel-CP & ModSel-CP-LOO & YK-baseline & YK-split & YK-adjust   \\
\hline
$50$ & 1.100 (0.009) & 1.100 (0.007) & 0.797 (0.005) & 1.261 (0.009) & $\infty$  \\
$100$ & 1.061 (0.008) & 1.058 (0.006) & 0.865 (0.005) & 1.107 (0.007) & $\infty$  \\
$200$ & 1.044 (0.007) & 1.034 (0.006) & 0.918 (0.005) & 1.073 (0.007) & $\infty$  \\
$300$ & 1.034 (0.007) & 1.027 (0.006) & 0.943 (0.006) & 1.063 (0.007) & $\infty$  \\
$400$ & 1.037 (0.007) & 1.031 (0.006) & 0.965 (0.006) & 1.059 (0.006) & $\infty$  \\
$500$ & 1.008 (0.006) & 1.007 (0.006) & 0.998 (0.005) & 1.036 (0.006) & 1.778 (0.010)  \\
$600$ & 1.021 (0.008) & 1.018 (0.007) & 0.991 (0.007) & 1.043 (0.007) & 1.497 (0.010)  \\
\bottomrule
\end{tabular}
\end{table}

\clearpage

\end{document}